\newif\ifshort
\newif\iflong
\longtrue

\ifshort
\documentclass[acmsmall,screen]{acmart}
\else
\documentclass[acmsmall,screen,nonacm]{acmart} 
\fi

\usepackage{amsthm}
\usepackage{cleveref}
\usepackage{comment}
\usepackage{mathtools}
\usepackage{amsfonts}
\usepackage{gensymb}
\usepackage{mathpartir}
\usepackage{natbib}
\usepackage{tikz-cd}
\usepackage{scalerel}
\usepackage{mttex}

\LetLtxMacro{\oldmathpar}{\mathpar}
\LetLtxMacro{\oldendmathpar}{\endmathpar}
\LetLtxMacro{\oldenumerate}{\enumerate}
\LetLtxMacro{\oldendenumerate}{\endenumerate}

\LetLtxMacro{\olddisplaymath}{\displaymath}
\LetLtxMacro{\oldenddisplaymath}{\enddisplaymath}

\LetLtxMacro{\oldalign}{\align}
\LetLtxMacro{\oldendalign}{\endalign}

\LetLtxMacro{\oldfigure}{\figure}
\LetLtxMacro{\oldendfigure}{\endfigure}

\LetLtxMacro{\oldcaption}{\caption}

\renewenvironment{mathpar}{
\begin{nop}\small\oldmathpar}{
\oldendmathpar\end{nop}\ignorespacesafterend}

\renewenvironment{displaymath}{
\begin{nop}\small\olddisplaymath}{
\oldenddisplaymath\end{nop}\ignorespacesafterend}

\renewenvironment{align}{
\small\oldalign}{
\oldendalign\ignorespacesafterend}

\renewenvironment{figure}
{\oldfigure}
{\vspace{-2ex}\oldendfigure}

\renewcommand{\caption}[1]{\vspace{-\baselineskip}\oldcaption{#1}}

\LetLtxMacro{\oldtfonttext}{\tfonttext}
\LetLtxMacro{\oldsfonttext}{\sfonttext}
\LetLtxMacro{\oldbefonttext}{\befonttext}

\renewcommand{\tfonttext}[1]{\oldtfonttext{\small #1}}
\renewcommand{\sfonttext}[1]{\oldsfonttext{\small #1}}

\newcommand{\obseqv}{\approx^{\text{obs}}}
\newcommand{\obsapprox}{\sqsubseteq^{\text{obs}}}

\newcommand{\by}[2]{(\text{#1, \ref{#2}})}

\newcommand{\tlangname}{\lambda_{T,\mho}}
\newcommand{\glangname}{\lambda_{G}}

\newcommand{\dynr}{\sqsubseteq}

\newcommand{\ltdyn}{\sqsubseteq}
\newcommand{\obsltdyn}{\mathrel{\ltdyn}^{\text{obs}}}
\newcommand{\logltdyn}{\mathrel{\ltdyn}}
\newcommand{\ptmapprox}[6]{{{#3}\mathrel{#1}#4\mathrel{#2}#5:#6}}
\newcommand{\tmlogapprox}[4]{\ptmapprox{\vDash}{\logltdyn}{#1}{#2}{#3}{#4}}
\newcommand{\ptmdyn}[8]{\ptmapprox{#1}{#2}{{#3}\ltdyn{#4}}{#5}{#6}{{#7}\ltdyn{#8}}}
\newcommand{\tmdynr}[6]{\ptmdyn{\vdash}{\ltdyn}{#1}{#2}{#3}{#4}{#5}{#6}}
\newcommand{\semtmdynr}[6]{\ptmdyn{\vDash}{\ltdyn}{#1}{#2}{#3}{#4}{#5}{#6}}
\newcommand{\obstmdynr}[6]{\ptmdyn{\vDash}{\obsltdyn}{#1}{#2}{#3}{#4}{#5}{#6}}
\newcommand{\logtmdynr}[6]{\ptmdyn{\vDash}{\logltdyn}{#1}{#2}{#3}{#4}{#5}{#6}}

\newcommand{\semhomoltdyn}[4]{#1\mathrel{\vDash}#2\ltdyn#3:#4}

\newcommand{\isim}[1]{\mathrel{\sim^{#1}}}
\newcommand{\iprec}[1]{\mathrel{\prec^{#1}}}

\newcommand{\pletemb}[2]{#1{\sem{\senvtwo}}
                {\mE_{e,\senvone,\senvtwo}\hw{\sem{\senvone}}}
                {#2}}
\newcommand{\letemb}[1]{\pletemb{\mtletvert}{#1}}
\newcommand{\letembnovert}[1]{\pletemb{\mtlet}{#1}}
\newcommand{\pletprj}[2]{#1{\sem{\senvone}}
                {\mE_{p,\senvone,\senvtwo}\hw{\sem{\senvtwo}}}
                {#2}}
\newcommand{\letprj}[1]{\pletprj{\mtletvert}{#1}}
\newcommand{\letprjnovert}[1]{\pletprj{\mtlet}{#1}}
\newcommand{\theemb}[2]{\mE_{e,#1,#2}\hw}
\newcommand{\theprj}[2]{\mE_{p,#1,#2}\hw}

\newcommand{\gtdyn}{\sqsupseteq}
\newcommand{\equidyn}{\mathrel{\gtdyn\ltdyn}}
\newcommand{\ltdynlt}{\mathrel{\ltdyn\prec}}
\newcommand{\ltdyngt}{\mathrel{\ltdyn\succ}}
\newcommand{\ltdynsim}{\mathrel{\ltdyn\sim}}

\newcommand{\paramlrdynr}[4]{\mathrel{{#1}^{#3}_{#2,#4}}}
\newcommand{\ltdynr}{\mathrel{\dynr\prec}}
\newcommand{\gtdynr}{\mathrel{\dynr\succ}}
\newcommand{\simdynr}{\mathrel{\dynr\sim}}
\newcommand{\tltdynr}{\paramlrdynr{\ltdynr}{t}}
\newcommand{\vltdynr}{\paramlrdynr{\ltdynr}{v}}
\newcommand{\tgtdynr}{\paramlrdynr{\gtdynr}{t}}
\newcommand{\vgtdynr}{\paramlrdynr{\gtdynr}{v}}
\newcommand{\vsimdynr}{\paramlrdynr{\simdynr}{v}}
\newcommand{\tsimdynr}{\paramlrdynr{\simdynr}{t}}

\newcommand{\dotltdynr}{\paramlrdynr{\ltdynr}{\cdot}{\cdot}{\cdot}}
\newcommand{\dotgtdynr}{\paramlrdynr{\gtdynr}{\cdot}{\cdot}{\cdot}}

\newcommand{\stepsin}[1]{\mathrel{\mapsto^{#1}}}

\newcommand{\obcast}[2]{\sfontsym{\langle{#2}}\sbin{\Leftarrow}{#1}\sfontsym{\rangle}}
\newcommand{\sbin}[1]{\mathbin{\sfontsym{#1}}}
\newcommand{\serr}{\sfontsym{\mho}}
\newcommand{\parened}[1]{({#1})}

\newcommand{\errord}{\mathrel{\sqsubseteq_{\text{err}}}}

\newcommand{\errordof}[1]{\mathrel{\sqsubseteq_{#1}}}

\newcommand{\sdynty}{\sfont{\dyn}}

\newcommand{\sjudgtc}[3]{{#1}\vdash{}{#2}:{#3}}

\newcommand{\sjudgtyprec}[2]{{#1}\preciser{}{#2}}

\newcommand{\sempenv}{\sfont{\cdot}}
\newcommand{\senvext}[3]{{#1},{#2}:{#3}}
\newcommand{\senv}{\sfont{\Gamma}}

\newcommand{\mvinpr}[2][1]{\mvmetavar{#2}{}{#1}}

\newcommand{\sparened}[1]{\sfontsym{(}{#1}\sfontsym{)}}

\newcommand{\snatty}{\sfont{\mathbb{N}}}

\newcommand{\srecordty}[1]{\sfontsym{\{}{#1}\sfontsym{\}}}

\newcommand{\shw}[1]{{\sfontsym{\lbrack{}}{#1}\sfontsym{\rbrack{}}}}
\newcommand{\shole}{\shw{\sfontsym{\cdot}}}

\newcommand{\mkvar}{x}
\newcommand{\mkvaralt}{y}
\newcommand{\mkvaraltalt}{z}
\newcommand{\mkterm}{t}
\newcommand{\mktermalt}{s}
\newcommand{\mkty}{A}
\newcommand{\mktyalt}{B}
\newcommand{\mktyaltalt}{C}
\newcommand{\mkcon}{c}
\newcommand{\mkval}{v}
\newcommand{\mktpder}{d}
\newlanguage{\scolor}{\sfont}{\sfontsym}{s}
{}
{val/\mkval,v/\mkval,var/\mkvar,varalt/\mkvaralt,varaltalt/\mkvaraltalt,con/\mkcon,term/\mkterm,termalt/\mktermalt,tagty/G,ty/\mkty,A/A,B/B,C/C,tyalt/\mktyalt,tyaltalt/\mktyaltalt,env/\Gamma,ctxt/C,fld/\eta,fldset/H,alpha/\alpha,tpder/\mktpder,row/\rho,eta/\eta,t/t,ectxt/E,E/E,x/x,s/s,y/y,gamma/\gamma}
{fun,pair,unit,void,sum,mu}
{}
\newcommand{\sto}{\mathbin{\sfontsym{\to}}}
\newcommand{\mto}{\mathbin{\mfontsym{\to}}}
\newcommand{\mplus}{\mathbin{\mfontsym{+}}}
\newcommand{\splus}{\mathbin{\sfontsym{+}}}
\newcommand{\mtimes}{\mathbin{\mfontsym{\times}}}
\newcommand{\stimes}{\mathbin{\sfontsym{\times}}}
\newexprs{\sfontsym}{\sfont}{st}{}
{roll,unroll,ufix,voidcase,ufun,fun,app,let,case,unit,pair,matchpair,matchpairvert,inj,injpr,letvert}

\newcommand{\mcolor}[1]{\tcolor{#1}}
\newcommand{\mfont}[1]{\tfont{#1}}
\newcommand{\mfontsym}[1]{\tfontsym{#1}}

\newcommand{\mkmty}{A}
\newcommand{\mkmtyalt}{B}
\newcommand{\mkrel}{R}
\newcommand{\mkrelalt}{Q}
\newcommand{\mktag}{\epsilon}
\newlanguage{\mcolor}{\mfont}{\mfontsym}{m}
{}
{ix/i,ixalt/j,ixaltalt/k,set/I,setalt/J,setaltalt/K,val/\mkval, v/\mkval,uu/u,var/\mkvar,varalt/\mkvaralt,tag/\mktag,t/t,term/\mkterm,termalt/\mktermalt,ty/\mkmty,A/A,B/B,C/C,tyalt/\mkmtyalt,tyaltalt/\mktyaltalt,renv/\Theta,env/\Gamma,rel/\mkrel,relalt/\mkrelalt,ectxt/E,ctxt/C,tyvar/\alpha,alpha/\alpha,row/\rho,x/x,E/E,y/y,s/s,gamma/\gamma,eta/\eta}
{fun,pair,unit,void,sum,mu}
{}

\newexprs{\mfontsym}{\mfont}{mt}{}
{roll,unroll,elsecase,elsecasevert,ufix,voidcase,ufun,fun,app,let,letvert,sum,case,unit,pair,pairvert,matchpair,matchpairvert,inj,injpr,casevert}

\newcommand{\mhole}{{\mfont{\lbrack{}{\cdot}\rbrack{}}}}

\newcommand{\inholestep}[2]{\inholestepsin{0}{#1}{#2}}
\newcommand{\inholestepsin}[3]{\mectxt[{#2}] \stepsin{#1} \mectxt[{#3}]}
\newcommand{\sinholestep}[2]{\sectxt\hw{#1} \step \sectxt\hw{#2}}
\newcommand{\bigstep}[1]{\mathrel{\Mapsto^{#1}}}
\newcommand{\bigstepany}{\mathrel{\Mapsto}}

\newcommand{\merr}{\mfontsym{\mho}}

\newcommand{\mempenv}{\mfontsym{\cdot}}
\newcommand{\menvext}[3]{{#1},{#2}:{#3}}

\newcommand{\mjudgtc}[3]{\judg{#1}{#2}{#3}}

\newcommand{\eppair}{\mathrel{\triangleleft}}

\newcommand{\id}{\text{id}}

\newcommand{\preciser}{\sqsubseteq}

\newcommand{\dyn}{\mathord{?}}

\newcommand{\sem}{\semantics}
\newcommand{\semantics}[1]{\llbracket{}{#1}\rrbracket{}}

\newcommand{\floor}[1]{\sfontsym{\lfloor}{#1}\sfontsym{\rfloor}}

\makeatletter
\def\slashedarrowfill@#1#2#3#4#5{%
  $\m@th\thickmuskip0mu\medmuskip\thickmuskip\thinmuskip\thickmuskip
\relax#5#1\mkern-7mu%
\cleaders\hbox{$#5\mkern-2mu#2\mkern-2mu$}\hfill
\mathclap{#3}\mathclap{#2}%
\cleaders\hbox{$#5\mkern-2mu#2\mkern-2mu$}\hfill
\mkern-7mu#4$%
}
\def\rightslashedarrowfill@{%
  \slashedarrowfill@\relbar\relbar\mapstochar\rightarrow}
\newcommand\xslashedrightarrow[2][]{%
  \ext@arrow 0055{\rightslashedarrowfill@}{#1}{#2}}
\makeatother

\newcommand{\defstartnegspace}{\vspace{-0.5ex}}
\newcommand{\defendnegspace}{\vspace{0ex}}

\newcommand{\etatagcases}[2]{
\begin{stackTL}
\kwopen{\mfont{case}} #1 \kwbin{\mfont{of}}\\
\quad \mtsum{\sunitty}{\mxin{\sunitty}} \mathpunct{\mfontsym{.}} #2{\mtsum{\sunitty}{\mxin{\sunitty}}}\\
\quad \mtsum{\splus}{\mxin{\splus}} \mathpunct{\mfontsym{.}} #2\mtsum{\splus}{\mxin{\splus}}\\
\quad \mtsum{\stimes}{\mxin{\stimes}} \mathpunct{\mfontsym{.}} #2\mtsum{\stimes}{\mxin{\stimes}}\\
\quad \mtsum{\sto}{\mxin{\sto}} \mathpunct{\mfontsym{.}} #2\mtsum{\sto}{\mxin{\sto}}\\
\quad \mtsum{\srecordty{\sdynty}}{\mxin{\srecordty{\sdynty}}} \mathpunct{\mfontsym{.}} #2\mtsum{\srecordty{\sdynty}}{\mxin{\srecordty{\sdynty}}}\\
\end{stackTL}
}

\setcitestyle{nosort}

\newcommand{\colornote}{{In this paper, we use $\sfont{\textsf{blue}}$ to
    typeset our gradual cast calculus $\glangname$ and
    $\mfont{\textbf{red}}$ to typeset our typed language with errors
    $\tlangname$. The paper will be much easier to read if
    viewed/printed in color.}}

\ifshort
\setcopyright{rightsretained}
\acmPrice{}
\acmDOI{10.1145/3234596}
\acmYear{2018}
\copyrightyear{2018}
\acmJournal{PACMPL}
\acmVolume{2}
\acmNumber{ICFP}
\acmArticle{7}
\acmMonth{9}
\fi
\begin{document}

\title{Graduality from Embedding-Projection Pairs}\titlenote{\colornote}
\iflong
\subtitle{Extended Version}
\fi

\author{Max S. New}
\affiliation{
  \institution{Northeastern University}            
}
\email{maxnew@ccs.neu.edu}          
\author{Amal Ahmed}
\affiliation{
  \institution{Northeastern University and Inria Paris}            
}
\email{amal@ccs.neu.edu}          

\begin{abstract}
  Gradually typed languages allow statically typed and dynamically
  typed code to interact while maintaining benefits of both styles.
  The key to reasoning about these mixed programs is
  Siek-Vitousek-Cimini-Boyland's 
  (dynamic) gradual guarantee, which
  says that giving components of a program more precise types only
  adds runtime type checking, and does not otherwise change behavior.
  In this paper, we give a semantic reformulation of the gradual
  guarantee called \emph{graduality}.
  We change the name to promote the analogy that graduality is to
  gradual typing what parametricity is to polymorphism.
  Each gives a local-to-global, syntactic-to-semantic reasoning
  principle that is formulated in terms of a kind of observational
  approximation.

  Utilizing the analogy, we develop a novel logical relation for
  proving graduality.
  We show that \emph{embedding-projection pairs (ep pairs)} are
  to graduality what relations are to parametricity.
  We argue that casts between two types where one is ``more dynamic''
  (less precise) than the other necessarily form an
  ep pair, and we use this to cleanly prove the
  graduality cases for casts from the ep-pair property.
  To construct ep pairs, we give an analysis of the \emph{type dynamism}
  relation---also known as type precision or na\"ive subtyping---that
  interprets the rules for type dynamism as compositional constructions on
  ep pairs, analogous to the coercion interpretation of subtyping. 
\end{abstract}



\maketitle

\section{Introduction}
\label{sec:intro}

Gradually typed programming languages are designed to resolve the
conflict between static and dynamically typed programming styles
\citep{tobin-hochstadt06,tobin-hochstadt08,siek-taha06}.
A gradual language allows a smooth transition from dynamic to static
typing through gradual addition of types to dynamically typed
programs, and allows for safe interactions between more statically
typed and more dynamically typed components.
With such an enticing goal there has
been extensive research on gradual typing---e.g.,
\cite{siek-taha06,gronski06,wadler-findler09,Ina:2011zr,swamy14,Allende:2013aa}---with
recent work aimed at extending gradual typing to more advanced
language features, such as parametric
polymorphism~\cite{ahmed17,igarashipoly17}, effect
tracking~\cite{gradeffects2014}, typestate~\cite{Wolff:2011:GT},
session types~\cite{igarashisession17}, and refinement
types \linebreak \cite{lehmann17}.

Formalizing the idea of a ``smooth transition'', a key property that
every gradually typed language should satisfy is
\citeauthor*{refined}'s \emph{(dynamic)\footnote{The same work also
    introduces a \emph{static} gradual guarantee that says that
    changing the types in a program to be less dynamic means type
    checking becomes stricter. We do not consider this in our paper
    because we only consider the semantics of cast calculi, not the
    type systems of gradual surface languages. We discuss the
    relationship further in \secref{sec:rel:gradual-guarantee}}
  gradual guarantee}, which we refer to as \emph{graduality} (by
analogy with parametricity).
Graduality enables programmers to modify their program from a
dynamically typed to a statically typed style, and vice-versa, with
confidence that the program's behavior only changes in predictable
ways. 
Specifically, it says that changing the types in a
program to be ``less dynamic''/''more precise''---i.e., changing from
the dynamic type to some more precise type such as integers or
functions---either produces the same behavior as the original program
or causes a dynamic type error.
Conversely, if a program does not error and some types are made ``more
dynamic''/''less precise'' then the program has the exact same
behavior.  This is an important reasoning principle for programmers as
the alternative would be quite counterintuitive: for instance,
changing certain type annotations might cause a terminating program to
diverge, or a program that prints your calendar to tweet your home
address!
This distinguishes dynamic type checking in gradual typing from
exceptions: raising an exception is a valid program behavior that can
be caught and handled by a caller, whereas a dynamic type error is
always considered to be a bug, and terminates the program.

More formally, the notion of when a type $\sA$ is ``less dynamic''
than another type $\sB$ is specified by a \emph{type dynamism}
relation (also known as type precision or na\"ive subtyping), written
$\sA \ltdyn \sB$, which is defined for simple languages as the least
congruence relation such that the dynamic type $\sdynty$ is the
\emph{most} dynamic type: $\sA \ltdyn \sdynty$.
Then, \emph{term dynamism} (also known as term
precision) is the natural extension of type dynamism to terms, written
$\st \ltdyn \ss$.
The graduality theorem is then that if $\st \ltdyn \ss$, then the
behavior of $\st$ must be ``less dynamic'' than the behavior of
$\ss$---that is, either $\st$ produces a runtime type error or both
terms have the exact same behavior.
We say $\st$ is ``less dynamic'' in the sense that it has \emph{fewer
  behaviors}.

Unfortunately, for the majority of gradually typed languages, the
(dynamic) gradual guarantee is considered quite challenging to prove,
and there is only limited guidance about how to design new languages
so that they satisfy this property.  There are two notable exceptions:
Abstracting Gradual Typing (AGT) \cite{garcia16} and the Gradualizer
\cite{gradualizer16, gradualizer17} provide systematic methods and
formal tools, respectively, for deriving a gradually typed language
from a statically typed language, and they both provide the gradual
guarantee by construction.  However, while they provide a proof of the
gradual guarantee for languages produced in the respective frameworks,
most gradually typed languages are not produced in this way; for
instance, Typed Racket’s approach to gradual typing
\cite{tobin-hochstadt06,tobin-hochstadt08} is not explained by either system.
Furthermore, both Gradualizer and AGT base their semantics on static
type checking itself, but this is the reverse of the semantic view of
type checking.
In the semantic viewpoint, type checking should be justified by a
sensible semantics, and not the other way around.

\paragraph{Type Dynamism and Embedding-Projection Pairs}

While the gradual guarantee as presented in \citet{refined} makes type
dynamism a central component, the semantic meaning of type dynamism is
unclear.
This is not just a philosophical question: it is unclear how to extend
type dynamism to new language features.
For instance, polymorphic gradually typed languages have been
developed recently by \citet{ahmed17} and \citet{igarashipoly17},
but the two papers have different definitions of type dynamism, 
and neither attempts a proof of the (dynamic) gradual guarantee.
The AGT \cite{garcia16} approach gives a systematic definition of
type dynamism in terms of sets of static types, but that definition is
difficult to separate from the rest of their framework, whereas we
would like a definition that can be interpreted in any gradually typed
language.
At present, the best guidance we have comes from the gradual guarantee
itself: the dynamic type should be the greatest element, and the
gradual guarantee should hold.

We propose a semantic definition for type dynamism that naturally
leads to a clean formulation and proof of the gradual guarantee:
An ordering $\sA \ltdyn \sB$ should hold when the casts between the
two types form an \emph{embedding-projection pair}.

What does this mean?
First, in order to support interaction between statically typed and
dynamic typed code while still maintaining the guarantees of the
static types, gradually typed languages include
\emph{casts}\footnote{It is not literally true that every gradual
  language uses this presentation of casts from cast calculi, but in
  order for a language to be gradually typed, some means of casting
  between types must be available, such as embedding dynamic code in
  statically typed code, or type annotations. We argue that the
  properties of casts we identify here should apply to those
  presentations as well.} $\obcast\sA\sB$ that dynamically check if a
value of type $\sA$ corresponds to a valid inhabitant of the type
$\sB$, and if so, transform its value to have the right type.
Then if $\sA \ltdyn \sB$, we say that the casts $\obcast\sA\sB$ and
$\obcast\sB\sA$ form an \emph{embedding-projection pair}, which means
that they satisfy the following two properties that describe acceptable
behaviors when casting between the two types: \emph{retraction} and
\emph{projection}. 

First, $\sA$ should be a \emph{stricter} type than $\sB$, so anything
satisfying $\sA$ should also  satisfy $\sB$.
This is captured in the \emph{retraction} property: if we cast a value
$\sv : \sA$ from $\sA$ to $\sB$ and then back down to $\sA$, we should
get back an equivalent value because $\sv$ should satisfy the type of
$\sA$ and $\sB$.
Formally, $\obcast\sB\sA\obcast\sA\sB\st \approx \st$ where $\approx$
means \emph{observational equivalence} of the programs: when placed in
the same spot in a program, they produce the same behavior.

Second, casts should only be doing type \emph{checking}, and not
otherwise changing the behavior of the term.
Since $\sB$ is a weaker property than $\sA$, if we cast a value of
$\sv : \sB$ down to $\sA$, there may be a runtime type error.
However, if $\sv$ really does satisfy $\sA$ the cast succeeds, and if
we cast back to $\sB$ we should get back a value with similar behavior
to $\sv$.
If $\sB$ is a first-order type like booleans or numbers, we should get
back exactly the same value.
However, if $\sA, \sB$ are higher-order types like functions or
objects, then it is impossible to check if a value of one type $\sB$
satisfies $\sA$.
For instance, if $\sB = \sdynty \sto \sdynty$ and $\sA = \sdynty \sto
\snatty$, then it is not decidable whether or not a value of $\sB$
will always return a number on every input.
Instead, following \cite{findler-felleisen02}, gradual type casts
\emph{wrap} the function with a cast on its outputs and if at any
point it returns something that is not a number, a type error is
raised.
So if $\sv : \sB$ is cast to $\sA$ and back, we cannot expect to
always get an equivalent value back, but the result should \emph{error
  more}---that is, either the cast to $\sA$ raises an error, or we get
back a new value $\svpr : \sB$ that has the same behavior as $\sv$
except it sometimes raises a type error.
We formalize this as \emph{observational error approximation} and write the
ordering $\st \sqsubseteq \stpr$ as ``$\st$ errors more than
$\stpr$''.
We then use this to formalize the \emph{projection} property:
$\obcast\sA\sB\obcast\sB\sA\ss \sqsubseteq \ss$.
%

Notice how the justification for the projection property uses the same
intuition as graduality: that casts should only be doing
\emph{checking} and not completely changing a program's behavior.
This is the key to why embedding-projection pairs help to formulate
and prove graduality: we view graduality as the natural extension of
the projection property from a property of casts to a property of
arbitrary gradually typed programs.

This gives us nice properties of some casts, but what do we know about
casts that are \emph{not} upcasts or downcasts?
In traditional formulations, gradual typing includes casts
$\obcast\sA\sB$ between types that are \emph{shallowly
  compatible}---i.e, that are not guaranteed to fail. For instance, we
can cast a pair where the left side is known to be a number $\snatty
\stimes \sdynty$ to a type where the right side is known to be a number $\sdynty
\stimes \snatty$ with casts succeeding on values where both sides are
numbers.
The resulting cast
$\obcast{\snatty\stimes\sdynty}{\sdynty\stimes\snatty}$ is neither an
upcast nor a downcast.
We argue that the formulation based on these ``general'' casts is ill
behaved from a meta-theoretic perspective: you are quite limited in
your ability to break casts for larger types into casts for smaller
types.
Most notably, the \emph{composition} of two general casts is very
rarely the same as the direct cast.
For instance, casting from $\snatty$ to $\sdynty\sto\sdynty$ and back
to $\snatty$ always errors, but obviously the direct cast $\obcast\snatty\snatty$
is the identity.
We show that upcast and downcasts on the other hand satisfy a
\emph{decomposition} theorem: if $\sAone \ltdyn \sAtwo\ltdyn
\sAin{3}$, then the upcast from $\sAone$ to $\sAin{3}$ factors through
$\sAtwo$ and similarly for the downcast.

Furthermore, if we disregard \emph{performance} of the casts, and only
care about the observational behavior, we show that any ``general''
cast is the composition of an upcast followed by a
downcast.\footnote{Note that this is not the same as the factorization
  of casts known as ``threesomes'', see \secref{section:related:casts} for a
  comparison.}
For instance, our cast from before
$\obcast{\snatty\stimes\sdynty}{\sdynty\stimes\snatty}$ is
observationally equivalent to the composition of first
\emph{up}casting to a pair where both sides are dynamically typed
$\sdynty\stimes \sdynty$ and then \emph{down}casting:
$\obcast{\sdynty\stimes\sdynty}{\sdynty\stimes\snatty}\obcast{\snatty\stimes\sdynty}{\sdynty\stimes\sdynty}$.
We show that \emph{all} the casts in a standard gradually typed
language exhibit this factorization, which means that for the purposes
of formulating and proving graduality, we need only discuss upcasts
and downcasts.
For implementation, it is more convenient to have a primitive notion
of coercion/direct cast to eliminate/collapse casts
\cite{herman-tomb-flanagan-2010,siek-wadler10}, but we argue that the
correctness of such an implementation should be justified by a
simulation relation with a simpler semantics, meaning the
implementation would inherit a proof of graduality from the simpler
semantics as well.

To prove these equivalence and approximation results, we develop a
novel step-indexed logical relation that is sound for observational
error approximation.
We also develop high-level reasoning principles from the relation so
that our main lemmas do not involve any manual step-manipulation. 

Finally, based on our semantic interpretation of type dynamism as
embedding-projection pairs, we provide a refined analysis of the proof
theory of type dynamism as a \emph{syntax for building ep pairs}.
We give a semantics for these proof terms analogous to the coercion
interpretation of subtyping derivations.
Similar to subtyping, we prove a \emph{coherence} theorem which gives, 
as a corollary, our decomposition theorem for upcasts and downcasts.

\paragraph{Graduality}
In \citet{refined}, they prove the (dynamic) gradual guarantee by
an operational simulation argument whose details are quite tied to the 
specific cast calculus used.
Using the ep pairs, we provide a more semantic formulation and proof
of graduality.
First, we use our analysis of type dynamism as denoting ep pairs to
define graduality as a kind of observational error approximation
\emph{up to upcast/downcast}, building on the axiomatic semantics of
graduality in \citet{newlicata2018}.
We then prove the graduality theorem using our logical relation for
error approximation.
Notably, the decomposition theorem for ep pairs leads to a clean,
uniform proof of the cast case of graduality.

\paragraph{Overview of Technical Development and Contributions}
In this paper, we show how to prove graduality for a standard gradually typed cast
calculus by translating it into a simple typed language with recursive
types and errors.  Specifically, our development proceeds as follows: 
\begin{enumerate}
\item We present a standard gradually typed cast calculus
  ($\glangname$) and its operational semantics, using ``general''
  casts (\secref{sec:gradual}).
\item We present a simple typed language with recursive types and a
  type error ($\tlangname$), into which we translate the cast
  calculus.  Casts in $\glangname$ are translated to contracts
  implemented in the typed language (\secref{sec:typed}). 
\item We develop a novel step-indexed logical relation that is
  sound for our notion of observational error approximation
  (\secref{sec:logrel}).  We prove transitivity of the logical
  relation and other high-level reasoning principles so that our main
  lemmas for ep-pairs and graduality do not involve any manual
  step-manipulation.   
\item We present a novel analysis of type dynamism as a coherent syntax
  for ep pairs and show that all of the casts of the gradual language
  can be factorized as an upcast followed by a downcast (\secref{sec:ep-pairs}). 
\item We give a semantic formulation of graduality and then prove it using
  our error-approximation logical relation and ep pairs
  (\secref{sec:graduality}). 
\end{enumerate}

\section{Gradual Cast Calculus}
\label{sec:gradual}

Our starting point is a fairly typical gradual cast calculus, called
$\glangname$, in the style of \citet{wadler-findler09} and
\citet{refined}.
A cast calculus is usually the target of an elaboration pass from a
gradually typed surface language.
The gradually typed surface language makes mixing static and dynamic
code seamless, for instance a typed function on numbers $f :
\mathbb{N} \to \mathbb{N}$ can be applied to a dynamically typed value
$x :\dyn$ and the result is well typed $f(x) : \mathbb{N}$.
Since $x$ is not known to be a number, at runtime a dynamic check is
performed: if $x$ is a number, $f$ is run with its value and otherwise
a dynamic type error is raised.
In the surface language, this checking behavior takes place at every
elimination form: pattern matching, referencing a field, etc.
The cast calculus makes the dynamic type checking separate from the
elimination forms using explicit cast forms.
If $\st : \sA$ in the cast calculus, then we can cast it to another
type $\sB$ using the cast form $\obcast\sA\sB\st$.
This means we can use the ordinary typed reduction rules for
elimination forms, and all the details of checking are isolated to the
cast reductions.
We choose to use a cast calculus, rather than a gradual surface
language, since we are chiefly concerned with the semantics of the
language, rather than gradual type checking.

\begin{figure}
  \begin{mathpar}
    \begin{array}{lrcl}
      \mbox{Types} & \sty,\styalt & \bnfdef &
    {\sdynty} \bnfalt
    {{\sunitty}}\bnfalt
    {{\spairty{\sty}{\styalt}}}\bnfalt
    {{\ssumty{\sty}{\styalt}}}\bnfalt
    {{\sfunty{\sty}{\styalt}}}\\
      \mbox{Tags} & \stagty & \bnfdef & \sunitty\bnfalt\sdynty\stimes\sdynty\bnfalt\sdynty\splus \sdynty\bnfalt\sdynty\sto\sdynty\\

      \mbox{Terms} & \sterm,\stermalt & \bnfdef &
    \serr \bnfalt {\svar} \bnfalt
    \obcast{\sA}{\sB}\st\bnfalt
    {\stunit} \bnfalt
    {\stpair{\stermone}{\stermtwo}}\bnfalt
    {\stmatchpair{\sx}{\sy}{\st}{\ss}}\\
      & & \bnfalt &
    {\stinj{\sterm}}\bnfalt \stinjpr{\sterm}\bnfalt
    {{\stcase{\sterm}{\svarone}{\stermone}{\svartwo}{\stermtwo}}} \bnfalt
    {\stfun{\svar}{\sA}{\sterm}} \bnfalt
    {\stapp{\st}{\ss}}\\
      \mbox{Values} & \sval & \bnfdef & \obcast{\stagty}{\sdynty}{\sval} \bnfalt
    {\stunit} \bnfalt
    {\stpair{\svalone}{\svaltwo}} \bnfalt
    {\stinj{\sval}} \bnfalt
    {\stinjpr{\sval}} \bnfalt
    {\stfun{\svar}{\sty}{\sterm}} 
    \\
    \mbox{Evaluation Contexts} & \sectxt & \bnfdef &
    \shole \bnfalt
    \obcast{\sA}{\sB}{\sectxt}\bnfalt
    \stpair{\sectxt}{\ss}\bnfalt
    \stpair{\sv}{\sectxt} \bnfalt
    \stmatchpair{\sx}{\sy}{\sectxt}{\ss}\\
    &&\bnfalt&
    \stinj{\sectxt}\bnfalt
    \stinjpr{\sectxt}\bnfalt
    \stcase{\sectxt}{\sxone}{\stone}{\sxtwo}{\sttwo}\bnfalt
    \sectxt\,\ss\bnfalt
    \sv\,\sectxt\\
    \mbox{Environments} & \senv & \bnfdef & \sempenv \bnfalt \senvext{\senv}{\svar}{\sty}\\
    \mbox{Substitutions} & \sgamma & \bnfdef & \cdot \bnfalt \sgamma, \sv/\sx
    \end{array}
  \end{mathpar}
  \caption{$\glangname$ Syntax}
  \label{fig:glang:syntax}
\end{figure}

\begin{figure}
\flushleft{\fbox{\small{$\sjudgtc{\senv}{\st}{\sA}$}}}

\vspace{-2ex}
    \begin{mathpar}
    \inferrule{~}{\sjudgtc{\senv}{\serr}{\sA}}\and
    \inferrule{~}{\sjudgtc{\senv,\sx:\sA,\senvpr}{\sx}{\sA}}\and
    \inferrule{\sjudgtc{\senv}{\st}{\sA}}
    {\sjudgtc{\senv}{\obcast{\sA}{\sB}{\st}}{\sB}}\and
    \inferrule{~}{\sjudgtc{\senv}{\stunit}{\sunitty}}\and
    \inferrule
    {\sjudgtc{\senv}{\stermone}{\styone}\and
    \sjudgtc{\senv}{\stermtwo}{\stytwo}}
    {\sjudgtc{\senv}{\stpair{\stermone}{\stermtwo}}{\spairty{\styone}{\stytwo}}}\and
    \inferrule
    {\sjudgtc{\senv}{\st}{\sAone\stimes\sAtwo}\and
      \sjudgtc{\senv,\sx:\sAone,\sy:\sAtwo}{\ss}{\sB}}
    {\sjudgtc{\senv}{\stmatchpair{\sx}{\sy}{\st}{\ss}}{\sB}}\and
\iflong
    \inferrule
    {\sjudgtc{\senv}{\st}{\sA}}
    {\sjudgtc{\senv}{\stinj{\st}}{\ssumty{\sA}{\sApr}}}\and
    \inferrule
    {\sjudgtc{\senv}{\st}{\sApr}}
    {\sjudgtc{\senv}{\stinjpr{\st}}{\ssumty{\sA}{\sApr}}}\and
    \inferrule{\sjudgtc{\senv}{\st}{\sA\splus\sApr}\and
      \sjudgtc{\senvext{\senv}{\sx}{\sA}}{\ss}{\sB}\and
      \sjudgtc{\senvext{\senv}{\sxpr}{\sApr}}{\sspr}{\sB}}
    {\sjudgtc{\senv}{\stcase{\st}{\sx}{\ss}{\sxpr}{\sspr}}{\sB}}\quad
\fi
    \inferrule
    {\sjudgtc{\senv,\sx:\sA}{\st}{\sB}}
    {\sjudgtc{{\senv}}{\stfun{\sx}{\sA}{\st}}{\sA\sto\sB}}\and
    \inferrule
    {\sjudgtc{\senv}{\st}{\sA\sto\sB}\and
      \sjudgtc{\senv}{\ss}{\sA}}
    {\sjudgtc{\senv}{\st\,\ss}{\sB}}
    \end{mathpar}
    \caption{$\glangname$ Typing Rules }
    \label{fig:glang:typing:extended}
\end{figure}

We present the syntax of $\glangname$ (pronounced ``lambda gee'' and
typeset in $\sfont{\textsf{blue sans-serif font}}$) in
\figref{fig:glang:syntax}, and the typing rules in
\figref{fig:glang:typing:extended}. 
The language is call-by-value and includes standard type formers,
namely, the unit type $\sunitty$, product type $\stimes$, sum type
$\splus$, and function type $\sto$, with standard typing rules.
The language also includes some features specific to
gradual typing: a dynamic type $\sdynty$, a dynamic type error $\serr$
and casts $\obcast\sA\sB\st$.
Following previous work, the interface for the dynamic type $\sdynty$
is given by the casts themselves, and not distinct introduction and
elimination forms.
The values of the dynamic type are of the form
$\obcast{\stagty}{\sdynty}\sv$ where $\stagty$ ranges over \emph{tag types},
defined in \figref{fig:glang:syntax}.
The tag types are so called because they represent the ``tags'' used
to distinguish between the basic sorts of dynamically typed values.
Every type except $\sdynty$ has an ``underlying'' tag type we write as
$\floor\sA$ and define in \figref{fig:glang:tag}.
These tag types are the cases of the dynamic type $\sdynty$ seen as a
sum type, which is how we model it in \secref{sec:typed:translated}.
For any two types $\sA, \sB$, we can form the cast $\obcast\sA\sB$
which at runtime will attempt to coerce a term $\sv : \sA$
into a valid term of type $\sB$.
If the value cannot sensibly be interpreted as a value in $\sB$, the
cast \emph{fails} and reduces to the dynamic type error $\serr$.
The type error is like an uncatchable exception, modeling the fact
that the program crashes with an error message when a dynamic type
error is encountered. In this paper we consider all type errors to be
equivalent.
The calculus is based on that of \citet{wadler-findler09}, but does not have
blame and removes the restriction that types must be compatible in
order to define a cast.

\begin{figure}
\flushleft{\fbox{\small{$\floor{\sA} \defeq \stagty$}}~\mbox{\small{where~$\sA \neq \sdynty$}}}

\vspace{-4ex}
  \begin{mathpar}
    \begin{array}{rcl}
      \floor{\sunitty} & \defeq & \sunitty\\
      \floor{\sA \stimes \sB} & \defeq & \sdynty \stimes \sdynty \\
      \floor{\sA \splus \sB} & \defeq & \sdynty \splus \sdynty \\
      \floor{\sA \sto \sB} & \defeq & \sdynty \sto \sdynty \\
    \end{array}
  \end{mathpar}
  \caption{$\glangname$: Tag of a (non-dynamic) Type}
  \label{fig:glang:tag}
\end{figure}

\begin{figure}
\begin{mathpar}
    \sinholestep
      {\stcase{\parened{\stinj{\sv}}}{\sx}{\st}{\sxpr}{\stpr}}
      {\st{}[\sv/\sx]}

    \inferrule
    {}
    {\sinholestep
      {\stcase{\parened{\stinjpr{\sv}}}{\sx}{\st}{\sxpr}{\stpr}}
      {\stpr{}[\sv/\sxpr]}}

    \inferrule
    {}
    {\sinholestep
      {\stmatchpair{\sxone}{\sxtwo}{\stpair{\svone}{\svtwo}}{\st}}
      {\st[\svone/\sxone,\svtwo/\sxtwo]}}

    \inferrule{}
    {\sinholestep{\stapp{\parened{\stufun{\svar}{\sterm}}}{\sval}}{\subst{\sterm}{\sval}{\svar}}}

    \inferrule{}
    {\sectxt\hw\serr \step \serr}
  \end{mathpar}\\
\hrule
%
%
  \begin{mathpar}
    \inferrule*[right=DynDyn]
    {}
    {\sinholestep{\obcast{\sdynty}{\sdynty}{\sv}}{\sv}}

    \inferrule*[right=TagUp]
    {\sA \neq \sdynty \and \floor\sA \neq \sA}
    {\sinholestep
      {\obcast{\sA}{\sdynty}{\sv}}
      {\obcast{\floor\sA}{\sdynty}{\sparened{\obcast{\sA}{\floor\sA}{\sv}}}}}

    \inferrule*[right=TagDn]
    {\sA \neq \sdynty \and \floor\sA \neq \sA}
    {\sinholestep
      {\obcast{\sdynty}{\sA}{\sv}}
      {\obcast{\floor\sA}{\sA}\obcast{\sdynty}{\floor\sA}{\sv}}}

    \inferrule*[right=TagMatch]
    {}
    {\sinholestep{{\obcast{\sdynty}{\stagty}{\obcast{\stagty}{\sdynty}{\sv}}}}{\sv}}

    \inferrule*[right=TagMismatch]
    {\stagty \neq \stagtypr}
    {\sectxt\hw{\obcast{\sdynty}{\stagty}{\obcast{\stagtypr}{\sdynty}{\sv}}} \step \serr}

    \inferrule*[right=TagMismatch']
    {\sA,\sB \neq \sdynty\and
      \floor\sA\neq\floor\sB}
    {\sectxt\hw{\obcast{\sA}{\sB}\sv} \step \serr}
    
    \inferrule*[right=Pair]
    {}
    {\sinholestep
      {\obcast{\spairty{\sAone}{\sBone}}{\spairty{\sAtwo}{\sBtwo}}{\stpair{\sv}{\svpr}}}
      {\stpair{\obcast{\sAone}{\sAtwo}{\sv}}{\obcast{\sBone}{\sBtwo}{\svpr}}}}

    \inferrule*[right=Sum]
    {}
    {\sinholestep
    {\obcast{\ssumty{\sAone}{\sBone}}{\ssumty{\sAtwo}{\sBtwo}}{\stinj{\sv}}}
    {\obcast{\sAone}{\sAtwo}{\sv}}}

    \inferrule*[right=Sum']
    {}
    {\sinholestep
    {\obcast{\ssumty{\sAone}{\sBone}}{\ssumty{\sAtwo}{\sBtwo}}{\stinjpr{\sv}}}
    {\obcast{\sBone}{\sBtwo}{\sv}}}

    \inferrule*[right=Fun]
        {}
        {\sinholestep{\obcast{\sfunty{\sAone}{\sBone}}{\sfunty{\sAtwo}{\sBtwo}}{\sv}}{\stfun{\sx}{\sAtwo}{\obcast{\sBone}{\sBtwo}{\sparened{\stapp{\sv}{\sparened{\obcast{\sAtwo}{\sAone}{\sx}}}}}}}}

  \end{mathpar}
  \caption{$\glangname$ Operational Semantics: non-casts (top) and
    casts (bottom)}
 \label{fig:glang:opsem}
\end{figure}

\figref{fig:glang:opsem} presents the operational semantics of the
gradual language in the style of \citet{felleisen-hieb}, using
\emph{evaluation contexts} $\sectxt$ to specify a left-to-right,
call-by-value evaluation order. 
The top of the figure shows the reductions \emph{not} involving
casts.
This includes the standard reductions for pairs, sums, and functions
using the obvious notion of substitution $\st{}[\sgamma]$, in addition
to a reduction $\sectxt\hw{\serr} \step \serr$ to propagate a dynamic
type error to the top level.

More importantly, the bottom of the figure shows the reductions of
\emph{casts}, specifying the dynamic type checking necessary for
gradual typing.
First (\textsc{DynDyn}), casting from dynamic to itself is the identity.
For any type $\sA$ that is not a tag type (checked by $\floor\sA \neq
\sA$) or the dynamic type, casting to the dynamic type first casts to
its underlying tag type $\floor\sA$ and then tags it at that type (\textsc{TagUp}).
Similarly, casting down from the dynamic type first casts to the
underlying tag type (\textsc{TagDn}).
The next two rules are the primitive reductions for tags: if you project at
the correct tag type, you get the underlying value out (\textsc{TagMatch}) and otherwise a
dynamic type error is raised (\textsc{TagMismatch}).
Similarly, the next rule (\textsc{TagMismatch'}) says that if two
types are incompatible in that they have distinct tag types and
neither is dynamic, then the cast errors.
The next three (\textsc{Fun, Pair, Sum}) are the standard
``wrapping'' implementations of contracts/casts \cite{findler-felleisen02}, also familiar from
subtyping.
For the function cast $\obcast{\sAone\sto\sBone}{\sAtwo\sto\sBtwo}$,
note that while the output type is the same direction
$\obcast{\sBone}{\sBtwo}$, the input cast is flipped:
$\obcast{\sAtwo}{\sAone}$.

We note that this standard operational semantics is quite complex for
such a \emph{small} language.
In particular, it is more complicated than the operational semantics
of typed and dynamically typed languages of similar size.
Typed languages have reductions for each elimination form and
dynamically typed languages add only the possibility of type error to
those reductions.
Here on the other hand, the semantics is \emph{not} modular in the
same way: there are five rules involving the dynamic type and four of them
involve comparing arbitrary types.

For these reasons, we find the cast calculus presentation inconvenient
for semantic analysis, and we choose not to develop our theory of
graduality or even prove type safety \emph{directly} for this language.
Instead, we will \emph{translate} the cast calculus into a typed
language where the casts are translated to functions implemented in
the language, i.e. contracts \cite{findler-felleisen02}.
This has the advantage of reducing the size of the language, making
``language-level'' theorems like type safety and soundness of a
logical relation easier to prove.
Finally, note that our central theorems are still about the gradual
language, but we will prove them by lifting results about their
translations using an \emph{adequacy} theorem (\cref{lem:adequacy}).

\section{Translating Gradual Typing}
\label{sec:typed}

We now translate our cast calculus into a simpler, non-gradual typed
language with errors.
We then prove an adequacy theorem that enables us to prove theorems
about gradual programs by reasoning about their translations.

\subsection{Typed Language with Errors}

The typed language we will translate into is $\lambda_{T,\mho}$
(pronounced ``lambda tee error'' and typeset in $\mfont{\textbf{bold
    red serif font}}$),
a call-by-value typed lambda calculus with iso-recursive types and an
uncatchable error. 
\figref{fig:tlang:syntax} shows the syntax of the language.  
\figref{fig:tlang:typing} shows some of the typing rules; the rest are
completely standard.

\begin{figure}
    \begin{mathpar}
    \begin{array}{lrcl}
      \mbox{Types} & \mty,\mtyalt & \bnfdef &
        \mmuty{\malpha}{\mty} \bnfalt \malpha \bnfalt \munitty\bnfalt \mA \mtimes \mB\bnfalt \msumty{\mty}{\mtyalt} \bnfalt \mfunty{\mty}{\mtyalt} 
      \\
      \mbox{Terms} & \mt,\ms & \bnfdef &
        {\merr} \bnfalt
        {\mvar} \bnfalt
        \mtlet\mx\mt\ms\bnfalt
        {\mtroll{\mA}{\mterm}} \bnfalt
        {\mtunroll{\mterm}} \bnfalt
        \mtunit\bnfalt
        \mtpair\mt\ms\\
          & & \bnfalt &
        \mtmatchpair{\mx}{\my}{\mt}{\ms}\bnfalt
        {\mtinj{\mterm}} \bnfalt
        {\mtinjpr{\mterm}} \\
          & & \bnfalt &
        {\mtcase{\mterm}{\mvarone}{\mtermone}{\mvartwo}{\mtermtwo}}\bnfalt
        {\mtfun{\mvar}{\mA}{\mterm}} \bnfalt
        {\mtapp{\mterm}{\mtermalt}} \\
      \mbox{Values} & \mval & \bnfdef & \mvar \bnfalt \mtroll{\mA}{\mval} \bnfalt \mtunit \bnfalt \mtpair\mv\mv \bnfalt\mtinj{\mval} \bnfalt\mtinjpr{\mval} \bnfalt \mtfun{\mvar}{\mA}{\mterm}\\
      \mbox{Evaluation Contexts} & \mectxt & \bnfdef & 
        {\mhole} \bnfalt
        \mtlet\mx\mE\ms\bnfalt
        {\mtroll{\mA}{\mectxt}} \bnfalt
        {\mtunroll{\mectxt}} \bnfalt
        {\mtpair\mE\mt}\bnfalt
        \mtpair\mv\mE\bnfalt
        \mtmatchpair{\mx}{\my}{\mE}{\ms}\\
          & & \bnfalt &
        {\mtinj{\mectxt}} \bnfalt
        {\mtinjpr{\mectxt}} \bnfalt
        {\mtcase{\mectxt}{\mvarone}{\mtermone}{\mvartwo}{\mtermtwo}}\bnfalt
        {\mtapp{\mectxt}{\mtermalt}} \bnfalt
        {\mtapp{\mval}{\mectxt}} \\
        \iflong
      \mbox{Contexts} & \mctxt & \bnfdef & 
        {\mhole} \bnfalt
        \mtlet\mx\mC\ms\bnfalt
        \mtlet\mx\mt\mC\bnfalt
        {\mtroll{\mA}{\mctxt}} \bnfalt
        {\mtunroll{\mctxt}}\\
        &&\bnfalt&
        {\mtpair{\mctxt}{\mt}}\bnfalt
        {\mtpair{\mt}{\mctxt}}\bnfalt
        {\mtmatchpair{\mx}{\my}{\mctxt}{\mt}}\bnfalt
        {\mtmatchpair{\mx}{\my}{\mt}{\mctxt}}
        \\ & & \bnfalt &
        {\mtinj{\mctxt}} \bnfalt
        {\mtinjpr{\mctxt}}\bnfalt
        {\mtcase{\mctxt}{\mvarone}{\mtermone}{\mvartwo}{\mtermtwo}}
        \\ &&\bnfalt&
        {\mtcase{\mterm}{\mvarone}{\mctxtone}{\mvartwo}{\mtermtwo}}\bnfalt
        {\mtcase{\mterm}{\mvarone}{\mtermone}{\mvartwo}{\mctxttwo}}
        \\ & & \bnfalt &
        {\mtfun{\mvar}{\mA}{\mctxt}} \bnfalt
        {\mtapp{\mctxt}{\mtermalt}} \bnfalt
        {\mtapp{\mterm}{\mctxt}}\\
        \fi
    \mbox{Environments} & \menv & \bnfdef & \mempenv \bnfalt \menvext{\menv}{\mvar}{\mty}\\
    \mbox{Substitutions} & \mgamma & \bnfdef & \cdot \bnfalt \mgamma, \mv/\mx
    \end{array}
    \end{mathpar}
    \caption{$\tlangname$ Syntax}
    \label{fig:tlang:syntax}
\end{figure}
\begin{figure}
\flushleft{\fbox{\small{$\mjudgtc{\menv}{\mt}{\mty}$}}}

\vspace{-2ex}
  \begin{mathpar}
      \inferrule{~}{\mjudgtc{\menv}{\merr}{\mty}}\and
      \inferrule
      {\mvar : \mty \in \menv}
      {\mjudgtc{\menv}{\mvar}{\mty}}\and
\iflong
      \inferrule
      {\mjudgtc{\menv}{\mt}{\mA} \and
      \mjudgtc{\menv,\mx:\mA}{\ms}{\mB}}
      {\mjudgtc{\menv}{\mtlet{\mx}{\mt}{\ms}}{\mB}}\and
\fi
      \inferrule{\mjudgtc{\menv}{\mterm}{\mA[\mmuty{\malpha}{\mA}/\malpha]}}
      {\mjudgtc{\menv}{\mtroll{\mmuty{\malpha}{\mA}}{\mterm}}{\mmuty{\malpha}{\mA}}}\and
      \inferrule{\mjudgtc{\menv}{\mterm}{\mmuty{\malpha}{\mA}}}
      {\mjudgtc{\menv}{\mtunroll{\mterm}}{\mA[\mmuty{\malpha}{\mA}/\malpha]}}\and
\iflong
      \inferrule
      {~}
      {\mjudgtc{\menv}{\mtunit}{\munitty}}\and
      \inferrule
      {\mjudgtc{\menv}{\mt}{\mA}\and \mjudgtc{\menv}{\ms}{\mB}}
      {\mjudgtc{\menv}{\mtpair{\mt}{\ms}}{\mA \mtimes \mB}}\and
      \inferrule
      {\mjudgtc{\menv}{\mt}{\mAone\mtimes \mAtwo}\and
        \mjudgtc{\menv,\mx:\mAone,\my:\mAtwo}{\ms}{\mB}}
      {\mjudgtc{\menv}{\mtmatchpair{\mx}{\my}{\mt}{\ms}}{\mB}}\and
      \inferrule
      {\mjudgtc{\menv}{\mterm}{\mA}}
      {\mjudgtc{\menv}{\mtinj{\mterm}}{\mA \mplus \mApr}}\and
      \inferrule
      {\mjudgtc{\menv}{\mterm}{\mApr}}
      {\mjudgtc{\menv}{\mtinjpr{\mterm}}{\msumty{\mA}{\mApr}}}\and
      \inferrule
      {\mjudgtc{\menv}{\mterm}{\mA \mplus \mApr} \and
        \mjudgtc{\menv,\mx:\mA}{\ms}{\mB} \and
      \mjudgtc{\menv,\mxpr:\mApr}{\mspr}{\mB}}
      {\mjudgtc{\menv}{\mtcase{\mterm}{\mx}{\ms}{\mxpr}{\mspr}}{\mB}}\quad
      \inferrule
      {\mjudgtc{\menv,\mx:\mA}{\mt}{\mB}}
      {\mjudgtc{\menv}{\mtfun{\mx}{\mA}{\mt}}{\mA \mto \mB}}\quad
      \inferrule
      {\mjudgtc{\menv}{\mt}{\mA \mto \mB} \and
      \mjudgtc{\menv}{\ms}{\mA}}
      {\mjudgtc{\menv}{\mt\,\ms}{\mB}}
\fi
  \end{mathpar}
    \caption{$\tlangname$ Typing Rules }
    \label{fig:tlang:typing}
\end{figure}

The types of the language are similar to the cast calculus: they
include the standard type formers of products, sums, and functions.
Rather than the specific dynamic type, we include the more general,
but standard, iso-recursive type $\mmuty{\malpha}{\mA}$,
which is isomorphic to the unfolding
$\mA[\mmuty{\malpha}{\mA}/\malpha]$ by the terms
$\mtroll{\mmuty{\malpha}{\mA}}{\cdot}$ and $\mtunroll{\cdot}$.
As in the source language we have an uncatchable error $\merr$.

\begin{figure}
  \begin{mathpar}
    \inferrule{}
    {\mectxt[\merr] \stepsin{0} \merr}\and
    \inferrule{}
    {\inholestep{\mtlet\mx\mv\ms}{\ms{}[\mv/\mx]}}\and
    \inferrule
    {}
    {\inholestepsin{1}
      {\mtunroll{(\mtroll{\mA}{\mval})}}
      {\mval}}\and
    \inferrule
    {}
    {\inholestep
      {\mtmatchpair{\mxone}{\mxtwo}{\mtpair{\mvone}{\mvtwo}}{\mt}}
      {\mt[\mvone/\mxone,\mvtwo/\mxtwo]}}\and

    \inferrule
    {}
    {\inholestep{(\mtfun{\mx}{\mA}{\mt})\,\mv}{\mt{}[\mv/\mx]}}\and
    \inferrule
    {}
    {\inholestep
      {\mtcase{\parened{\mtinj{\mv}}}{\mx}{\mt}{\mxpr}{\mtpr}}
      {\mt{}[\mv/\mx]}}\and
    \inferrule
    {}
    {\inholestep
      {\mtcase{\parened{\mtinjpr{\mv}}}{\mx}{\mt}{\mxpr}{\mtpr}}
      {\mtpr{}[\mv/\mxpr]}}
  \end{mathpar}
  \begin{mathpar}
    \inferrule
    {~}
    {\mt \bigstep{0} \mt}\and
    \inferrule
    {\mt \stepsin{i} \mtpr \and \mtpr \bigstep{j} \mtpr[2]}
    {\mt \bigstep{i + j} \mtpr[2]}
  \end{mathpar}
  \caption{$\tlangname$ Operational Semantics}
  \label{fig:tlang:opsem}
\end{figure}

\figref{fig:tlang:opsem} presents the operational semantics of
the language.
For the purposes of later defining a step-indexed logical relation, we
assign a weight to each small step of the operational semantics that 
is $1$ for unrolling a value of recursive type and $0$ for other
reductions.
We then define a ``quantitative'' reflexive, transitive closure of
the small-step relation $\mt \bigstep{i} \mtpr$
that adds the weights of its constituent small steps.
When the number of steps is irrelevant, we just use $\step$ and 
$\bigstepany$. 
We can then establish some simple facts about this operational
semantics.

\begin{lemma}[Subject Reduction]
  If $\cdot\vdash \mt : \mA$ and $\mt \bigstepany \mtpr$ then
  $\cdot\vdash\mtpr : \mA$.
\end{lemma}

\begin{lemma}[Progress]
  If $\cdot \vdash \mt : \mA$ and $\mt$ is not a value or $\merr$,
  then there exists $\mtpr$ with $\mt \step \mtpr$.
\end{lemma}
\iflong
\begin{proof}
  By induction on the typing derivation for $\mt$.
\end{proof}
\fi

\begin{lemma}[Determinism]
  If $\mt \step \ms$ and $\mt \step \mspr$, then $\ms = \mspr$.
\end{lemma}

\subsection{Translating Gradual Typing}
\label{sec:typed:translated}

Next we translate the cast calculus into our typed language, and prove
that the cast calculus semantics is in a simulation relation with the
typed language.
Since the two languages share so much of their syntax, most of the
translation is a simple ``color change'', only the parts that are
truly components of gradual typing need much translation.

\begin{figure}
  \begin{mathpar}
    \begin{array}{rcl}
      \sem{\sdynty} &\defeq& \mmuty{\malpha}{\munitty \mplus (\mpairty{\malpha}{\malpha}) \mplus (\malpha \mplus \malpha) \mplus  (\mfunty{\malpha}{\malpha})} \\
      \sem{\sunitty} &\defeq& \munitty\\
      \sem{\spairty{\sA}{\sB}} &\defeq& \mpairty{\sem{\sA}}{\sem{\sB}}\\
      \sem{\ssumty{\sA}{\sB}} &\defeq& \msumty{\sem{\sA}}{\sem{\sB}}\\
      \sem{\sfunty{\sA}{\sB}} &\defeq& \mfunty{\sem{\sA}}{\sem{\sB}}\\
    \end{array}
  \end{mathpar}
  \caption{Type Translation}
  \label{fig:type-translation}
\end{figure}

Our translation is type preserving, so we first define a \emph{type
  translation} in \figref{fig:type-translation}.
The dynamic type is interpreted as a recursive sum of the translations
of the tag types of the gradual language.
The unit, pair, sum and function types are all interpreted as the
corresponding connectives in the typed language.

\begin{figure}
  \fbox{\small{$\sem{\st}$}} \small{$~~$where if $~\sxone:\sAone,\ldots,\sxn:\sAn \vdash \st : \sA~$ then
$~\mxone:\sem{\sAone},\ldots,\mxn : \sem{\sAn} \vdash \sem{\st} : \sem{\sA}$} \hfill
  \begin{mathpar}
    \begin{array}{rcl}
      \sem{\svar} & \defeq & \mvar\\
      \sem{\obcast\sA\sB\st} & \defeq & \mectxt_{\obcast\sA\sB}\hw{\sem{\st}}\\
      \sem\stunit & \defeq & \mtunit\\
      \sem{\stpair{\stone}{\sttwo}} & \defeq& \mtpair{\sem{\stone}}{\sem{\sttwo}}\\
      \sem{\stmatchpair\sx\sy\st\ss} & \defeq & \mtmatchpair\mx\my{\sem\st}{\sem\ss}\\
      \sem{\stinj\st} & \defeq & \mtinj{\sem\st}\\
      \sem{\stinjpr\st} & \defeq & \mtinjpr{\sem\st}\\
      \sem{\stcase{\st}{\sx}{\ss}{\sxpr}{\sspr}} & \defeq & \mtcase{\sem\st}{\mx}{\sem\ss}{\mxpr}{\sem\sspr}\\
      \sem{\stfun{\sx}{\sA}{\st}} & \defeq & \mtfun{\mx}{\sem{\sA}}{\sem{\st}}\\
      \sem{\st\,\ss} & \defeq & \sem{\st}\,\sem{\ss}\\
    \end{array}
  \end{mathpar}
  \caption{Term Translation}
  \label{fig:term-translation}
\end{figure}
\begin{figure}
  \fbox{\small{$\mectxt_{\obcast{\sA}{\sB}}$}} \small{$~~$where $~\mx : \sem \sA \vdash \mE_{\obcast{\sA}{\sB}}\hw{\mx} : \sem \sB$} \hfill
  \begin{mathpar}
    \begin{array}{rcl}
      \mectxt_{\obcast{\sdynty}{\sdynty}} & \defeq & \mhole\\
      \mectxt_{\obcast{\sAone\stimes\sBone}{\sAtwo\stimes\sBtwo}} & \defeq &
      \mectxt_{\obcast{\sAone}{\sAtwo}} \mtimes \mectxt_{\obcast{\sBone}{\sBtwo}}\\
      \mectxt_{\obcast{\sAone\splus\sBone}{\sAtwo\splus\sBtwo}} & \defeq &
      \mectxt_{\obcast{\sAone}{\sAtwo}} \mplus \mectxt_{\obcast{\sBone}{\sBtwo}}\\
      \mectxt_{\obcast{\sAone \sto \sBone}{\sAtwo\sto\sBtwo}} & \defeq &
      \mectxt_{\obcast{\sAtwo}{\sAone}} \mto \mectxt_{\obcast{\sBone}{\sBtwo}}\\
      \mectxt_{\obcast{\stagty}{\sdynty}} & \defeq & \mtroll{\sem\sdynty}{\mtsum{\stagty}\mhole}\\
      \mectxt_{\obcast{\sdynty}{\stagty}} & \defeq & \mtelsecase{(\mtunroll\mhole)}{\stagty}{\mx}{\mx}{\merr}\\
      \mectxt_{\obcast{\sA}{\sdynty}} & \defeq & 
           \mectxt_{\obcast{\floor\sA}{\sdynty}}\hw{\mectxt_{\obcast{\sA}{\floor\sA}}\mhole}
           \qquad \mbox{if $\sA \neq \sdynty, \floor\sA$}\\
      \mectxt_{\obcast{\sdynty}{\sA}} & \defeq & 
           {\mectxt_{\obcast{\floor\sA}{\sA}}}\hw{\mectxt_{\obcast{\sdynty}{\floor\sA}}\mhole}
           \qquad \mbox{if $\sA \neq \sdynty,\floor\sA$}\\
      \mectxt_{\obcast{\sA}{\sB}} & \defeq &
           \mtlet{\mx}{\mhole}{\merr}
           \qquad \mbox{if $\sA,\sB\neq\sdynty$ and $\floor\sA\neq\floor\sB$}\\
    \end{array}
  \end{mathpar}
  \caption{Direct Cast Translation}
  \label{fig:direct-cast-translation}
\end{figure}

\begin{figure}
      \begin{mathpar}
        \begin{array}{rcl}
          \mectxt \mtimes \mectxtpr & \defeq &
          \mtmatchpair{\mx}{\mxpr}{\mhole}{\mtpair{\mectxt\hw\mx}{\mectxtpr\hw\mxpr}}\\

          \mectxt \mplus \mectxtpr & \defeq &
          \mtcase{\mhole}
          {\mx}{\mectxt\hw\mx}
          {\mxpr}{\mectxtpr\hw\mxpr}\\

          \mectxt \mto \mectxtpr & \defeq &
          \mtlet{\mvarin{f}}{\mhole}{\mtufun{\mvarin{a}}{\mectxtpr\hw{\mvarin{f}\,(\mectxt\hw{\mvarin{a}})}}}
        \end{array}
      \end{mathpar}
      \caption{Functorial Action of Type Connectives}
      \label{fig:functor}
\end{figure}

Next, we define the translation of terms in
\figref{fig:term-translation}, which is type preserving in that if
$\sxone:\sAone,\ldots,\sxn:\sAn \vdash \st : \sA$ then
$\mxone:\sem{\sAone},\ldots,\mxn : \sem{\sAn} \vdash \sem{\st} :
\sem{\sA}$.
Again, most of the translation is just a change of hue.
The most important rule of the term translation is that of casts.
%
A cast $\obcast\sA\sB$ is translated to an \emph{evaluation context}
$\mectxt_{\obcast\sA\sB}$ of the appropriate type, which are defined
in \figref{fig:direct-cast-translation}.
Each case of the definition corresponds to one or more rules of the
operational semantics.
The product, sum, and function rules use the definitions of functorial
actions of their types from \figref{fig:functor}.
We separate them because we will use the functoriality property in
several definitions, theorems, and proofs later.

\subsection{Operational Properties}

Next, we consider the relationship between the operational semantics
of the two languages and how to lift properties of the typed language
to the gradual language.
We want to view the translation of the cast calculus into the typed
language as \emph{definitional}, and in that regard view the
operational semantics of the source language as being based on the
typed language. 
We capture this relationship in the following \emph{forward}
simulation theorem, which says that any reduction in the cast calculus
corresponds to (and is \emph{justified by}) multiple steps in the target:

\begin{lemma}[Translation Preserves Values, Evaluation Contexts]
\label{lem:sem-val-evctx}{~}
  \begin{enumerate}
  \item For any value $\sv$, $\sem{\sv}$ is a value.
  \item For any evaluation context $\sectxt$, $\sem\sectxt$ is an
    evaluation context.
  \end{enumerate}
\end{lemma}
\begin{lemma}[Simulation of Operational Semantics]
\label{lem:forward-simulation}{~}\\
  If $\st \step \stpr$ then there exists $\ms$ with $\sem{\st} \step
  \ms$ and $\ms \bigstepany \sem{\stpr}$.
\end{lemma}
\iflong\begin{proof}
  By cases of $\st \step \stpr$. The non-cast cases are clear by
  \cref{lem:sem-val-evctx}.
  \begin{enumerate}
  \item DynDyn
    \begin{align*}
      \mE_{\obcast\sdynty\sdynty}\hw{\sem\sv} &=
      \mtlet\mx{\sem\sv}\mx\\
      &\step
      \sem\sv
    \end{align*}
  \item TagUp: Trivial because $\sem{\obcast\sA\sdynty\sv} =
    \sem{\obcast{\floor\sA}{\sdynty}\obcast\sA{\floor\sA}\sv}$.
  \item TagDn: Trivial because $\sem{\obcast\sdynty\sA\sv} =
    \sem{\obcast{\floor\sA}\sA\obcast{\sdynty}{\floor\sA}\sv}$.
  \item (TagMatch) Valid because
    \begin{align*}
    \mtelsecasevert{\mtunroll\mtroll{\sem\sdynty}{\mtsum{\stagty}{\sem\sv}}}{\stagty}{\mx}{\mx}{\merr}
    &\step
    \mtelsecasevert{{\mtsum{\stagty}{\sem\sv}}}{\stagty}{\mx}{\mx}{\merr}\\
    &\step
    \sem\sv
    \end{align*}
  \item (TagMismatch) Valid because
    \begin{align*}
    \mtelsecasevert{\mtunroll\mtroll{\sem\sdynty}{\mtsum{\stagtypr}{\sem\sv}}}{\stagty}{\mx}{\mx}{\merr}
    &\step
    \mtelsecasevert{{\mtsum{\stagtypr}{\sem\sv}}}{\stagty}{\mx}{\mx}{\merr}\\
    &\step
    \merr\\
    & = \sem\serr
    \end{align*}
  \item (TagMismatch') Valid because
    \[ \mtlet\mx{\sem\sv}\merr \step \merr \]
  \item Pair Valid by
    \begin{align*}
      \mectxt_{\obcast{\sAone\stimes\sBone}{\sAtwo\stimes\sBtwo}}\hw{\mtpair{\sem\sv}{\sem\svpr}}
      &=
      \mtlet{\mx}{\my}{\mtpair{\sem\sv}{\sem\svpr}}{\mtpair{\mectxt_{\obcast{\sAone}{\sAtwo}}\hw\mx}{\mectxt_{\obcast{\sBone}{\sBtwo}}\hw\my}}\\
      &\step {\mtpair{\mectxt_{\obcast{\sAone}{\sAtwo}}\hw{\sem\sv}}{\mectxt_{\obcast{\sBone}{\sBtwo}}\hw{\sem\svpr}}}\\
      &= \sem{\stpair{{\obcast{\sAone}{\sAtwo}}{\sv}}{{{\obcast{\sBone}{\sBtwo}}{\svpr}}}}
    \end{align*}
  \item Sum
    \begin{align*}
      \mectxt_{\obcast{\sAone\stimes\sBone}{\sAtwo\stimes\sBtwo}}\hw{\mtinj{\sem\sv}}
      &=
      \mtcasevert{\mtinj{\sem\sv}}
      {\mx}{\mE_{\obcast{\sAone}\sAtwo}\hw{\mx}}
      {\mxpr}{\mE_{\obcast{\sBone}\sBtwo}\hw{\mxpr}}\\
      &\step
      \mE_{\obcast\sAone\sAtwo}\hw{\sem\sv}\\
      &= \sem{\obcast\sAone\sAtwo\sv}
    \end{align*}
  \item Sum'
    \begin{align*}
      \mectxt_{\obcast{\sAone\stimes\sBone}{\sAtwo\stimes\sBtwo}}\hw{\mtinjpr{\sem\sv}}
      &=
      \mtcasevert{\mtinjpr{\sem\sv}}
      {\mx}{\mE_{\obcast{\sAone}\sAtwo}\hw{\mx}}
      {\mxpr}{\mE_{\obcast{\sBone}\sBtwo}\hw{\mxpr}}\\
      &\step
      \mE_{\obcast\sBone\sBtwo}\hw{\sem\sv}\\
      &= \sem{\obcast\sBone\sBtwo\sv}
    \end{align*}
  \item (Fun) Valid because
    \begin{align*}
      \mectxt_{\obcast{\sAone\sto\sBone}{\sAtwo\sto\sBtwo}}\hw{\sem\sv}
      &\step
      \mtletvert{\mxin f}{\sem\sv}{}\\
      &\step
      \mtfun{\mxin a}{\mAtwo}{\mectxt_{\obcast\sBone\sBtwo}\hw{\sem\sv\,(\mectxt_{\obcast\sAtwo\sAone}\hw{\mxin a})}}\\
      &=
      \sem{\stfun{\sxin a}{\sAtwo}
        {\obcast{\sBone}{\sBtwo}
          {({\sv}\,({\obcast{\sAtwo}{\sAone}{\sxin{a}}}))}}}
    \end{align*}
  \end{enumerate}
\end{proof}\fi

To lift theorems for the gradual language from the typed language, we
need to establish an \emph{adequacy} theorem, which says
that the operational behavior of a translated term determines the
behavior of the original source term.
To do this we use the following backward simulation theorem.
\ifshort
\else
\begin{lemma}[Translation reflects Results]
\label{lem:reflect}
  \begin{enumerate}
  \item If $\sem{\st}$ is a value, $\st \stepstar \sv$ for some $\sv$
    with $\sem\st=\sem\sv$.
  \item If $\sem{\st} = \merr$, then $\st \stepstar \serr$.
  \end{enumerate}
\end{lemma}
\begin{proof}
  By induction on $\st$. For the non-casts, follows by inductive
  hypothesis. For the casts, only two cases can be values:
  \begin{enumerate}
  \item $\obcast\sdynty\sdynty\st$: if $\sem{\obcast\sdynty\sdynty\st}
    = \sem\st$ is a value then by inductive hypothesis, $\st$ is a
    value, so $\obcast\sdynty\sdynty\sv \step \sv$.
  \item $\obcast{\stagty}\sdynty\st$: if
    $\mtroll{\sem\sdynty}{\mtsum{\stagty}{\sem\st}}$ is a value, then
    $\sem\st$ is a value so by inductive hypothesis $\st \stepstar
    \sv$ so $\obcast{\stagty}{\sdynty}\st \stepstar
    \obcast{\stagty}{\sdynty}\sv$.
  \end{enumerate}

  For the error case, there is only one case where it is possible for
  $\sem\st = \merr$ without $\st = \serr$:
  \begin{enumerate}
  \item For $\obcast\sdynty\sdynty\ss$, if
    $\sem{\obcast\sdynty\sdynty}\sem\ss = \sem\ss$ is an error then
    clearly $\sem\ss = \merr$ so by inductive hypothesis $\ss \stepstar \serr$
    and because casts are strict,
    \[ \obcast\sdynty\sdynty\ss \stepstar \ss \]
  \end{enumerate}
\end{proof}
\begin{lemma}[Backward Simulation]
\label{lem:backward-simulation}
  If $\sem\st \step \ms$ then there exists $\sspr$ with $\st \step
  \sspr$ and $\ms \bigstepany \sem\sspr$.
\end{lemma}
\begin{proof}
  By induction on $\st$. We show two illustrative cases, the rest
  follow by the same reasoning.
  \begin{enumerate}
  \item $\sem{\stmatchpair{\sx}{\sy}{\st}{\ss}} =
    \mtmatchpair{\mx}{\my}{\sem\st}{\sem\ss}$. If $\sem\st$ is not a
    value, then we use the inductive hypothesis. If $\sem\st$ is a
    value and $\st$ is then by \lemref{lem:reflect} $\st
    \stepstar \sv$ and then we can reduce the pattern-match in source
    and target.
  \item $\sem{\obcast{\sAone\sto\sBone}{\sAtwo\sto\sBtwo}\st} =
    (\mE_{\obcast{\sAtwo}{\sAone}} \mto
    \mE_{\obcast{\sBone}{\sBtwo}})\hw{\sem\st}$.  If $\sem\st$ is
    not a value, we use the inductive hypothesis.  Otherwise, if it is
    a value and $\st$ is not, we use \lemref{lem:reflect} to get
    ${\obcast{\sAone\sto\sBone}{\sAtwo\sto\sBtwo}\st}
    \stepstar{\obcast{\sAone\sto\sBone}{\sAtwo\sto\sBtwo}\sv}$. Then
    we use the same argument as the proof of
    \lemref{lem:forward-simulation}.
  \item
    $\sem{\obcast{\sA}{\sdynty}\st}=\mectxt_{\obcast{\floor\sA}{\sdynty}}\hw{\mectxt_{\obcast{\sA}{\floor\sA}}\hw{\sem\st}}$
    then we use the same argument as the case for
    $\obcast{\sA}{\floor\sA}\st$, e.g., the function case above.
  \end{enumerate}
\end{proof}
\fi

\begin{theorem}[Adequacy]
\label{lem:adequacy}{~}
  \begin{enumerate}
  \item $\sem\st \bigstepany \mv$ if and only if $\st \stepstar \sv$ with
    $\sem\sv = \mv$.
  \item$\sem\st \bigstepany \merr$ if and only if $\st \stepstar \serr$.
  \item $\sem\st$ diverges if and only if $\st$ diverges
  \end{enumerate}
\end{theorem}
\iflong\begin{proof}
  The forward direction for values and errors is given by forward
  simulation \lemref{lem:forward-simulation}. The backward direction for values
  and errors is given by induction on $\sem\st\bigstepany \mtpr$,
  backward simulation \lemref{lem:backward-simulation} and reflection of results
  \cref{lem:reflect}.

  If $\st$ diverges, then by the backward value and error cases, it
  follows that $\sem\st$ does not run to a value or error. By type safety
  of the typed language, $\sem\st$ diverges.

  Finally, if $\sem\st$ diverges, we show that $\st$ diverges. If
  $\sem\st \step \ms$, then by backward simulation, there exists
  $\sspr$ with $\sem\st \step \sspr$ and $\ms \bigstepany
  \sem\sspr$. Since $\sem\st \bigstepany \sem\sspr$, we know
  $\sem\sspr$ diverges, so by coinduction $\sspr$ diverges and
  therefore $\st$ diverges.
\end{proof}\fi

While this has reduced the number of primitives of the language,
reasoning about the behavior of the translated casts isn't any simpler
than the original operational semantics since they have the same
behavior.
For simpler reasoning about cast behavior, we will move further away
from a direct simulation of the source operational semantics, to a
second semantics based on ep pairs that is observationally equivalent
but also conceptually simpler and helps prove the gradual guarantee.
However, in order to prove that the second semantics is equivalent, we
first need to develop a usable theory of observational equivalence and approximation.

\section{A Logical Relation for Error Approximation}
\label{sec:logrel}

Next, we define observational equivalence and error approximation of
programs in the gradual and typed languages, the two properties with
which we formulate embedding-projection pairs.
To facilitate proofs of error approximation, we develop a novel step-indexed
logical relation. 
Since our notion of approximation is non-standard, the use of
step-indexing in our logical relation is inconvenient to use directly.
So, on top of the ``implementation'' of the logical relation as a
step-indexed relation, we prove many high-level lemmas so that all
proofs in the next sections are performed relative to these lemmas,
and none manipulate step indices directly.

\subsection{Observational Equivalence and Approximation}

A suitable notion of equivalence for programs is \emph{observational
  equivalence}.
We say $\st$ is observationally equivalent to $\stpr$ if replacing one
with the other in the context of a larger program produces the same
result (termination, error, or divergence).
We formalize this saying a program \emph{context} $\sC$ is a term
with a single hole $\hole$.
A context is typed $\senvpr \vdash \sC\hw{\senv\vdash \cdot : \sA} :
\sApr$ when for any term $\senv\vdash \st : \sA$, replacing the hole
with $\st$ results in a well-typed $\senvpr\vdash\sC\hw{\st} : \sApr$

While this notion of observational equivalence is entirely standard,
the notion of \emph{approximation} we use---which we call error
approximation---is not the standard notion of observational
approximation.
Usually, we would say $\st$ observationally approximates $\stpr$ if, when
placing them into the same context $\sC$, either $\sC[\st]$ diverges or
they both terminate or both error.
We call this form of approximation \emph{divergence approximation}.
However, for \emph{gradual typing} we are not particularly interested
in when one program diverges more than another, but rather when it
produces more \emph{type errors}.
We might be tempted to conflate the two, but their behavior is quite
distinct!
We can never truly know if a black-box program will continue
indefinitely: it would frustrate any programmer to use a language that
runs forever when accidentally using a function as a number.
The reader should keep this difference in mind when seeing how our
logical relation differs form the standard treatment.
In the rest of this paper, when discussing the two together we will
clearly distinguish between divergence and error approximation, but
when there is no qualifier, 
approximation is meant as \emph{error approximation}. 

\begin{definition}[Gradual Observational Equivalence, Error
  Approximation]
  For any well typed terms $\senv \vdash \st, \stpr : \sA$,
  \begin{enumerate}
  \item Define $\senv \vDash \st \obseqv \stpr : \sA$, pronounced
    ``$\st$ is observationally equivalent to $\stpr$'' to hold when for
    any $\cdot \vdash \sC\hw{\senv\vdash\cdot:\sA} : \sB$, either $\sC[\st]$ and
    $\sC[\stpr]$ both reduce to a value, both reduce to an error, or
    both diverge.
  \item Define $\senv \vDash \st \obsapprox \stpr : \sA$, pronounced
    ``$\st$ observationally (error) approximates $\stpr$'' to hold
    when for any $\cdot \vdash \sC\hw{\senv\vdash\cdot:\sA} :
    \sB$, either $\sC[\st]$ reduces to $\serr$ or both $\sC[\st]$
    and $\sC[\stpr]$ reduce to a value or both diverge.
  \end{enumerate}  
\end{definition}

As with divergence approximation, we can prove two programs
are observationally equivalent by showing each error approximates the other.
\begin{lemma}[Equivalence is Approximation Both Ways]
  $\senv \vDash \stone \obseqv \sttwo : \sA$ if and only if both
  $\senv \vDash \stone \obsapprox \sttwo : \sA$ and $\senv \vDash
  \sttwo \obsapprox \stone : \sA$.
\end{lemma}

We define typed observational equivalence $\menv \vDash \mt \obseqv
\mtpr : \mA$ and observational error approximation $\menv \vDash \mt
\obsapprox \mtpr : \mA$ with the exact same definition as for the
gradual language above, but in $\mfont{\textbf{red}}$ instead of
$\sfont{\textsf{blue}}$.
We rarely work with the gradual language directly, instead we prove
approximation results for their translations.
\iflong
\begin{definition}[Typed Observational equivalence, Error Approximation]
  For any well typed terms $\menv \vdash \mt, \mtpr : \mA$,
  \begin{enumerate}
  \item Define $\menv \vDash \mt \obseqv \mtpr : \mA$, pronounced
    ``$\mt$ is observationally equivalent to $\mtpr$'' to hold when for
    any $\mC[\cdot:\mA] : \sB$, either $\mC[\mt]$ and
    $\mC[\mtpr]$ both reduce to a value, both reduce to an error, or
    both diverge.
  \item Define $\menv \vDash \mt \obsapprox \mtpr : \mA$, pronounced
    ``$\mt$ observationally (error) approximates to $\mtpr$'' to hold when
    for any $\mC[\cdot:\mA] : \sB$, either $\mC[\mt]$ reduces to
    $\merr$ or both $\mC[\mt]$ and $\mC[\mtpr]$ reduce to a value or
    both diverge.
  \end{enumerate}  
\end{definition}
\fi
This is justified by the following lemma, a consequence of our
\emph{adequacy} result (\cref{lem:adequacy}).
\begin{lemma}[Typed Observational Approximation implies Gradual Observational Approximation]
\label{lem:typed-gradual-approx}
  If $\sem\senv \vDash \sem\stone \obsapprox \sem\sttwo : \sem\sA$ then
  $\senv \vDash \stone \obsapprox \sttwo : \sA$
\end{lemma}
\iflong\begin{proof}
  For any $\sC$, by compositionality of the translation,
  $\sem{\sC[\stone]} = \sem{\mC}[\sem\stone]$ and $\sem{\sC[\sttwo]} =
  \sem{\sC}[\sem\sttwo]$.
  Then we analyze $\sem\stone \obsapprox
  \sem\sttwo$
  \begin{enumerate}
  \item If $\sem\sC[\sem\stone] \bigstep \merr$ then \cref{lem:adequacy}
    states that $\sC[\stone] \bigstep \serr$.
  \item If $\sem\sC[\sem\stone]$ diverges, then $\sem\sC[\sem\sttwo]$
    also diverges and therefore by \cref{lem:adequacy}, $\sC[\stone]$
    and $\sC[\sttwo]$ diverge.
  \item If $\sem\sC[\sem\stone] \bigstepany \svone$, then
    $\sem\sC[\sem\sttwo] \bigstepany \svtwo$ and therefore by
    \cref{lem:adequacy}, $\sC[\stone] \stepstar \svone$ and
    $\sC[\sttwo] \stepstar \svtwo$.
  \end{enumerate}  
\end{proof}\fi

\subsection{Logical Relation}  Observational equivalence and approximation are extremely difficult to
prove directly, so we use the usual method of proving observational
results by using a \emph{logical relation} that we prove \emph{sound}
with respect to observational approximation.
Due to the non-well-founded nature of recursive types (and the dynamic
type specifically), we develop a \emph{step-indexed} logical relation
following \citet{ahmed06:lr}.
We define our logical relation for error approximation in
\figref{fig:lr}.
Because our notion of error approximation is not the standard notion
of approximation, the definition is a bit unusual, but this is
necessary for technical reasons.

It is key to compositional reasoning about embedding-projection pairs
that approximation be \emph{transitive} and care must be taken to show
transitivity for a step-indexed relation.
However, for standard definitions of logical relations for
observational equivalence, it is difficult to prove
transitivity directly.  Therefore, it is often established through indirect
reasoning---e.g., by setting up a biorthogonal ($\top\top$-closed) 
logical relation so one can easily show it is complete with respect to
observational equivalence, which in turn implies that it must be
transitive since observational equivalence is easily proven transitive.
The reason establishing transitivity is tricky is that a step-indexed relation
is \emph{not} transitive at a \emph{fixed} index, i.e., if $e_1 \isim i e_2$ and
$e_2 \isim i e_3$ it is not necessarily the case that $e_1 \isim i
e_3$.  For instance, $e_1 \isim i e_2$ might be related because $e_1$ terminates
in less than $i$ steps and has the same behavior as $e_2$ which takes more than
$i$ steps to terminate, whereas $e_2 \isim i e_3$ are related because they both
take $i$ steps of reduction so cannot be distinguished in $i$ steps but have
different behavior when run for more steps. 
%
One direct method for proving transitivity, originally presented in
\citet{ahmed06:lr}, is to observe that two terms are observationally
equivalent when each divergence approximates the other, and then use a
step-indexed relation for divergence approximation.
Because a conjunction of transitive relations is transitive, this
proves transitivity of equivalence.
A step-indexed relation for divergence approximation can be shown to
have a kind of ``half-indexed'' transitivity, i.e., if $e_1 \iprec i
e_2$ and for every natural $j$, we know $e_2 \iprec j e_3$ then $e_1
\iprec i e_3$.
We have a similar issue with error approximation: the na\"ive logical
relation for error approximation is not clearly transitive.
Inspired by the case of observational equivalence, we similarly
``split'' our logical relation into two relations that can be proven
transitive by an argument similar to divergence approximation.
However, unlike for observational equivalence, our two relations are not the
same.  Instead, one $\ltdynr$ is error approximation up to divergence on the
\emph{left} and the other $\gtdynr$ is error approximation up to
divergence on the \emph{right}.

For a given natural number $i \in \mathbb{N}$ and type $\mA$, and
\emph{closed terms} $\mtone,\mttwo$ of type $\mA$, $\mtone
\tltdynr{i}{\mA} \mttwo$ intuitively means that, if we only inspect
$\mtone$'s behavior up to $i$ uses of $\mtunroll\cdot$, then it
appears that $\mtone$ error approximates $\mttwo$.
Less constructively, it means that we cannot show that $\mtone$ does
\emph{not} error approximate $\mttwo$ when limited to $i$ uses of
$\mtunroll\cdot$.
However, even if we knew $\mtone \tltdynr{i}{\mA} \mttwo$ for
\emph{every} $i \in \mathbb{N}$, it still might be the case that
$\mtone$ diverges, since no finite number of unrolling can ever
exhaust $\mtone$'s behavior.
So we also require that we know $\mtone \tgtdynr{i}{\mA} \mttwo$,
which means that up to $i$ uses of unroll on $\mttwo$, it appears that
$\mtone$ error approximates $\mttwo$.

\begin{figure}
  \begin{mathpar}
    \begin{array}{rcl}
      \tltdynr{i}{\mA}, \tgtdynr{i}{\mA} & \subseteq & \{\mt \mathrel{|} \cdot \vdash \mt : \mA\}^2\\
     \mtone \tltdynr{i}{\mA} \mttwo & \defeq & (\exists \mtonepr.~ \mtone \bigstep{i+1}\mtonepr)\\
      & & \vee (\exists j\leq i.~ \mtone \bigstep{j} \merr)\\
      & & \vee (\exists j\leq i, \mvone \vltdynr{i-j}{\mA} \mvtwo.~ \mtone \bigstep{j} \mvone \wedge \mttwo \bigstepany \mvtwo)\\
      \mtone \tgtdynr{i}{\mA} \mttwo & \defeq & (\exists \mttwopr.~ \mttwo \bigstep{i+1}\mttwopr)\\
      & & \vee (\exists j\leq i.~ \mttwo \bigstep{j} \merr \wedge \mtone \bigstepany \merr)\\
      & & \vee (\exists j\leq i, \mvtwo.~ \mttwo \bigstep{j} \mvtwo \wedge\\
      & & \qquad (\mtone \bigstepany \merr \vee \exists \mvone.~ \mtone \bigstepany \mvone \wedge \mvone \vgtdynr{i-j}{\mA} \mvtwo)\\\\

     \vltdynr{i}{\mA}, \vgtdynr{i}{\mA} & \subseteq & \{\mv \mathrel{|} \cdot \vdash \mv : \mA \}^2
          \qquad \text{where } {\simdynr} \in \{{\dotltdynr},{\dotgtdynr} \}\\
  \mvone \vsimdynr{0}{\mmuty{\malpha}{\mA}} \mvtwo & \defeq&  \top\\
  \mtroll{\mmuty{\malpha}{\mA}}{\mvone} \vsimdynr{i+1}{\mmuty{\malpha}{\mA}} \mtroll{\mmuty{\malpha}{\mA}}{\mvtwo} & \defeq &
  \mvone \vsimdynr{i}{\mA[\malpha\mapsto \mmuty{\malpha}{\mA}]} \mvtwo\\
  \mtunit \vsimdynr{i}{\munitty} \mtunit & \defeq & \top\\
  \mtpair{\mvone}{\mvonepr} \vsimdynr{i}{\mpairty{\mA}{\mApr}} \mtpair{\mvtwo}{\mvtwopr} & \defeq & \mvone \vsimdynr{i}{\mA} \mvtwo \wedge \mvonepr \vsimdynr{i}{\mApr} \mvtwopr\\
  \mvone \vsimdynr{i}{\msumty{\mA}{\mB}} \mvtwo & \defeq &
  (\exists (\mvonepr \vsimdynr{i}{\mA} \mvtwopr) \wedge \mvone = \mtinj{\mvonepr}\wedge \mvtwo = \mtinj{\mvtwopr}) \\
  & & \vee (\exists (\mvonepr \vsimdynr{i}{\mB} \mvtwopr) \wedge \mvone = \mtinjpr{\mvonepr}\wedge \mvtwo = \mtinjpr{\mvtwopr})\\
  \mvone \vsimdynr{i}{\mfunty{\mA}{\mB}} \mvtwo & \defeq & \forall j \leq i. \forall (\mvonepr \vsimdynr{j}{\mA} \mvtwopr).~ \mtapp{\mvone}{\mvonepr} \tsimdynr{i}{\mB} \mtapp{\mvtwo}{\mvtwopr}\\\\
  \cdot \vsimdynr{i}{\cdot} \cdot & \defeq & \top\\
    \mgammain{1},\mvin{1}/\mx \vsimdynr{i}{\menv,\mvar:\mA} \mgammain{2},\mvin{2}/\mx & \defeq & \mgammain{1} \vsimdynr{i}{\menv} \mgammain{2} \wedge \mvin{1} \vsimdynr{i}{\mA} \mvin{2}\\\\
     \menv \vDash \mtone \simdynr \mttwo : \mA & \defeq & \forall i \in \mathbb{N}, (\mgammain{1}
      \vsimdynr{i}{\menv} \mgammain{2}).~ \mtone[\mgammain{1}]
      \tsimdynr{i}{\mA} \mttwo[\mgammain{2}]\\\\
      \menv \vDash \mtone \dynr \mttwo : \mA & \defeq & \menv \vdash
      \mtone \ltdynr \mttwo : \mA \wedge \menv \vdash \mtone \gtdynr
      \mttwo : \mA
    \end{array}
  \end{mathpar}
  \caption{$\tlangname$ Error Approximation Logical Relation}
  \label{fig:lr}
\end{figure}

The above intuition should help to understand the definition of error
approximation for terms (i.e., the relations $\ltdynr_t$ and $\gtdynr_t$).
The relation $\mtone\tltdynr{i}{\mA} \mttwo$ is defined by inspection
of $\mtone$'s behavior: it holds if $\mtone$ is still running after
$i+1$ unrolls; or if it steps to an error in fewer than $i$ unrolls; or
if it results in a value in fewer than $i$ unrolls and also $\mttwo$
runs to a value and those values are related for the remaining steps.
The definition of $\mtone \tgtdynr{i}{\mA} \mttwo$ is
defined by inspection of $\mttwo$'s behavior: it holds if $\mttwo$ is
still running after $i+1$ unrolls; or if $\mttwo$ steps to an error in
fewer than $i$ steps then $\mtone$ errors as well; or if $\mttwo$
steps to a value, either $\mtone$ errors or steps to a value 
related for the remaining steps.

While the relations 
and $\gtdynr_t$ on terms are different, fortunately, the relations on
values are essentially the same, so we abstract over the cases by
having the symbol $\simdynr$ to range over either $\ltdynr$ or
$\gtdynr$.
For values of recursive type, if the step-index is $0$, we consider
them related, because otherwise we would need to perform an unroll to 
inspect them further.
Otherwise, we decrement the index and check if they are related.
Decrementing the index here is exactly what makes the definition of
the relation well-founded.
For the standard types, the value relation definition is indeed
standard: pairs are related when the two sides are related, sums must
be the same case and functions must be related when applied to any
related values in the future (i.e., when we may have exhausted some of
the available steps). 

Finally, we extend these relations to \emph{open} terms in the
standard way: we define substitutions to be related pointwise (similar
to products) and then say that $\menv \vDash \mtone \simdynr \mttwo :
\mA$ holds if for every pair of substitutions $\mgammaone, \mgammatwo$
related for $i$ steps, the terms after substitution, written
$\mtone[\mgammaone]$ and $\mttwo[\mgammatwo]$, are related for $i$ steps.
Then our resulting relation $\menv \vDash \mtone \dynr \mttwo$ is
defined to hold when $\mtone$ error approximates $\mttwo$ up to
divergence of $\mtone$ ($\ltdynr$), \emph{and} up to divergence of
$\mttwo$ ($\gtdynr$).

\iflong
\begin{figure}
  \begin{mathpar}
    \inferrule{~}{\semhomoltdyn{\menv}{\merr}{\merr}{\mty}}
    
      \inferrule
      {\mvar : \mty \in \menv}
      {\semhomoltdyn{\menv}{\mvar}{\mvar}{\mty}}

      \inferrule
      {\semhomoltdyn{\menv}{\mtone}{\mttwo}{\mA} \and
      \semhomoltdyn{\menv,\mx:\mA}{\msone}{\mstwo}{\mB}}
      {\semhomoltdyn{\menv}{\mtlet{\mx}{\mtone}{\msone}}{\mtlet{\mx}{\mttwo}{\mstwo}}{\mB}}

      \inferrule{\semhomoltdyn{\menv}{\mtermone}{\mtermtwo}{\mA[\mmuty{\malpha}{\mA}/\malpha]}}
      {\semhomoltdyn{\menv}{\mtroll{\mmuty{\malpha}{\mA}}{\mtermone}}{\mtroll{\mmuty{\malpha}{\mA}}{\mtermtwo}}{\mmuty{\malpha}{\mA}}}

      \inferrule{\semhomoltdyn{\menv}{\mtermone}{\mtermtwo}{\mmuty{\malpha}{\mA}}}
      {\semhomoltdyn{\menv}{\mtunroll{\mtermone}}{\mtunroll{\mtermtwo}}{\mA[\mmuty{\malpha}{\mA}/\malpha]}}

      \inferrule
      {~}
      {\semhomoltdyn{\menv}{\mtunit}{\mtunit}{\munitty}}

      \inferrule
      {\semhomoltdyn{\menv}{\mtone}{\mttwo}{\mA}\and \semhomoltdyn{\menv}{\msone}{\mstwo}{\mB}}
      {\semhomoltdyn{\menv}{\mtpair{\mtone}{\msone}}{\mtpair{\mttwo}{\mstwo}}{\mA \mtimes \mB}}

      \inferrule
      {\semhomoltdyn{\menv}{\mtone}{\mttwo}{\mAone\mtimes \mAtwo}\and
        \semhomoltdyn{\menv,\mx:\mAone,\my:\mAtwo}{\msone}{\mstwo}{\mB}}
      {\semhomoltdyn{\menv}{\mtmatchpair{\mx}{\my}{\mtone}{\msone}}{\mtmatchpair{\mx}{\my}{\mttwo}{\mstwo}}{\mB}}
      
      \inferrule
      {\semhomoltdyn{\menv}{\mtermone}{\mtermtwo}{\mA}}
      {\semhomoltdyn{\menv}{\mtinj{\mtermone}}{\mtinj{\mtermtwo}}{\mA \mplus \mApr}}

      \inferrule
      {\semhomoltdyn{\menv}{\mtermone}{\mtermtwo}{\mApr}}
      {\semhomoltdyn{\menv}{\mtinjpr{\mtermone}}{\mtinjpr{\mtermtwo}}{\msumty{\mA}{\mApr}}}

      \inferrule
      {\semhomoltdyn{\menv}{\mtermone}{\mtermtwo}{\mA \mplus \mApr} \and
        \semhomoltdyn{\menv,\mx:\mA}{\msone}{\mstwo}{\mB} \and
      \semhomoltdyn{\menv,\mxpr:\mApr}{\msonepr}{\mstwopr}{\mB}}
      {\semhomoltdyn{\menv}{\mtcase{\mtermone}{\mx}{\msone}{\mxpr}{\msonepr}}{\mtcase{\mtermtwo}{\mx}{\mstwo}{\mxpr}{\mstwopr}}{\mB}}

      \inferrule
      {\semhomoltdyn{\menv,\mx:\mA}{\mtone}{\mttwo}{\mB}}
      {\semhomoltdyn{\menv}{\mtufun{\mx}{\mtone}}{\mtufun{\mx}{\mttwo}}{\mA \mto \mB}}

      \inferrule
      {\semhomoltdyn{\menv}{\mtone}{\mtone}{\mA \mto \mB} \and
      \semhomoltdyn{\menv}{\msone}{\mstwo}{\mA}}
      {\semhomoltdyn{\menv}{\mtone\,\msone}{\mttwo\,\mstwo}{\mB}}
  \end{mathpar}
  \caption{$\tlangname$ Error Approximation Congruence Rules}
  \label{fig:congruence}
\end{figure}
\fi

We need the following standard lemmas.
\begin{lemma}[Downward Closure]
  \label{lem:downward-closed}
  If $j \leq i$ then
  \begin{enumerate}
  \item If $\mtone \tsimdynr{i}{\mA} \mttwo$ then $\mtone
    \tsimdynr{j}{\mA} \mttwo$
  \item If $\mvone \vsimdynr{i}{\mA} \mvtwo$ then $\mvone
    \vsimdynr{j}{\mA} \mvtwo$.
  \end{enumerate}
\end{lemma}
\iflong
\begin{proof}
  By lexicographic induction on the pair $(i,\mA)$.
\end{proof}
\fi
\begin{lemma}[Anti-Reduction]
  \label{lem:anti-red}
  This theorem is different for the two relations as we allow arbitrary
  steps on the ``divergence greater-than'' side.
  \begin{enumerate}
  \item If $\mtone \ltdynlt^{i}_{t,\mA} \mttwo$ and $\mtonepr
    \bigstep{j} \mtone$ and $\mttwopr \bigstepany \mttwo$ then
    $\mtonepr \ltdynlt^{i + j}_{t,\mA} \mttwopr$.
  \item If $\mtone \ltdyngt^{i}_{t,\mA} \mttwo$ and $\mttwopr
    \bigstep{j} \mttwo$ and $\mtonepr \bigstepany \mtone$, then
    $\mtonepr \ltdyngt^{i + j}_{t,\mA} \mttwopr$.
  \end{enumerate}
  \iflong
  A simple corollary that applies in common cases to both relations is
  that if $\mtone \ltdynsim^{i}_{t,\mA} \mttwo$ and $\mtonepr
  \bigstep{0} \mtone$ and $\mttwopr \bigstep{0} \mttwo$, then
  $\mtonepr \ltdynsim^{i}_{t,\mA} \mttwopr$.
  \fi
\end{lemma}
\iflong
\begin{proof}
  By direct inspection and downward closure
  (\lemref{lem:downward-closed}).
\end{proof}
\fi
\begin{lemma}[Monadic Bind]
\label{lem:monadic-bind}
  For any $i \in \mathbb{N}$, if for any $j \leq i$ and $\mvone
  \tsimdynr{j}{\mA} \mvtwo$, we can show $\mectxtone\hw{\mvone}
  \tsimdynr{j}{\mA} \mectxttwo\hw{\mvtwo}$ holds, then for any $\mtone
  \tsimdynr{i}{\mA} \mttwo$, it is the case that
  $\mectxtone\hw{\mvone} \tsimdynr{i}{\mA}\mectxttwo\hw{\mvtwo}$.
\end{lemma}
\iflong
\begin{proof}
  We consider the proof for $\tgtdynr{i}{\mA}$, the other is
  similar/easier.  By case analysis of $\mtone \tgtdynr{i}{\mA}
  \mttwo$.
  \begin{enumerate}
  \item If $\mttwo$ takes $i+1$ steps, so does $\mectxttwo\hw{\mttwo}$.
  \item If $\mttwo \stepsin{j\leq i} \merr$ and $\mtone \stepstar
    \merr$, then first of all $\mectxttwo\hw{\mttwo}
    \stepsin{j}\mectxttwo\hw{\merr}\step \merr$. If $j+1 = i$, we are
    done. Otherwise $\mectxttwo\hw{\mttwo} \stepsin{j+1\leq i} \merr$
    and $\mectxtone\hw{\mtone} \stepstar \merr$.
  \item Assume there exist $j \leq i$, $\mvone \vgtdynr{i - j}{\mA}
    \mvtwo$ and $\mttwo \stepsin{j} \mvtwo$ and $\mtone \stepstar
    \mvone$.
    Then by assumption, $\mectxtone\hw{\mvone} \tgtdynr{i - j} \mectxttwo\hw{\mvtwo}$. Then by antireduction
    (\lemref{lem:anti-red}),
    $\mectxtone\hw{\mtone}\tgtdynr{i}\mectxttwo\hw{\mttwo}$.
    
  \end{enumerate}
\end{proof}
\fi

We then prove that our logical relation is sound for observational
error approximation by showing that it is a congruence relation and showing that if
we can prove error approximation up to divergence on the
left \emph{and} on the right, then we have true error approximation.
\iflong
\begin{lemma}[Congruence for Logical Relation]
  \label{lem:cong}
  All of the congruence rules in \figref{fig:congruence} are valid.
\end{lemma}
\begin{proof}
  Each case is done by proving the implication for $\ltdynlt^i$ and
  $\ltdyngt^i$. Most cases follow by monadic bind
  (\lemref{lem:monadic-bind}), downward closure
  (\lemref{lem:downward-closed}) and direct use of the inductive
  hypotheses.
  We show some illustrative cases.
  \begin{enumerate}
  \item Given $\mgammaone \simdynr{i}{\menv} \mgammatwo$, we need to
    show $\mtfun{\mx}{\mA}{\mtone[\mgammaone]} \tsimdynr{i}{\mA\mto\mB}
    \mtfun{\mx}{\mA}{\mttwo[\mgammatwo]}$. Since they are values, we
    show they are related values. Given any $\mvone \vsimdynr{j}{\mA}
    \mvtwo$ with $j\leq i$, each side $\beta$ reduces in $0$ unroll
    steps so it is sufficient to show
    \[ \mtone[\mgammaone,\mvonepr/\mx] \tltdynr{j}{\mB} \mttwo[\mgammatwo,\mvtwopr/\mx]\]
    Which follows by inductive hypothesis and downward-closure and the
    substitution relation.
  \end{enumerate}
\end{proof}
\fi

\begin{theorem}[Logical Relation implies Observational Error Approximation]
\label{thm:log-to-obs}
  If $\menv \vDash \mtone \ltdyn \mttwo : \mA$, then $\menv \vDash
  \mtone \obsapprox \mttwo : \mA$
\end{theorem}
\iflong\begin{proof}
  If $\menv\vDash \mtone \ltdyn \mttwo : \mA$, then for any closing
  context, by \cref{lem:cong}, $\cdot \vDash \mC[\mtone] \ltdyn
  \mC[\mttwo] : \sB$ holds.

  Then we do a case analysis of $\mC[\mtone]$'s behavior.
  \begin{enumerate}
  \item If $\mC[\mtone]$ diverges, then for any $i$, since
    $\mC[\mtone] \tgtdynr{i}{\sB} \mC[\mttwo]$, only the
    $\mC[\mttwo] \bigstep{i} \mtpr$ is possible, so $\mC[\mttwo]$ also
    diverges.
  \item If $\mC[\mtone] \bigstep{i} \merr$ we're done.
  \item If $\mC[\mtone] \bigstep{i} \mtunit$, then because
    $\mC[\mtone] \tltdynr{i}{\sB} \mC[\mttwo]$, we know
    $\mC[\mttwo] \bigstepany \mtunit$.
  \end{enumerate}
\end{proof}\fi

\subsection{Approximation and Equivalence Lemmas}
\label{sec:lemmas}

The step-indexed logical relation is on the face of it quite complex,
especially due to the splitting of error approximation into two
step-indexed relations.
However, we should view the step-indexed relation as an
``implementation'' of the high-level concept of error approximation,
and we work as much as possible with the error
approximation relation $\menv \vDash \mtone \ltdyn \mttwo : \mA$.
In order to do this we now prove some high-level lemmas, which are
proven using the step-indexed relations, but allow us to develop
conceptual proofs of the key theorems of the paper.

First, there is reflexivity, also known as the \emph{fundamental
  lemma}, which is proved using the same congruence cases as the
soundness theorem (\cref{thm:log-to-obs}.)
Note that by the definition of our logical relation, this is really a
kind of \emph{monotonicity} theorem for every term in the language,
the first component of our graduality proof.
\begin{corollary}[Reflexivity]
\label{lem:fund-lemma}
  If $\menv \vdash \mt : \mA$ then $\menv \vDash \mt \ltdyn \mt : \mA$
\end{corollary}
\iflong
\begin{proof}
  By induction on the typing derivation of $\mt$, in each case using
  the corresponding congruence rule \ifshort proved earlier. \fi
  \iflong from \lemref{lem:cong}.\fi
\end{proof}  
\fi
Next, crucial to reasoning about ep pairs is the use of
\emph{transitivity}, a notoriously tedious property to prove for
step-indexed logical relations.
This is where our splitting of error-approximation into two pieces
proves essential, adapting the approach for divergence-approximation
relations introduced in \citet{ahmed06:lr}.
The proof works as follows: due to the function and open-term cases,
we cannot simply prove transitivity in the limit directly.
Instead we get a kind of ``asymmetric'' transitivity: if $\mtone
\tltdynr{i}{\mA} \mttwo$ and \emph{for any} $j\in\mathbb{N}$, $\mttwo
\tltdynr{j}{\mA} \mtthree$, then we know $\mtone \tltdynr{i}{\mA}
\mtthree$. We abbreviate the $\forall j$ part as $\mttwo
\tltdynr{\omega}{\mA} \mtthree$ in what follows.
The key to the proof is in the function and open terms cases, which
rely on reflexivity, \cref{lem:fund-lemma}, as in
\citet{ahmed06:lr}. Reflexivity says that when we have $\mvone
\vltdynr{i}{\mA} \mvtwo$ then we also have $\mvtwo
\vltdynr{\omega}{\mA} \mvtwo$, which allows us to use the inductive
hypothesis. 

\begin{lemma}[Transitivity for Closed Terms/Values]
  \label{lem:trans-closed-homo}
  The following are true for any $\mA$.
  \begin{enumerate}
  \item If $\mtone \tltdynr{i}{\mA} \mttwo$ and $\mtin{2}
    \tltdynr{\omega}{\mA} \mtin{3}$ then $\mtin{1}
    \tltdynr{i}{\mA}\mtin{3}$.
    
  \item If $\mvone \tltdynr{i}{\mA} \mvtwo$ and $\mvin{2}
    \tltdynr{\omega}{\mA} \mvin{3}$ then $\mvin{1}
    \tltdynr{i}{\mA}\mvin{3}$.
  \end{enumerate}
  
  Similarly,
  \begin{enumerate}
  \item If $\mtone \tgtdynr{\omega}{\mA} \mttwo$ and
    $\mtin{2} \tgtdynr{i}{\mA} \mtin{3}$
    then $\mtin{1} \tgtdynr{i}{\mA}\mtin{3}$.

  \item If $\mvone \tgtdynr{\omega}{\mA} \mvtwo$ and
    $\mvin{2} \tgtdynr{i}{\mA} \mvin{3}$ then
    $\mvin{1} \tgtdynr{i}{\mA}\mvin{3}$.
  \end{enumerate}
\end{lemma}
\iflong
\begin{proof}
  We prove the $\tltdynr{i}{\mA}$ and $\vltdynr{i}{\mA}$ mutually by
  induction on $(i,\mA)$. The other logical relation is similar. Most
  value cases are simple uses of the inductive hypotheses.
  
  \begin{enumerate}
  \item (Terms) By case analysis of $\mtone \tltdynr{i}{\mA} \mttwo$.
    \begin{enumerate}
    \item If $\mtone \bigstep{i+1} \mtonepr$ or $\mtone \bigstep{j\leq
      i} \merr$, we have the result.
    \item Let $j\leq i$, $k \in \mathbb{N}$ and $(\mvone
      \vltdynr{i}{\mA} \mvtwo)$ with $\mtone \bigstep{j} \mvone$ and
      $\mttwo \bigstep{k} \mvtwo$. By inductive hypothesis for values,
      it is sufficient to show that $\mtin{3} \bigstepany \mvin{3}$
      and $\mvtwo \vltdynr{\omega}{\mA} \mvin{3}$.

      Since $\mtin{2} \tltdynr{\omega}{\mA} \mtin{3}$, in particular
      we know $\mtin{2} \tltdynr{k + l}{\mA} \mtin{3}$ for every $l
      \in \mathbb{N}$, so since $\mtin{2} \bigstep{k} \mvin{2}$, we
      know that $\mtin{3} \bigstepany \mvin{3}$ and $\mvin{2}
      \vltdynr{l}{\mA} \mvin{3}$, for every $l$, i.e., $\mvin{2}
      \vltdynr{\omega}{\mA} \mvin{3}$.
    \end{enumerate}

  \item (Function values) Suppose $\mvin{1}
    \vltdynr{i}{\mfunty{\mA}{\mB}} \mvin{2}$ and $\mvin{2}
    \vltdynr{\omega}{\mfunty{\mA}{\mB}} \mvin{3}$. Then, let $j \leq
    i$ and $\mvinpr{1} \vltdynr{j}{\mA} \mvinpr{2}$. We need to show
    $\mtapp{\mvin{1}}{\mvinpr{1}} \tltdynr{j}{\mB}
    \mtapp{\mvin{3}}{\mvinpr{2}}$. By inductive hypothesis, it is
    sufficient to show $\mtapp{\mvin{1}}{\mvinpr{1}} \tltdynr{j}{\mB}
    \mtapp{\mvin{2}}{\mvinpr{2}}$ and $\mtapp{\mvin{2}}{\mvinpr{2}}
    \tltdynr{\omega}{\mB} \mtapp{\mvin{3}}{\mvinpr{2}}$.

    The former is clear. The latter follows by the congruence rule for
    application \lemref{lem:cong} and reflexivity
    \cref{lem:fund-lemma} on $\mvinpr{2}$: since $\cdot \vdash \mvinpr{2}:
    \mA$, we have $\mvinpr{2} \vltdynr{\omega}{\mA} \mvinpr{2}$.
  \end{enumerate}
\end{proof}
\fi

\begin{lemma}[Transitivity]
\label{lem:trans}
  If $\menv \vdash \mtone \dynr \mttwo : \mA$ and $\menv \vdash \mttwo
  \dynr \mtin{3} : \mA$ then $\menv \vdash \mtone \dynr \mtin{3} :
  \mA$.
\end{lemma}
\iflong
\begin{proof}
  The argument is essentially the same as the function value case,
  invoking the fundamental lemma \cref{lem:fund-lemma} for each
  component of the substitutions and transitivity for the closed
  relation \lemref{lem:trans-closed-homo}
\end{proof}
\fi

Next, we want to extract approximation and equivalence principles for
\emph{open} programs from syntactic operational properties of
\emph{closed} programs.
First, obviously any operational reduction is a contextual
equivalence, and the next lemma extends that to open programs.
Note that we use $\equidyn$ to mean approximation in both directions,
i.e., equivalence:
\begin{lemma}[Open $\beta$ Reductions]
\label{lem:open-beta}
  Given $\menv \vdash \mt : \mA$, $\menv \vdash \mtpr : \mA$, if for
  every $\mgamma : \Gamma$, $\mt[\mgamma] \bigstepany \mtpr{}[\mgamma]$,
  then $\menv \vDash \mt \equidyn \mtpr : \mA$.
\end{lemma}
\iflong\begin{proof}
  By reflexivity \cref{lem:fund-lemma} on $\mtpr, \mgamma$
  and anti-reduction \cref{lem:anti-red}.
\end{proof}\fi
We call this open $\beta$ reduction because we will use it to justify
equivalences that look like an operational reduction, but have open
values (i.e. including variables) rather than closed as in the
operational semantics.
For instance,
\[ \mtlet{\mx}{\my}{\mt} \equidyn \mt[\my/\mx] \]
and
\[\mtmatchpair{\mx}{\my}{\mtpair{\mxpr}{\mypr}}{\mt} \equidyn \mt[\mxpr/\mx,\mypr/\my] \]

Additionally, it is convenient to use \emph{$\eta$} expansions for our
types as well.
Note that since we are using a call-by-value language, the $\eta$
expansion for functions is restricted to \emph{values}.
\begin{lemma}[$\eta$ Expansion]
  {~}
  \begin{enumerate}
  \item For any $\menv \vdash \mv : \mA \mto \mB$, $\mv \equidyn \mtfun{\mx}{\mA}{\mv\,\mx}$
  \item For any $\menv,\mx:\mA\mplus \mApr,\menvpr \vdash \mt : \mB$,
    \[ \mt \equidyn
    \mtcase{\mx}{\my}{\mt[\mtinj\mypr/\mx]}{\mypr}{\mt[\mtinjpr\mypr/\mx]}
    \]
  \item For any $\menv,\mx:\mA\mtimes \mApr,\menvpr \vdash \mt : \mB$,
    \[ \mt \equidyn
    \mtmatchpair{\my}{\mypr}{\mx}{\mt[\mtpair{\my}{\mypr}/\mx]}
    \]
  \end{enumerate}
\end{lemma}
\iflong
\begin{proof}
  All are consequences of \cref{lem:open-beta}.
\end{proof}
\fi

Next, with term constructors that involve continuations, we often need
to rearrange the programs such as the ``case-of-case'' transformation.
These are called commuting conversions and are presented in \figref{fig:comm-conv}.
\begin{figure}
  \begin{mathpar}
    \begin{array}{rcl}
      \mectxt\hw{\mtlet{\mx}{\mt}{\ms}} & \equidyn & \mtlet{\mx}{\mt}{\mectxt\hw\ms}\\
      \mectxt\hw{\mtmatchpair{\mx}{\my}{\mt}{\ms}} & \equidyn & \mtmatchpair{\mx}{\my}{\mt}{\mectxt\hw\ms}\\
      \mectxt\hw{\mtcase{\mt}{\mx}{\ms}{\mxpr}{\mspr}} & \equidyn & \mtcase{\mt}{\mx}{\mectxt\hw\ms}{\mxpr}{\mectxt\hw\mspr}
    \end{array}
  \end{mathpar}
  \caption{Commuting Conversions}
  \label{fig:comm-conv}
\end{figure}

\begin{lemma}[Commuting Conversions]
\label{lem:comm-conv}
  All of the commuting conversions in \figref{fig:comm-conv} are
  equivalences.
\end{lemma}
\iflong\begin{proof}
  By monadic bind, anti-reduction and the reflexivity
  (\cref{lem:monadic-bind,lem:anti-red,lem:fund-lemma}).
\end{proof}\fi

Next, the following theorem is the main reason we so heavily use
\emph{evaluation contexts}.
It is a kind of open version of the monadic bind lemma
\cref{lem:monadic-bind}.
\begin{lemma}[Evaluation contexts are linear]
\label{lem:evctx-linear}
  If $\menv \vdash \mt : \mA$ and $\menv,\mx:\mA \vdash
  \mectxt\hw{\mx} : \mB$, then
  \[
  \mtlet{\mx}{\mt}{\mectxt\hw{\mx}}
  \equidyn
  \mectxt\hw{\mt}
  \]
\end{lemma}
\iflong\begin{proof}
  By a commuting conversion and an open $\beta$ reduction, connected
  by transitivity \cref{lem:comm-conv,lem:open-beta,lem:trans}
  \begin{align*}
  \mtlet{\mx}{\mt}{\mectxt\hw{\mx}}
  &\equidyn \mectxt\hw{\mtlet{\mx}{\mt}{\mx}}\\
  &\equidyn \mectxt\hw{\mt}
  \end{align*}
\end{proof}\fi

\iflong
As a simple example, consider the following standard equivalence of
let and $\lambda$, which we will need later and prove using the above
lemmas:
\begin{lemma}[Let-$\lambda$ Equivalence]
\label{lem:let-lambda}
  For any $\menv, \mx:\mA \vdash \mt : \mB$ and $\menv \vdash \ms :
  \mA$,
  \[
  (\mtfun{\mx}{\mA}{\mt})\,\ms
  \equidyn
  \mtlet{\mx}{\ms}{\mt}
  \]
\end{lemma}
\begin{proof}
  First, we lift $\mt$ using linearity of evaluation contexts, then an
  open $\beta$-reduction, linked by transitivity:
  \begin{align*}
    (\mtfun{\mx}{\mA}{\mt})\,\ms
    &\equidyn
    \mtletvert{\mx}{\ms}{(\mtfun{\mx}{\mA}{\mt})}
    &\equidyn\mtletvert{\mx}{\ms}{\mt}
  \end{align*}
\end{proof}
\fi

The concepts of pure and terminating terms are useful because when
subterms are pure or terminating, they can be moved around to prove
equivalences more easily.
\begin{definition}[Pure, Terminating Terms]
  \begin{enumerate}
  \item A term $\menv\vdash \mt : \mA$ is \emph{terminating} if for
    any closing $\mgamma$, either $\mt[\mgamma] \bigstepany \merr$ or
    $\mt[\mgamma] \bigstepany \mv$ for some $\mv$.
  \item A term $\menv\vdash \mt : \mA$ is \emph{pure} if for any
    closing $\mgamma$, $\mt[\mgamma] \bigstepany \mv$ for some $\mv$.
  \end{enumerate}
\end{definition}
\iflong
The following terminology and proof are taken from \cite{fuhrmann1999direct}.
\begin{lemma}[Pure Terms are Thunkable]
  For any pure $\menv \vdash \mt : \mA$,
  \[
  \mtlet\mx\mt{\mtfun{\my}{\mB}{\mx}}
  \equidyn
  \mtfun\my\mB\mt
  \]
\end{lemma}
\begin{proof}
  There are two cases $\ltdynr,\gtdynr$.
  \begin{enumerate}
  \item Let $\mgammaone \vltdynr{i}{\menv} \mgammatwo$ and define
    $\mtone = \mt{}[\mgammaone]$ and $\mttwo =
    \mt{}[\mgammatwo]$. Then we know $\mtone \tltdynr i \mA \mttwo$.
    \begin{enumerate}

    \item If $\mtone \bigstep{i+1}$ we're done.

    \item It is impossible that $\mtone \bigstep \merr$ because $\mt$
      is terminating.

    \item If $\mtone\bigstep{j\leq i} \mvone$, then we know that
      $\mtone \bigstepany \mvtwo$ with $\mvone \vltdynr{i-j}{\mA}
      \mvtwo$. Next,
      \[
      \mtlet\mx\mtone{\mtfun{\my}{\mB}{\mx}}
      \bigstep{j}
      \mtfun{\my}{\mB}{\mvone}
      \]

      Then it is sufficient to show $\mtfun{\my}{\mB}
      \vltdynr{i-j}{\mB\mto\mA} \mtfun{\my}{\mttwo}$, i.e. that for
      any $\mvonepr \vltdynr{k\leq(i-j)}{\mB} \mvtwopr$ that
      
      \[ (\mtfun{\my}{\mB}{\mvone})\,\mvonepr \tltdynr{k}{\mA} (\mtfun{\my}{\mB}{\mttwo})\,\mvtwopr \]
      The left side steps
      \[
      (\mtfun{\my}{\mB}{\mvone})\,\mvonepr \bigstep{0} \mvone
      \]
      And the right side steps
      \[
      (\mtfun{\my}{\mB}{\mtone})\,\mvtwopr \bigstep{0} \mtone \bigstep{j} \mvtwo
      \]
      And $\mvone \vltdynr{k}{\mA} \mvtwo$ by assumption above.
    \end{enumerate}
  \item Let $\mgammaone \gtdynr^{i-j}_{\menv} \mgammatwo$ and define
    $\mtone = \mt{}[\mgammaone]$ and $\mttwo =
    \mt{}[\mgammatwo]$. Then we know $\mtone \tgtdynr i \mA \mttwo$.

    Since $\mt$ is terminating we know $\mtone \bigstepany \mvone$ and
    for some $k$, $\mttwo \bigstep{k} \mvtwo$.
    Then we need to show
    And we need to show ${\mtfun\my\mB\mvone}
    \tgtdynr{i}{\mB\mto\mA} \mtfun\my\mB\mttwo$.
    Given any $\mvonepr \vgtdynr{k\leq i} \mvtwopr$, we need to show
    \[
    (\mtfun\my\mB\mvone)\,\mvonepr \tltdynr{k}{\mB} (\mtfun\my\mB\mttwo)\,\mvtwopr
    \]
    The $\beta$ reduction takes $0$ steps, then $\mttwo$ starts
    running.  If $k > i$, there is nothing left to show. Otherwise, $k
    \leq i$ and we know $\mvone \vgtdynr{i-k}{\mA} \mvtwo$ which is
    the needed result.
  \end{enumerate}
\end{proof}
\fi
\begin{lemma}[Pure Terms are Essentially Values]
\label{lem:pure-subst}
  If $\menv \vdash \mt : \mA$ is a pure term, then for any $\menv,\mx
  : \mA \vdash \ms : \mB$, $\mtlet{\mx}{\mt}{\ms} \equidyn \ms[\mt/\mx]$ holds.
\end{lemma}
\iflong
\begin{proof}
  First, since by open $\beta$ we have
  $(\mtfun{\my}{\munitty}\mt\,\mtunit) \equidyn \mt$, by congruence (\cref{lem:cong})
  \[
  \ms[\mt/\mx] \equidyn
  \ms[(\mtfun{\my}{\munitty}\mt\,\mtunit)\,\mtunit/\mx]
  \]
  And by reverse $\beta$ reduction, this is further equivalent to
  \[
  \mtletvert{\mxin{f}}{\mtfun\my\munitty\mt}{
    \ms[(\mxin{f}\,\mtunit)/\mx]
  }
  \]
  By thunkability of $\mt$ and a commuting conversion this is equivalent to:
  \[
  \mtletvert\mx\mt{
  \mtletvert{\mxin{f}}{\mtfun\my\munitty\mx}{
    \ms{}[(\mxin{f}\,\mtunit)/\mx]
  }
  }
  \]
  Which by $\beta$ reduction at each $\mx$ in $\mx$ is:
  \[
  \mtletvert\mx\mt{
  \mtletvert{\mxin{f}}{\mtfun\my\munitty\mx}{
    \ms[\mx]
  }
  }
  \]
  And a final $\beta$ reduction eliminates the auxiliary $\mxin{f}$:
  \[
  \mtletvert\mx\mt{
    \ms[\mx]}
  \]
\end{proof}
\fi

Also, since we consider all type errors to be equal, terminating terms
can be reordered:
\begin{lemma}[Terminating Terms Commute]
\label{lem:term-comm}
  If $\menv \vdash \mt : \mA$ and $\menv\vdash \mtpr : \mApr$ and
  $\menv,\mx:\mA,\my:\mApr \vdash \ms : \mB$, then
  \(
  \mtlet{\mx}{\mt}{
  \mtlet{\mxpr}{\mtpr}{
  \ms}
  }
  \equidyn
  \mtlet{\mxpr}{\mtpr}{
  \mtlet{\mx}{\mt}{
  \ms}
  }
  \)
\end{lemma}
\iflong
\begin{proof}
  By symmetry it is sufficient to prove one direction.
  Let $i \in \mathbb{N}$.
  \begin{enumerate}
  \item Let $\mgammaone \ltdynr{i}{\menv} \mgammatwo$.
    We need to show
    \[
    \mtletvert{\mx}{\mt{}[\mgammaone]}{
    \mtletvert{\mxpr}{\mtpr{}[\mgammaone]}{
      \ms{}[\mgammaone]}
    }
    \tltdynr{i}{\mB}
    \mtletvert{\mxpr}{\mtpr{}[\mgammatwo]}{
    \mtletvert{\mx}{\mt{}\mgammatwo}{
      \ms{}[\mgammatwo]}
    }
    \]
    Note that this is true if the left side diverges or errors, so this
    is true with no conditions on $\mt,\mtpr$

    By \cref{lem:fund-lemma}, we know $\mt[\mgammaone] \tltdynr{i}{\mA}
    \mt[\mgammatwo]$ and $\mtpr{}[\mgammaone] \tltdynr{i}{\mApr}
    \mtpr{}[\mgammatwo]$. We do a joint case analysis on these two
    facts.
    \begin{enumerate}
    \item If $\mt[\mgammaone] \bigstep{i+1}$, done.
    \item If $\mt[\mgammaone] \bigstep{j\leq i} \merr$, done.
    \item If $\mt[\mgammaone] \bigstep{j\leq i} \mvone$, then also
      $\mt[\mgammatwo] \bigstepany \mvtwo$.
      \begin{enumerate}
      \item If $\mtpr{}[\mgammaone] \bigstep{(i-j)+1}$, done.
      \item If $\mtpr{}[\mgammaone] \bigstep{k \leq (i-j)}\merr$, done.
      \item If $\mtpr{}[\mgammaone] \bigstep{k\leq(i-j)} \mvonepr$, then
        $\mtpr{}[\mgammatwo] \bigstepany \mvtwopr$ with $\mvonepr
        \vltdynr{i-(j+k)}{\mApr} \mvtwopr$ and the result follows by
        \cref{lem:fund-lemma} for $\ms$ because we know
        \[ \ms{}[\mgammaone,\mvone/\mx,\mvonepr/\mxpr] \tltdynr{i-(j+k)}{\mApr} \ms{}[\mgammatwo,\mvtwo/\mx,\mvtwopr/\mxpr]\]
      \end{enumerate}
    \end{enumerate}
  \item Let $\mgammaone \vgtdynr{i}{\menv} \mgammatwo$.
    We need to show
    \[
    \mtletvert{\mx}{\mt{}[\mgammaone]}{
    \mtletvert{\mxpr}{\mtpr{}[\mgammaone]}{
      \ms{}[\mgammaone]}
    }
    \tgtdynr{i}{\mB}
    \mtletvert{\mxpr}{\mtpr{}[\mgammatwo]}{
    \mtletvert{\mx}{\mt{}\mgammatwo}{
      \ms{}[\mgammatwo]}
    }
    \]
    By \cref{lem:fund-lemma}, we know $\mt[\mgammaone] \tgtdynr{i}{\mA}
    \mt[\mgammatwo]$ and $\mtpr{}[\mgammaone] \tgtdynr{i}{\mApr}
    \mtpr{}[\mgammatwo]$. We do a joint case analysis on these two
    facts.
    \begin{enumerate}
    \item If $\mtpr{}[\mgammatwo] \bigstep{i+1}$, done.
    \item If $\mtpr{}[\mgammatwo] \bigstep{j\leq i} \merr$ and
      $\mtpr{}[\mgammaone] \bigstepany \merr$. In this case we know
      the right hand side errors, so we must show the left side
      errors. Since $\mt$ is \emph{terminating}, either
      $\mt[\mgammaone] \bigstepany \merr$ (done) or $\mt[\mgammaone]
      \bigstepany \mvone$.
      In the latter case we are also done because:
      \[
      \mtletvert{\mx}{\mt{}[\mgammaone]}{
        \mtletvert{\mxpr}{\mtpr{}[\mgammaone]}{
          \ms{}[\mgammaone]}
      } \bigstepany
      {
        \mtletvert{\mxpr}{\mtpr{}[\mgammaone]}{
          \ms{}[\mgammaone,\mv/\mx]}
      } \bigstepany
      \merr
      \]
    \item If $\mtpr{}[\mgammatwo] \bigstep{j\leq i} \mvtwopr$ then
      either $\mtpr{}[\mgammaone] \bigstep \merr$ or
      $\mtpr{}[\mgammaone] \bigstep \mvonepr \vltdynr{i-j}{\mApr}
      \mvtwopr$.
      Next, consider $\mt[\mgammatwo]$.
      \begin{enumerate}
      \item If $\mt[\mgammatwo] \bigstep{(i-j)+1}$ done.
      \item If $\mt[\mgammatwo] \bigstep{k\leq (i-j)} \merr$, then we
        know also that $\mt[\mgammaone] \bigstepany \merr$, there is
        nothing left to show.
      \item If $\mt[\mgammatwo] \bigstep{k\leq (i-j)} \mvtwo$, then
        either $\mt[\mgammaone] \bigstepany \merr$ or
        $\mt[\mgammaone]\bigstepany \mvone \vltdynr{i-(j+k)}{\mA}
        \mvtwo$. If $\mtpr{}[\mgammaone]$ or $\mt{}[\mgammaone]$
        errors, done, otherwise both return values and the
        result follows by \cref{lem:fund-lemma} for $\ms$.
      \end{enumerate}
    \end{enumerate}
  \end{enumerate}
\end{proof}
\fi

\section{Casts as Embedding-Projection Pairs}
\label{sec:ep-pairs}

In this section, we show how arbitrary casts can be broken down into \emph{embedding-projection} pairs.
First, we define type dynamism and show that casts between less and
more dynamic types form an ep pair.
Then we will show that every cast is a composition of an upcast and a
downcast.

\subsection{Embedding-Projection Pairs}

First, we define ep pairs with respect to \emph{logical}
approximation.
Note that since logical approximation implies observational error
approximation, these are also ep pairs with respect to observational
error approximation.  However, in theorems where we \emph{construct}
new ep pairs from old ones, we will need that the input ep pairs are
logical ep pairs, not just observational, because we have not proven
that logical approximation is \emph{complete} for observational error
approximation.
As with casts, we use evaluation contexts for convenience.
\begin{definition}[EP Pair]
  A (logical) ep pair $(\mEin{e},\mEin{p}) : \mA \eppair \mB$ is a
  pair of an \emph{embedding} $\mEin{e}\hw{\cdot:\mA} : \mB$ and a
  \emph{projection} $\mEin{p}\hw{\cdot:\mB} :\mA$ satisfying
  \begin{enumerate}
  \item Retraction: ${\mx : \mA \vdash \mx \equidyn \mEin{p}\hw{\mEin{e}\hw{\mx}} : \mA}$
  \item Projection: ${\my : \mB \vdash \mEin{e}\hw{\mEin{p}\hw{\my}} \dynr \my : \mB}$
  \end{enumerate}
\end{definition}

Next, we prove that in any embedding-projection pair that embeddings
are pure (always produce a value with no effects) and that projections
are terminating (either error or produce a value).
Paired with the lemmas we have proven about pure and terminating
programs in the previous section, we will be able to prove theorems
about ep pairs more easily.

\begin{lemma}[Embeddings are Pure]
\label{lem:emb-pure}
  If $\mEin{e},\mEin{p} : \mA \eppair \mB$ is an
  embedding-projection pair then $\mx : \mA \vdash
  \mEin{e}\hw{\mx} : \mB$ is pure.
\end{lemma}
\begin{proof}
  The ep pair property states that
  \( \mx : \mA \vDash \mEin{p}\hw{\mEin{e}\hw{\mx}} \equidyn \mx : \mA \)
  Given any value $\cdot \vdash \mv : \mA$, by \lemref{lem:fund-lemma}, we know
  \( \mv \tltdynr{0}{\mA} \mEin{p}\hw{\mEin{e}\hw{\mv}} \)
  and since $\mv \bigstep{0} \mv$, this means there exists $\mvpr$
  such that $\mEin{p}\hw{\mEin{e}\hw{\mv}} \bigstepany \mvpr$,
    and since $\mEin{p}$ is an evaluation context, this means
    there must exist $\mvpr[2]$ with
    $\mEin{p}\hw{\mEin{e}\hw{\mv}} \bigstepany
      \mEin{p}\hw{\mvpr[2]}\bigstepany \mvpr$.
\end{proof}

\begin{lemma}[Projections are Terminating]
\label{prj-term}
  If $\mEin{e},\mEin{p} : \mA \eppair \mB$ is an
  embedding-projection pair then $\my : \mB \vdash
  \mEin{p}\hw{\my} : \mA$ is terminating.
\end{lemma}
\begin{proof}
  The ep pair property states that
  \( \my : \mB \vDash \mEin{e}\hw{\mEin{p}\hw{\my}} \ltdyn \my : \mB \)
  Given any $\mv : \mB$, by \lemref{lem:fund-lemma}, we know
  \( \mEin{e}\hw{\mEin{p}\hw{\mv}} \tgtdynr{0}{\mB} \mv \)
  so therefore either
  $\mEin{e}\hw{\mEin{p}\hw{\mv}}\bigstepany \merr$, which
  because $\mEin{e}$ is pure means ${\mEin{p}\hw{\mv}}
  \bigstepany \merr$, or $\mEin{e}\hw{\mEin{p}\hw{\mv}}
  \bigstepany \mvpr$ which by strictness of evaluation contexts means
  ${\mEin{p}\hw{\mv}} \bigstepany\mvpr[2]$ for some
  $\mvpr[2]$.
\end{proof}

Crucially, ep pairs can be constructed using simple function
composition.
First, the identity function is an ep pair by reflexivity.
\begin{lemma}[Identity EP Pair]
  For any type $\mA$, $\mhole,\mhole : \mA \eppair \mA$.
\end{lemma}

Second, if we compose the embeddings one way and projections the
opposite way, the result is an ep pair, by congruence.
\begin{lemma}[Composition of EP Pairs]
  For any ep pairs $\mE_{e,1},\mE_{p,1} : \mAone \eppair \mAtwo$ and
  $\mE_{e,2},\mE_{p,2} : \mAtwo \eppair \mAin{3}$,
  $\mE_{e,2}\hw{\mE_{e,1}},\mE_{p,1}\hw{\mE_{p,2}} : \mAone \eppair
  \mAin{3}$.
\end{lemma}

\subsection{Type Dynamism}
\newcommand{\bifdynrrule}[2]{\inferrule*[right={#1}]{\sAone
    \dynr \sAtwo \and \sBone \dynr \sBtwo}{#2{\sAone}{\sBone} \dynr
    #2{\sAtwo}{\sBtwo}}}
\begin{figure}
\flushleft{\fbox{\small{$\sA \ltdyn \sB$}}}

\vspace{-2ex}
\begin{mathpar}
    \inferrule*[right=Reflexivity]
    {~}
    {{\sty} \ltdyn {\sty}}\and
    \inferrule*[right=Transitivity]
    {{\sAone}\dynr{\sAtwo}\and
     {\sAtwo}\dynr{\sAthree}}
    {{\sAone} \dynr{\sAthree}}\and
    \inferrule*[right=Dyn-Top]
    {~}
    {\sA \dynr {\sdynty}}\and
    \bifdynrrule{Sum}{\ssumty}\and
    \bifdynrrule{Prod}{\spairty}\and
    \bifdynrrule{Fun}{\sfunty}\and
    \end{mathpar}  
  \caption{Type Dynamism}
  \label{fig:type-dynamism}
\end{figure}

Next, we consider type dynamism and its relationship to the casts.
The type dynamism relation is presented in
\figref{fig:type-dynamism}.
The relation $\sA \ltdyn \sB$ reads as ``$\sA$ is less dynamic than
$\sB$'' or ``$\sA$ is more precise than $\sB$''.
For the purposes of its definition, we can say that it is the least
reflexive and transitive relation such that every type constructor is
monotone and $\sdynty$ is the greatest element.
Even the function type is monotone in both input and output argument
and for this reason, type dynamism is sometimes called \emph{na\"ive
  subtyping} \cite{wadler-findler09}.
However, this gives us no \emph{semantic} intuition about what it
could possibly mean.
We propose that $\sA \ltdyn \sB$ should hold when the casts between
$\sA$ and $\sB$ form an embedding-projection pair
$\mE_{\obcast\sA\sB},\mE_{\obcast\sB\sA} : \sA \eppair \sB$.
We can then view each of the cases of the gradual guarantee as being
\emph{compositional rules} for constructing ep pairs.
Reflexivity and transitivity correspond to the identity and
composition of ep pairs, and the monotonicity of types comes from the
fact that \emph{every} functor preserves ep pairs.

Taking this idea further, we can view type dynamism not just as an
\emph{analysis} of pre-existing gradual type casts, but by considering
its simple \emph{proof theory}, we can view proofs of type dynamism
as \emph{synthesizing} the definitions of casts.
To accomplish this, we give a \emph{refined} formulation of the proof
theory of type dynamism in \figref{fig:type-dynamism-proofs}, which
includes explicit proof terms $c : \sA \ltdyn \sB$.
The methodology behind the presentation is to make reflexivity,
transitivity, and the fact that $\dyn$ is a greatest element into
\emph{admissible} properties of the system, rather than primitive
rules.
First, by making proofs admissible, we see in detail how bigger casts
are built up from small pieces.

Second, this formulation satisfies a \emph{canonicity} property: there
is exactly one proof of any given derivation, which simplifies the
definition of the semantics.
By giving a presentation where derivations are canonical, the typical
``coherence'' theorem, that says any two derivations have equivalent
semantics, becomes trivial.
An alternative formulation would define an ep-pair semantics where
reflexivity and transitivity denote identity and composition of ep
pairs, and then prove that any two derivations have equivalent
semantics.
Instead, we define admissible constructions for reflexivity and
transitivity, and then prove a \emph{decomposition lemma}
(\cref{lem:decomposition}) that states that the ep-pair semantics
interprets our admissible reflexivity derivation as identity and
transitivity derivation as a composition.
In short, our presentation makes it obvious that the semantics is
coherent, but not that it is built out of composition, whereas the
alternative makes it obvious that the semantics is built out of
composition, but not that it is coherent.

\begin{figure}
\flushleft{\fbox{\small{$c : \sA \ltdyn \sB$}}}

\vspace{-4ex}
  \begin{mathpar}
    \inferrule
    {\sty \in \{\sunitty,\sdynty \}}
    {id(\sA) : {\sty} \dynr {\sty}}\and
    \inferrule
    {\sA \neq \sdynty \and c : \sty \dynr \floor\sA\\
      \inferrule{~}{tag({\floor\sA}) : \floor\sA \ltdyn \sdynty}}
    {tag({\floor\sA}) \circ c : \sty \dynr {\sdynty}}\\
    \inferrule
    {c : \sAone \ltdyn \sAtwo\and
      d : \sBone \ltdyn \sBtwo}
    {c \times d : \sAone \stimes \sBone \ltdyn \sAtwo \stimes \sBtwo}\and
    \inferrule
    {c : \sAone \ltdyn \sAtwo\and
      d : \sBone \ltdyn \sBtwo}
    {c \plus d : \sAone \splus \sBone \ltdyn \sAtwo \splus \sBtwo}\and
    \inferrule
    {c : \sAone \ltdyn \sAtwo\and
      d : \sBone \ltdyn \sBtwo}
    {c \to d : \sAone \sto \sBone \ltdyn \sAtwo \sto \sBtwo}
  \end{mathpar}
  \caption{Canonical Proof Terms for Type Dynamism}
  \label{fig:type-dynamism-proofs}
\end{figure}

We present the proof terms for type dynamism in
\figref{fig:type-dynamism-proofs}.
As in presentations of sequent calculus, we include the identity ep
pair (reflexivity) only for the base types $\sunitty,\sdynty$.
The next rule $tag(\floor{\sA}) \circ c : \sA \ltdyn \sdynty$ states
that any casts between a non-dynamic type $\sA$ and the dynamic type
$\sdynty$ are the composition $\circ$ of, first, a tagging-untagging
ep-pair with its underlying tag type $tag(\floor{\sA}) : \floor{\sA}
\ltdyn \sdynty$ and an ep pair from $\sA$ to its tag type $c : \sA
\ltdyn \floor{\sA}$.
The product, sum, and function rules are written to evoke that their ep
pairs use the functorial action.

As mentioned the proof terms are \emph{canonical}, meaning there is at
most one derivation of any $\sA \ltdyn \sB$.
\begin{lemma}[Canonical Type Dynamism Derivations]
  Any two derivations $c, d : \sA \ltdyn \sB$ are equal $c = d$.
\end{lemma}
\iflong
\begin{proof}
  By induction on $c$, noting in each case that exactly one case is
  possible.
\end{proof}
\fi

\begin{figure}
  \begin{minipage}{0.45\textwidth}
    \begin{mathpar}
      \begin{array}{rcl}
      id(\sdynty) &\defeq& id(\sdynty)\\
      id(\sunitty) &\defeq& id(\sunitty)\\
      id({\sAone \stimes \sAtwo}) &\defeq & id(\sAone) \times id(\sAtwo)\\
      id({\sAone \splus \sAtwo}) &\defeq & id(\sAone) \plus id(\sAtwo)\\ 
      id({\sAone \sto \sAtwo}) &\defeq & id(\sAone) \to id(\sAtwo)
      \end{array}
    \end{mathpar}
  \end{minipage}
  \begin{minipage}{0.45\textwidth}
    \begin{mathpar}
      \begin{array}{rcl}
      (tag(\floor\sA) \circ c) \circ d &\defeq& tag(\floor\sA) \circ (c \circ d)\\
      (id(\sA)) \circ d &\defeq & d\\
      (c \times d) \circ (c' \times d') & \defeq & (c \circ c') \times (d \circ d')\\
      (c \plus d) \circ (c' \plus d') & \defeq & (c \circ c') \plus (d \circ d')\\
      (c \to d) \circ (c' \to d') & \defeq & (c \circ c') \to (d \circ d')
      \end{array}
    \end{mathpar}
  \end{minipage}
  \begin{mathpar}
    \begin{array}{rcl}
      top(\sdynty) & \defeq & id(\sdynty)\\
      top(\sunitty) & \defeq & tag(\sunitty) \circ id(\sunitty)\\
      top(\sA \stimes \sB) & \defeq& tag(\sdynty\stimes\sdynty) \circ (top(\sA) \times top(\sB))\\
      top(\sA \splus \sB) & \defeq& tag(\sdynty\splus\sdynty) \circ (top(\sA) \plus top(\sB))\\
      top(\sA \sto \sB) & \defeq& tag(\sdynty\sto\sdynty) \circ (top(\sA) \to top(\sB))\\
    \end{array}
  \end{mathpar}
  \caption{Type Dynamism Admissible Proof Terms}
  \label{fig:type-dyn-admissible}
\end{figure}

Next, we need to show that the rules in \figref{fig:type-dynamism} are all
\emph{admissible} in the refined system \figref{fig:type-dynamism-proofs}.
The proof of admissibility is given in
\figref{fig:type-dyn-admissible}.
First, to show reflexivity is
admissible, we construct the proof $id(\sA) : \sA \ltdyn \sA$.
It is primitive for $\sdynty$ and $\sunitty$ and we use the congruence
rule to lift the others.
Second, to show transitivity is admissible, for every $d : \sAone
\ltdyn \sAtwo$ and $c : \sAtwo \ltdyn \sAin{3}$, we construct their
\emph{composite} $c \circ d : \sAone \ltdyn \sAin{3}$ by recursion on
$c$.
If $c$ is a primitive composite with a tag, we use associativity of
composition to push the composite in.
If $c$ is the identity, the composite is just $d$.
Otherwise, both $c$ and $d$ must be between a connective, and we push
the compositions in.
Finally, we show that $\sdynty$ is the most dynamic type by
constructing a derivation $top(\sA) : \sA \ltdyn \sdynty$ for every
$\sA$.
For $\sdynty$, it is just the identity; for the remaining types, we use the tag ep
pair and compose with lifted uses of $top$.

\begin{figure}
  \begin{minipage}{0.45\textwidth}
    \begin{mathpar}
      \begin{array}{rcl}
      m \in \{e,p\}\\
      \overline e & \defeq & p\\
      \overline p & \defeq & e\\
      \mE_{m,id(\sA)} & \defeq & \mhole\\
      \mE_{e,tag(\stagty)} & \defeq & \mtroll{\sem\sdynty}{\mtsum{\stagty}{\mhole}}\\
      \mE_{p,tag(\stagty)} & \defeq & \mtelsecase{\mtunroll\mhole}{{\stagty}}{\mx}{\mx}{\merr}\\
      \end{array}
    \end{mathpar}
  \end{minipage}
  \begin{minipage}{0.45\textwidth}
    \begin{mathpar}
      \begin{array}{rcl}
      \mE_{e,tag(\stagty) \circ d} & \defeq & \mE_{e,tag(\stagty)}\hw{\mE_{e,d}}\\
      \mE_{p,tag(\stagty) \circ d} & \defeq & \mE_{p,d}\hw{\mE_{p,tag(\stagty)}}\\
      \mE_{m,c \times c'} & \defeq & \mE_{m,c} \mtimes \mE_{m,c'}\\
      \mE_{m,c \plus c'} & \defeq & \mE_{m,c} \mplus \mE_{m,c'}\\
      \mE_{m,c \to c'} & \defeq & \mE_{\overline m,c} \mto \mE_{m,c'}
      \end{array}
    \end{mathpar}
  \end{minipage}
  \caption{Type Dynamism Cast Translation}
  \label{fig:type-dyn-cast-translation}
\end{figure}

Next, we construct a \emph{semantics} for the type dynamism proofs
that justifies the intuition we have given so far; it is presented in
\figref{fig:type-dyn-cast-translation}.
Every type dynamism proof $c : \sA \ltdyn \sB$ defines a \emph{pair}
of an embedding $\mectxt_{e,c}$ and a projection $\mectxt_{p,c}$.
Since many rules are the same for embeddings and projections, we use
$m \in \{e,p\}$ to abstract over the \emph{mode} of the cases.
We define the \emph{complement} of a mode $\overline m$ to swap between
embeddings and projections; it is used in the function case. 
The primitive identity casts are interpreted as the identity, and the
primitive composition of casts $tag(\stagty) \circ d$ is interpreted
as the composition of ep pairs of $tag(\stagty)$ and $d$.
The tag type derivation is interpreted by the same definition as the
cast in \figref{fig:direct-cast-translation}: tagging puts the correct
sum case and $\mtroll{\sdynty}{}$, and untagging unwraps if its the
correct sum case and otherwise errors.
We abbreviate this as pattern matching with an ``else'' clause, where
the else clause stands for all of the clauses that do not match the
tag type $\stagty$.
The desugaring to repeated case statements on sums should be clear.
The product and sum type are just given by their functorial action
with the same mode.
The function type similarly uses its functorial action, but swaps from
embedding to projection or vice-versa on the domain side.
This shows that there is nothing strange about the function rule: it is
the same construction as for subtyping, but constructing arrows back
and forth at the same time.
The fact that contravariant functors are covariant with respect to ep
pairs in this way is \emph{precisely} the reason they are used
extensively in domain theory.

We next verify that these actually are embedding-projection pairs.
To do this, we use the identity and composition lemmas proved before,
but we also need to use \emph{functoriality} of the actions of type
constructors, meaning that the action of the type interacts well with
identity and composition of evaluation contexts.

\newcommand{\qandq}{\quad\text{and}\quad}
\begin{lemma}[Identity Extension]
\label{lem:id-ext}
  \[
  \mhole \mtimes \mhole \equidyn \mhole \qandq
  \mhole \mplus \mhole \equidyn \mhole\qandq
  \mhole \mto \mhole \equidyn \mhole
  \]
\end{lemma}
\iflong
\begin{proof}
  All are instances of $\eta$ expansion.
\end{proof}
\fi

In a call-by-value language, the functoriality rules do not hold in
general for the product functor, but they do for terminating programs
because their order of evaluation is irrelevant.
Also notice that when composing using the functorial action of the
function type $\mto$, the composition flips on the domain side,
because the function type is \emph{contravariant} in its domain.
\begin{lemma}[Functoriality for Terminating Programs]
\label{lem:functor}
  The following equivalences are true for any well-typed,
  \emph{terminating} evaluation contexts.
  \begin{mathpar}
    \begin{array}{rcl}
      (\mectxttwo \mtimes \mectxttwopr)\hw{\mectxtone \mtimes \mectxtonepr} & \equidyn &
      (\mectxttwo\hw\mectxtone) \mtimes (\mectxttwopr\hw\mectxtonepr)\\
      (\mectxttwo \mplus \mectxttwopr)\hw{\mectxtone \mplus \mectxtonepr} & \equidyn &
      (\mectxttwo\hw\mectxtone) \mplus (\mectxttwopr\hw\mectxtonepr)\\
      (\mectxttwo \mto \mectxttwopr)\hw{\mectxtone \mto \mectxtonepr} & \equidyn &
      (\mectxtone\hw\mectxttwo) \mto (\mectxttwopr\hw\mectxtonepr)
    \end{array}
  \end{mathpar}
\end{lemma}
\iflong
\begin{proof}
  \begin{enumerate}
  \item ($\mto$) We need to show (after a commuting conversion)
    \[
    \mtletvert{\mx_f}{\mhole}{
    \mtletvert{\my_f}{\mtufun{\mx_a}{\mectxtonepr\hw{\mx_f\,(\mectxtone\hw{\mx_a})}}}{
    {\mtufun{\my_a}{\mectxttwopr\hw{\my_f\,(\mectxttwo\hw{\my_a})}}}}}
    \equidyn
    \mtletvert{\mx_f}{\mhole}{
      \mtufun{\my_a}{\mectxttwopr\hw{\mectxtonepr\hw{\mx_f\,(\mectxtone\hw{\mectxttwo\hw{\my_a}})}}}}
    \]
    First, we substitute for $\my_f$ and then lift the argument
    $\mectxttwo\hw{\my_a}$ out and $\beta$ reduce:
    \begin{align*}
      \mtletvert{\mx_f}{\mhole}{
      \mtletvert{\my_f}{\mtufun{\mx_a}{\mectxtonepr\hw{\mx_f\,(\mectxtone\hw{\mx_a})}}}{
        {\mtufun{\my_a}{\mectxttwopr\hw{\my_f\,(\mectxttwo\hw{\my_a})}}}}}
      &\equidyn
      \mtletvert{\mx_f}{\mhole}{\mtufun{\my_a}{\mectxttwopr\hw{({\mtufun{\mx_a}{\mectxtonepr\hw{\mx_f\,(\mectxtone\hw{\mx_a})}}})\,(\mectxttwo\hw{\my_a})}}}\\
      &\equidyn
      \mtletvert{\mx_f}{\mhole}{
      {\mtufun{\my_a}{\mtletvert{\mx_a}{\mectxttwo\hw{\my_a}}{\mectxttwopr\hw{({\mtufun{\mx_a}{\mectxtonepr\hw{\mx_f\,(\mectxtone\hw{\mx_a})}}})\,\mx_a}}}}
      }\\
      &\equidyn
      \mtletvert{\mx_f}{\mhole}{
      {\mtufun{\my_a}{\mtletvert{\mx_a}{\mectxttwo\hw{\my_a}}{\mectxttwopr\hw{{{\mectxtonepr\hw{\mx_f\,(\mectxtone\hw{\mx_a})}}}}}}}
      }\\
      &\equidyn
      \mtletvert{\mx_f}{\mhole}{
      {\mtufun{\my_a}{{\mectxttwopr\hw{{{\mectxtonepr\hw{\mx_f\,(\mectxtone\hw{\mectxttwo\hw{\my_a}})}}}}}}}}
    \end{align*}
  \item ($\mplus$) We need to show
    \[
    \mtcasevert{\mtcasevert{\mx}{\my}{\mtinj{\mectxtone\hw{\my}}}{\mypr}{\mtinjpr{\mectxtonepr\hw{\mypr}}}}{\my}{\mtinj{\mectxttwo\hw{{\my}}}}{\mypr}{\mtinjpr{\mectxttwo\hw{{\mypr}}}}
    \equidyn
    \mtcasevert{\mx}{\my}{\mtinj{\mectxttwo\hw{\mectxtone\hw{\my}}}}{\mypr}{\mtinjpr{\mectxttwo\hw{\mectxtonepr\hw{\mypr}}}}
    \]
    First, we do a case-of-case commuting conversion, then lift the
    discriminees out, $\beta$ reduce and restore them.
    \begin{align*}
      \mtcasevert{\mtcasevert{\mx}{\my}{\mtinj{\mectxtone\hw{\my}}}{\mypr}{\mtinjpr{\mectxtonepr\hw{\mypr}}}}{\my}{\mtinj{\mectxttwo\hw{{\my}}}}{\mypr}{\mtinjpr{\mectxttwo\hw{{\mypr}}}}
      &\equidyn
      {\mtcasevert{\mx}{\my}{\mtcasevert{\mtinj{\mectxtone\hw{\my}}}{\my}{\mtinj{\mectxttwo\hw{{\my}}}}{\mypr}{\mtinjpr{\mectxttwo\hw{{\mypr}}}}}{\mypr}{\mtcasevert{\mtinjpr{\mectxtonepr\hw{\mypr}}}{\my}{\mtinj{\mectxttwo\hw{{\my}}}}{\mypr}{\mtinjpr{\mectxttwo\hw{{\mypr}}}}}}\\
      &\equidyn
      {\mtcasevert{\mx}{\my}{\mtletvert{\my}{{{\mectxtone\hw{\my}}}}{\mtcasevert{\mtinj\my}{\my}{\mtinj{\mectxttwo\hw{{\my}}}}{\mypr}{\mtinjpr{\mectxttwo\hw{{\mypr}}}}}}
        {\mypr}{\mtletvert{\mypr}{{\mectxtonepr\hw{\mypr}}}{\mtcasevert{\mtinjpr\mypr}{\my}{\mtinj{\mectxttwo\hw{{\my}}}}{\mypr}{\mtinjpr{\mectxttwo\hw{{\mypr}}}}}}}\\
      &\equidyn
      {\mtcasevert{\mx}{\my}{\mtletvert{\my}{{{\mectxtone\hw{\my}}}}{{\mtinj{\mectxttwo\hw{{\my}}}}}}
        {\mypr}{\mtletvert{\mypr}{{\mectxtonepr\hw{\mypr}}}{{\mtinjpr{\mectxttwo\hw{{\mypr}}}}}}}\\
      &\equidyn
      {\mtcasevert{\mx}{\my}{{{\mtinj{\mectxttwo\hw{{\mectxtone\hw{\my}}}}}}}
        {\mypr}{{{\mtinjpr{\mectxttwo\hw{{\mectxtonepr\hw{\mypr}}}}}}}}\\
    \end{align*}
  \item ($\mtimes$) We need to show
    \[
    \mtmatchpairvert{\mx}{\mxpr}{\mx_p}{
    \mtmatchpairvert{\my}{\mypr}{\mtpair{\mectxtone\hw{\mx}}{\mectxtonepr\hw{\mxpr}}}
    {\mtpair{\mectxttwo\hw{\my}}{\mectxttwopr\hw{\mypr}}}}
    \equidyn
    \mtmatchpairvert{\mx}{\mxpr}{\mx_p}{
      \mtpair{\mectxttwo\hw{\mectxtone\hw{\mx}}}{\mectxttwopr{\mectxtonepr\hw{\mxpr}}}}
    \]
    First, we make the evaluation order explicit, then re-order using
    the fact that terminating programs commute \cref{lem:term-comm}.

    \begin{align*}
      \mtmatchpairvert{\mx}{\mxpr}{\mx_p}{
      \mtmatchpairvert{\my}{\mypr}{\mtpair{\mectxtone\hw{\mx}}{\mectxtonepr\hw{\mxpr}}}
      {\mtpair{\mectxttwo\hw{\my}}{\mectxttwopr\hw{\mypr}}}}
      &\equidyn
      \mtmatchpairvert{\mx}{\mxpr}{\mx_p}{
      \mtletvert{\mxone}{{\mectxtone\hw{\mx}}}{
      \mtletvert{\mxonepr}{\mectxtonepr\hw{\mxpr}}{
      \mtmatchpairvert{\my}{\mypr}{\mtpair{\mxone}{\mxonepr}}{
      \mtletvert{\mxtwo}{{\mectxttwo\hw{\my}}}{
      \mtletvert{\mxtwopr}{{\mectxttwopr\hw{\mypr}}}{
        \mtpair{\mxtwo}{\mxtwopr}
      }}}}}}\\
      &\equidyn
      \mtmatchpairvert{\mx}{\mxpr}{\mx_p}{
      \mtletvert{\mxone}{{\mectxtone\hw{\mx}}}{
      \mtletvert{\mxonepr}{\mectxtonepr\hw{\mxpr}}{
      \mtletvert{\mxtwo}{{\mectxttwo\hw{\mxone}}}{
      \mtletvert{\mxtwopr}{{\mectxttwopr\hw{\mxonepr}}}{
        \mtpair{\mxtwo}{\mxtwopr}
      }}}}}\\
      &\equidyn
      \mtmatchpairvert{\mx}{\mxpr}{\mx_p}{
      \mtletvert{\mxone}{{\mectxtone\hw{\mx}}}{
      \mtletvert{\mxtwo}{{\mectxttwo\hw{\mxone}}}{
      \mtletvert{\mxonepr}{\mectxtonepr\hw{\mxpr}}{
      \mtletvert{\mxtwopr}{{\mectxttwopr\hw{\mxonepr}}}{
        \mtpair{\mxtwo}{\mxtwopr}
      }}}}}\\
      &\equidyn
      \mtmatchpairvert{\mx}{\mxpr}{\mx_p}{
      \mtpair{\mectxttwo\hw{\mectxtone\hw{\mx}}}{{{\mectxttwopr\hw{\mectxtonepr\hw{\mxpr}}}}}}
    \end{align*}
  \end{enumerate}
\end{proof}
\fi

With these cases covered, we can show the casts given by type dynamism
really are ep pairs.
\begin{lemma}[Type Dynamism Derivation denotes EP Pair]
\label{lem:dyn-der-ep}
For any derivation $c : \sA \ltdyn \sB$, then $\mE_{e,c},
\mE_{p,c} : \sem{\sA} \eppair \sem{\sB}$ are an ep pair.
\end{lemma}
\iflong
\begin{proof}
  By induction on the derivation $c$.
  \begin{enumerate}
  \item (Identity) ${\mE_{e,id(\sA)},\mE_{p,id(\sA)} :
    \sem\sA \eppair \sem\sA}$. This case is immediate by
    \cref{lem:fund-lemma}.
  \item (Composition) 
     $\inferrule{\mE_{e,c},\mE_{p,c} :
    \sem\sAone \eppair \sem\sAtwo \and \mE_{e,c'},\mE_{p,c'} :
    \sem\sAtwo \eppair
    \sem\sAthree}{\mE_{e,c'}\hw{\mE_{e,c}},\mE_{p,c}\hw{\mE_{p,c'}}
    : \sem\sAone\eppair \sem\sAthree}$. We need to show the retract property:
    \[ \mx:\sem\sAone \vDash \mE_{p,c}\hw{\mE_{p,c'}\hw{\mE_{e,c'}\hw{\mE_{e,c}\hw{\mx}}}} \equidyn \mx : \sem\sAone \]
    and the projection property:
    \[ \my:\sem\sAthree \vDash
      \mE_{e,c'}\hw{\mE_{e,c}\hw{\mE_{p,c}\hw{\mE_{p,c'}\hw{\my}}}}
      \ltdyn \my : \sem\sAthree \]
    Both follow by congruence and the inductive hypothesis, we show the projection property:
    \begin{align*}
      \mE_{e,c'}\hw{\mE_{e,c}\hw{\mE_{p,c}\hw{\mE_{p,c'}\hw{\my}}}}
      & \ltdyn \mE_{e,c'}\hw{\mE_{p,c'}\hw{\my}}\tag{inductive hyp, cong \ref{lem:cong}}\\
      & \ltdyn \my \tag{inductive hyp}
    \end{align*}
  \item (Tag) $\inferrule{~}{\mtroll{\sem\sdynty}{\mtsum{\stagty}{\mhole}},\mtelsecasevert{\mtunroll{\mhole}}{\sem{\stagty}}{\mx}{\mx}{\merr} : \sem\sA \eppair \sem\sB}$.\\
    The retraction case follows by $\beta$ reduction
    \begin{align*}
      \mtelsecasevert{\mtunroll{\mtroll{\sem\sdynty}{\mtsum{\stagty}{\mx}}}}
                     {\mtsum{\stagty}{\mxin{\stagty}}}{\mxin{\stagty}}{\merr}
      & \equidyn \mx
    \end{align*}
    For the projection case, we need to show
    \[
    \mtroll{\sem\sdynty}{\mtsum{\stagty}{\left(\mtelsecasevert{\mtunroll{\my}}
                     {\mtsum{\stagty}{\mxin{\stagty}}}{\mxin{\stagty}}{\merr}
\right)}} \ltdyn \my
    \]
    First, on the left side, we do a commuting conversion
    (\cref{lem:comm-conv}) and then use linearity of evaluation
    contexts to reduce the cases to error:
    \begin{align*}
      \mtroll{\sem\sdynty}{\mtsum{\stagty}{\left(\mtelsecasevert{\mtunroll{\my}}
                     {\mtsum{\stagty}{\mxin{\stagty}}}{\mxin{\stagty}}{\merr}
\right)}} & \equidyn
      \mtelsecasevert{\mtunroll{\my}}
                     {\mtsum{\stagty}{\mxin{\stagty}}}{\mtroll{\sem\sdynty}{\mtsum{\stagty}{\mxin{\stagty}}}}{\mtroll{\sem\sdynty}{\mtsum{\stagty}{\merr}}}\\
          & \equidyn                      
      \mtelsecasevert{\mtunroll{\my}}
                     {\mtsum{\stagty}{\mxin{\stagty}}}
                     {\mtroll{\sem\sdynty}{\mtsum{\stagty}{\mxin{\stagty}}}}
                     {\merr}\\
    \end{align*}
    Next, we $\eta$-expand the right hand side
    \begin{align*}
      \my & \equidyn \mtroll{\sem\sdynty}{\mtunroll{\my}}\\
      & \equidyn \etatagcases{\mtunroll{\my}}{\mtroll{\sem\sdynty}}
    \end{align*}
    The result follows by congruence because $\merr \ltdyn \mt$ for any $\mt$.
  \item (Functions) $\inferrule{\mE_{e,c},\mE_{p,c} : \sAone
    \eppair \sAtwo\and\mE_{e,c'},\mE_{p,c'} : \sBone \eppair
    \sBtwo}{\mE_{p,c} \mto \mE_{e,c'},\mE_{e,c} \mto
    \mE_{p,c'} : \sem\sAone \mto \sem\sBone \eppair \sem\sAtwo
    \mto \sem\sBtwo}$ We prove the projection property, the retraction
    proof is similar. We want to show
    \[
    \my : \sem\sAtwo \mto \sem\sBtwo\vDash (\mE_{e,c} \mto \mE_{p,c'})\hw{(\mE_{p,c} \mto \mE_{e,c'})\hw{\my}} \ltdyn \my : \sem\sAtwo \mto \sem\sBtwo
    \]
    Since embeddings and projections are terminating, we can apply
    functoriality \lemref{lem:functor} to show the left hand side is
    equivalent to
    \[ ((\mE_{p,c}\hw{\mE_{e,c}}) \mto (\mE_{p,c'}\hw{\mE_{e,c'}}))\hw{\my}\]
    which by congruence and inductive hypothesis is $\ltdyn$:
    \[ (\mhole \mto \mhole)\hw{\my} \]
    which by identity extension \cref{lem:id-ext} is equivalent to $\my$.
  \item (Products) By the same argument as the function case.    
  \item (Sums) By the same argument as the function case.
  \end{enumerate}
\end{proof}
\fi

Next, while we showed that transitivity and reflexivity were
admissible with the $id(\sA)$ and $c \circ d$ definitions, their
semantics are not given directly by the identity and composition of
evaluation contexts.
We justify this notation by the following theorems.
First, $id(\sA)$ is the identity by identity extension.
\begin{lemma}[Reflexivity Proofs denote Identity]
\label{lem:id-casts}
  For every $\sA$, $\mE_{e,id(\sA)} \equidyn \mhole$ and $\mE_{p,id(\sA)} \equidyn \mhole$.
\end{lemma}
\iflong
\begin{proof}
  By induction on $\sA$, using the identity extension lemma.
\end{proof}
\fi

Second, we have our key \emph{decomposition} theorem.
While the composition theorem says that the composition of any two ep
pairs is an ep pair, the \emph{decomposition} theorem is really a
theorem about the \emph{coherence} of our type dynamism proofs.
It says that given any ep pair given by $c : \sAin{1} \ltdyn
\sAin{3}$, if we can find a middle type $\sAin{2}$, then we can
\emph{decompose} $c$'s ep pairs into a composition.
This theorem is used extensively, especially in the proof of the
gradual guarantee.
\begin{lemma}[Decomposition of Upcasts, Downcasts]
\label{lem:decomposition}
  For any derivations $c : \sAone \ltdyn \sAtwo$ and $c' : \sAtwo \ltdyn \sAthree$,
  the upcasts and downcasts given by their composition $c' \circ c$
  are equivalent to the composition of their casts given by $c,c'$:
  \begin{align*}
    \mx : \sem\sAone \vDash \mE_{e, c' \circ c}\hw{\mx} &\equidyn \mE_{e,c'}\hw{\mE_{e,c}\hw{\mx}} : \sem\sAthree\\
    \my : \sem\sAthree \vDash \mE_{p, c' \circ c}\hw{\my} &\equidyn \mE_{p,c}\hw{\mE_{p,c'}\hw{\my}} : \sem\sAone
  \end{align*}
\end{lemma}
\iflong
\begin{proof}
  By induction on the pair $c,c'$, following the recursive definition
  of $c \circ c'$.
  \begin{enumerate}
    \item $(tag(\stagty) \circ c) \circ d \defeq tag(\stagty) \circ (c
      \circ d)$. By inductive hypothesis and strict associativity of
      composition of evaluation contexts.
    \item $id(\sAone) \circ d \defeq d$ reflexivity.
    \item $(c \times d) \circ (c' \times d')  \defeq  (c \circ c') \times (d \circ d')$
      By inductive hypothesis and functoriality \cref{lem:functor}.
    \item $(c \plus d) \circ (c' \plus d')  \defeq  (c \circ c') \plus (d \circ d')$
      By inductive hypothesis and functoriality \cref{lem:functor}.
    \item $(c \to d) \circ (c' \to d')  \defeq  (c \circ c') \to (d \circ d')$
      By inductive hypothesis and functoriality \cref{lem:functor}.
  \end{enumerate}
\end{proof}
\fi

Finally, now that we have established the meaning of type dynamism
derivations and proven the decomposition theorem, we can
dispense with direct manipulation of derivations.
So we define the following notation for ep pairs that just uses the
types:
\begin{definition}[EP Pair Semantics]
  Given $c : \sA \ltdyn \sB$, we define $\mE_{m,\sA,\sB} = \mE_{m,c}$.
\end{definition}

\subsection{Casts Factorize into EP Pairs}

Next, we show how the upcasts and downcasts are sufficient to
construct all the casts of $\glangname$.

First, when $\sA \ltdyn \sB$, the ep pair semantics and the cast
semantics coincide:
\begin{lemma}[Upcasts and Downcasts are Casts]
\label{lem:ud-are-casts}
  If $\sA \ltdyn \sB$ then $\mE_{\obcast{\sA}{\sB}} \equidyn
  \mE_{e,\sA,\sB}$ and $\mE_{\obcast{\sB}{\sA}} \equidyn
  \mE_{p,\sA,\sB}$.
\end{lemma}
\iflong
\begin{proof}
By induction following the recursive definition of $\mE_{\obcast\sA\sB}$
  \begin{enumerate}
    \item $\mectxt_{\obcast{\sdynty}{\sdynty}}  \defeq  \mhole$ By reflexivity.
    \item $\mectxt_{\obcast{\sAone\splus\sBone}{\sAtwo\stimes\sBtwo}}  \defeq  \mectxt_{\obcast{\sAone}{\sAtwo}} \mplus \mectxt_{\obcast{\sBone}{\sBtwo}}$ By inductive hypothesis and congruence.
    \item $\mectxt_{\obcast{\sAone\stimes\sBone}{\sAtwo\stimes\sBtwo}} \defeq \mectxt_{\obcast{\sAone}{\sAtwo}} \mtimes \mectxt_{\obcast{\sBone}{\sBtwo}}$ By inductive hypothesis and congruence.
    \item $\mectxt_{\obcast{\sAone \sto \sBone}{\sAtwo\sto\sBtwo}}  \defeq  \mectxt_{\obcast{\sAtwo}{\sAone}} \mto \mectxt_{\obcast{\sBone}{\sBtwo}}$ By inductive hypothesis and congruence.
    \item $\mectxt_{\obcast{\stagty}{\sdynty}}  \defeq  \mtroll{\sem\sdynty}{\mtsum{\stagty}\mhole}$ By reflexivity
    \item $\mectxt_{\obcast{\sdynty}{\stagty}}  \defeq  \mtelsecase{\mtunroll\mhole}{\stagty}{\mx}{\mx}{\merr}$ By reflexivity.
    \item $(\sA \neq \sdynty,\floor\sA)\quad\mectxt_{\obcast{\sA}{\sdynty}}  \defeq  \mectxt_{\obcast{\floor\sA}{\sdynty}}\hw{\mectxt_{\obcast{\sA}{\floor\sA}}\mhole}$ By inductive hypothesis and decomposition of ep pairs.
    \item $(\sA \neq \sdynty,\floor\sA)\quad\mectxt_{\obcast{\sdynty}{\sA}}  \defeq  {\mectxt_{\obcast{\floor\sA}{\sA}}}\hw{\mectxt_{\obcast{\sdynty}{\floor\sA}}\mhole}$ By inductive hypothesis and decomposition of ep pairs.
    \item $(\sA,\sB\neq\sdynty \wedge \floor\sA\neq\floor\sB)\quad\mectxt_{\obcast{\sA}{\sB}}  \defeq  \mtlet{\mx}{\mhole}{\merr}$ Not possible that $\sA \ltdyn \sB$.
  \end{enumerate}
\end{proof}
\fi

Next, we show that the ``general'' casts of the gradual language can
be \emph{factorized} into a composition of an upcast followed by a
downcast.
First, we show that factorizing through any type is equivalent to
factorizing through the dynamic type, as a consequence of the
\emph{retraction} property of ep pairs.
\begin{lemma}[Any Factorization is equivalent to Dynamic]
  For any $\sAone,\sAtwo,\sApr$ with $\sAone \ltdyn \sApr$ and $\sAtwo
  \ltdyn \sApr$,
  $\mE_{p,\sAtwo,\sdynty}\hw{\mE_{e,\sAone,\sdynty}}\equidyn
  \mE_{p,\sAtwo,\sApr}\hw{\mE_{e,\sAone,\sApr}}$.
\end{lemma}
\begin{proof}
  By decomposition and the retraction property:
  \[
    \mE_{p,\sAtwo,\sdynty}\hw{\mE_{e,\sAone,\sdynty}}
    \equidyn
    \mE_{p,\sAtwo,\sdynty}\hw{\mE_{p,\sApr,\sdynty}\hw{\mE_{e,\sApr,\sdynty}\hw{\mE_{e,\sAone,\sdynty}}}}
    \equidyn
    \mE_{p,\sAtwo,\sApr}\hw{\mE_{e,\sAone,\sApr}}
    \]
\end{proof}

By transitivity of equivalence, this means that factorization through one $\sB$ is as good as any other.
So to prove that every cast factors as an upcast followed by a
downcast, we can choose whatever middle type is most convenient.
This lets us choose the simplest type possible in the proof.
For instance, when factorizing a function cast
$\obcast{\sAone\sto\sBone}{\sAtwo\sto\sBtwo}$, we can use
the function tag type as the middle type $\sdynty \sto \sdynty$ and
then the equivalence is a simple use of the inductive hypothesis and
the functoriality principle.
\begin{lemma}[Every Cast Factors as Upcast, Downcast]
\label{lem:up-down-factorization}
  For any $\sAone,\sAtwo,\sApr$ with $\sAone \ltdyn \sApr$ and $\sAtwo \ltdyn \sApr$,
  the cast from $\sAone$ to $\sAtwo$ factors through $\sApr$:
  \( \mx : \sem\sA \vDash \mE_{\obcast{\sA}{\sAtwo}}\hw\mx \equidyn \mE_{p,\sAtwo,\sApr}\hw{\mE_{e,\sA,\sApr}\hw\mx} : \sem\sAtwo \)
\end{lemma}
\begin{proof}
  \begin{enumerate}
  \item If $\sAone \ltdyn \sAtwo$, then we choose $\sApr = \sAtwo$ and we need to show that
    \( \mE_{\obcast\sAone\sAtwo} \equidyn \mE_{p,\sAtwo\sAtwo}\hw{\mE_{e,\sAone\sAtwo}} \)
    this follows by \cref{lem:ud-are-casts} and \cref{lem:id-casts}.
  \item If $\sAtwo \ltdyn \sAone$, we use a dual argument to the previous
    case. We choose $\sApr = \sAone$ and we need to show that
    \[ \mE_{\obcast\sAone\sAtwo} \equidyn \mE_{p,\sAone\sAtwo}\hw{\mE_{e,\sAone\sAone}} \]
    this follows by \cref{lem:ud-are-casts} and \cref{lem:id-casts}.
  \item $\mectxt_{\obcast{\sAone\sto\sBone}{\sAtwo\stimes\sBtwo}}
      \defeq \mectxt_{\obcast{\sAone}{\sAtwo}} \mto
      \mectxt_{\obcast{\sBone}{\sBtwo}}$
      We choose $\sApr = \sdynty \sto \sdynty$.
      By inductive hypothesis,
      \[ \mE_{\obcast{\sAtwo}{\sAone}} \equidyn \mE_{p,\sAone,\sdynty}\hw{\mE_{e,\sAtwo,\sdynty}}
      \qandq
      \mE_{\obcast{\sBone}{\sBtwo}} \equidyn \mE_{p,\sBtwo,\sdynty}\hw{\mE_{e,\sBone,\sdynty}}\]
      Then the result holds by functoriality:
      \begin{align*}
        \mE_{\obcast{\sAone\sto\sBone}{\sAtwo\sto\sBtwo}}
        &= \mE_{\obcast\sAtwo\sAone} \mto \mE_{\obcast\sBone\sBtwo}\\
        &\equidyn
        (\mE_{p,\sAone,\sdynty}\hw{\mE_{e,\sAtwo,\sdynty}}) \mto (\mEin{p,\sBtwo,\sdynty}\hw{\mEin{e,\sBone,\sdynty}})\\
        &\equidyn
        (\mE_{e,\sAtwo,\sdynty} \mto \mE_{p,\sBtwo,\sdynty})\hw{\mE_{p,\sAone,\sdynty} \mto \mE_{e,\sBone,\sdynty}}\\
        & =
        \mE_{p,\sAtwo\sto\sBtwo,\sdynty\sto\sdynty}\hw{\mE_{e,\sAone\sto\sBtwo,\sdynty\sto\sdynty}}
      \end{align*}      
    \item (Products, Sums) Same argument as function case.
    \item $(\sAone,\sAtwo\neq\sdynty \wedge
      \floor\sAone\neq\floor\sAtwo)\quad\mectxt_{\obcast{\sAone}{\sAtwo}} \defeq
      \mtlet{\mx}{\mhole}{\merr}$ We choose $\sApr = \sdynty$, so we need to show:
      \(
      \mtlet{\mx}{\mhole}{\merr} \equidyn \mE_{p,\sAtwo,\sdynty}\hw{\mE_{e,\sAone,\sdynty}}
      \).
      By embedding, projection decomposition this is equivalent to
      \[
      \mtlet{\mx}{\mhole}{\merr} \equidyn \mE_{p,\sAtwo,\floor\sAtwo}\hw{\mE_{p,\floor\sAtwo,\sdynty}\hw{\mE_{e,\sAone,\floor\sAone}\hw{\mE_{e,\sAone,\floor\sAone}}}}
      \]
      Which holds by open $\beta$ because the embedding
      ${\mE_{e,\sAone,\floor\sAone}}$ is pure and $\floor\sAone\neq\floor\sAtwo$.
  \end{enumerate}  
\end{proof}

\section{Graduality from EP Pairs}
\label{sec:graduality}

We now define and prove graduality of our cast calculus. 
Graduality, briefly stated, means that if a program is changed to make
its types less dynamic, but otherwise the syntax is the same, then the
\emph{operational behavior} of the term is ``less
dynamic''\footnote{Here we invoke the meaning of dynamic as
  ``active'': less dynamic terms are less active in that they kill the
  program with a type error where a more dynamic program would have
  continued to run.}  in that either the new term has the same
behavior as the old, or it raises a type error, \emph{hiding} some
behavior of the original term.
Graduality, like parametricity, says that a certain type of syntactic
change (making types less dynamic) results in a predictable semantic
change (make behavior less dynamic).
We define these two notions as \emph{syntactic} and \emph{semantic}
term dynamism.

\begin{figure}
\fbox{\small{$\tmdynr{\senvone}{\senvtwo}{\stone}{\sttwo}{\sAone}{\sAtwo}$}}\hfill
  \begin{mathpar}
    \inferrule
    {\tmdynr{\senvone}{\senvtwo}{\stone}{\sttwo}{\sAone}{\sAtwo} \and \sBone \ltdyn \sBtwo}
    {\tmdynr{\senvone}{\senvtwo}{\obcast\sAone\sBone\stone}{\obcast\sAtwo\sBtwo\sttwo}{\sBone}{\sBtwo}}

    \inferrule
    {{\senvone}\dynr{\senvtwo}\and
      \sAone\dynr\sAtwo\and
      \senvonepr \dynr \senvtwopr}
    {\tmdynr{\senvone,\sxone:\sAone,\senvtwo}{\senvtwo,\sxtwo:\sAtwo,\senvtwopr}{\sxone}{\sxtwo}{\sAone}{\sAtwo}}

    \inferrule
    {{\senvone}\dynr{\senvtwo}}
    {\tmdynr{\senvone}{\senvtwo}{\stunit}{\stunit}{\sunitty}{\sunitty}}\and

    \inferrule
    {\tmdynr{\senvone}{\senvtwo}{\stone}{\sttwo}{\styone}{\stytwo}\\
      \tmdynr{\senvone}{\senvtwo}{\stonepr}{\sttwopr}{\styonepr}{\stytwopr}}
    {\tmdynr{\senvone}{\senvtwo}{\stpair{\stone}{\stonepr}}{\stpair{\sttwo}{\sttwopr}}{\spairty{\sAone}{\sAonepr}}{\spairty{\sAtwo}{\sAtwopr}}}

    \inferrule
    {\tmdynr{\senvone}{\senvtwo}{{\stone}}{{\sttwo}}{\spairty{\sAone}{\sAonepr}}{\spairty{\sAtwo}{\sAtwopr}}\\
      \tmdynr{\senvone,{\sxone:\sAone},{\sxonepr:\sApr}}{\senvtwo,{\sxtwo:\sAtwo},{\sxtwopr:\sAtwopr}}{{\stonepr}}{{\sttwopr}}{\sBone}{\sBtwo}
    }
    {\tmdynr{\senvone}{\senvtwo}{\stmatchpair{\sxone:\sAone}{\sxonepr:\sApr}{\stone}{\stonepr}}{\stmatchpair{\sxtwo:\sAtwo}{\sxtwopr:\sAtwopr}{\sttwo}{\sttwopr}}{\sBone}{\sBtwo}}

    \inferrule
    {\tmdynr{\senvone}{\senvtwo}{\stone}{\sttwo}{\sAone}{\sAtwo}\and \sjudgtyprec{\sAonepr}{\sAtwopr}}
    {\tmdynr{\senvone}{\senvtwo}{\stinj{\stermone}}{\stinj{\stermtwo}}{\ssumty{\sAone}{\sAonepr}}{\ssumty{\sAtwo}{\sAtwopr}}}\and

    \inferrule
    {\tmdynr{\senvone}{\senvtwo}{\stone}{\sttwo}{\sAonepr}{\sAtwopr}\and \sjudgtyprec{\sAone}{\sAtwo}}
    {\tmdynr{\senvone}{\senvtwo}{\stinjpr{\stermonepr}}{\stinjpr{\stermtwopr}}{\ssumty{\sAone}{\sAonepr}}{\ssumty{\sAtwo}{\sAtwopr}}}\and

    \inferrule
    {\tmdynr{\senvone}{\senvtwo}{\stone}{\sttwo}{\ssumty{\sAone}{\sAonepr}}{\ssumty{\sAtwo}{\sAtwopr}}\\
      \tmdynr{\senvone,\sxone:\sAone}{\senvtwo,\sxtwo:\sAtwo}{\ssone}{\sstwo}{\sBone}{\sBtwo}\\
      \tmdynr{\senvone,\sxonepr:\sAonepr}{\senvtwo,\sxtwopr:\sAtwopr}{\ssonepr}{\sstwopr}{\sBone}{\sBtwo}}
    {\tmdynr{\senvone}{\senvtwo}{\stcase{\stone}{\sxone:\sAone}{\ssone}{\sxonepr:\sAonepr}{\ssonepr}}{\stcase{\sttwo}{\sxtwo:\sAtwo}{\sstwo}{\sxtwopr:\sAtwopr}{\sstwopr}}{\sBone}{\sBtwo}}
    
    \inferrule
    {\tmdynr{{\senvone},{\svarone}:\sAone}{{\senvtwo},{\svartwo}:\sAtwo}{\stone}{\sttwo}{\sBone}{\sBtwo}}
    {\tmdynr{\senvone}{\senvtwo}{\stfun{\svarone}{\sAone}{\stone}}{\stfun{\svartwo}{\sAtwo}{\sttwo}}{\sfunty{\sAone}{\sBone}}{\sfunty{\sAtwo}{\sBtwo}}}\and

    \inferrule{\tmdynr{\senvone}{\senvtwo}{\stone}{\sttwo}{\sfunty{\sAone}{\sBone}}{\sfunty{\sAtwo}{\sBtwo}}\and
      \tmdynr{\senvone}{\senvtwo}{\ssone}{\sstwo}{{\sAone}}{{\sAtwo}}}
    {\tmdynr{\senvone}{\senvtwo}{\stapp{\stone}{\ssone}}{\stapp{\sttwo}{\sstwo}}{\sBone}{\sBtwo}}
    \end{mathpar}
    \caption{Syntactic Term Dynamism}
    \label{fig:term-dynamism}
\end{figure}
\begin{figure}
\flushleft{\fbox{\small{$\senvone \ltdyn \senvtwo$}}} 

\vspace{-4ex}
  \begin{mathpar}
    \inferrule{~}{\cdot \ltdyn \cdot}\and
    \inferrule
    {\senvone\ltdyn\senvtwo \and \sAone \ltdyn \sAtwo}
    {\senvone,\sxone:\sAone \ltdyn\senvtwo,\sxtwo : \sAtwo}
  \end{mathpar}
  \caption{Environment Dynamism}
  \label{fig:env-dynamism}
\end{figure}

We present syntactic term dynamism in \figref{fig:term-dynamism}, based
on the rules of \citet{refined}.
Syntactic term dynamism captures the above idea of changing a program
to use less dynamic types.
If $\stone \ltdyn \sttwo$, we think of $\sttwo$ as being rewritten to
$\stone$ by changing the types to be less dynamic.
While we will sometimes abbreviate syntactic term dynamism as $\vdash
\stone \ltdyn \sttwo$, the full form is
$\tmdynr{\senvone}{\senvtwo}{\stone}{\sttwo}{\sAone}{\sAtwo}$
and is read as ``$\stone$ is syntactically less dynamic than
$\sttwo$''.
The syntax evokes the invariant that if you rewrite $\sttwo$ to use
less dynamic types $\stone$, then its inputs must be given less
dynamic types $\senvone\ltdyn\senvtwo$ and its outputs must be given
less dynamic types $\sAone\ltdyn\sAtwo$.
We extend type dynamism to environment dynamism in
\figref{fig:env-dynamism} to say $\senvone \ltdyn \senvtwo$ when
$\senvone,\senvtwo$ have the same length and the corresponding types
are related.
The rules of syntactic term dynamism capture exactly the idea of
``types on the left are less dynamic''.
Viewed order-theoretically, these rules say that all term constructors
are \emph{monotone} in types and terms.

The second piece of graduality is a \emph{semantic formulation} of
term dynamism.
The intuition described above is that $\stone$ should be
\emph{semantically} less dynamic than $\sttwo$ when it has the same
behavior as $\sttwo$ except possibly when it errors.
Note that if $\senvone = \senvtwo$ and $\sAone = \sAtwo$, this is
exactly what observational error approximation formalizes.
Of course, since we can cast between any two types, we can cast any
term to be of a different type.
Our definition for semantic term dynamism will then be contextual
approximation \emph{up to cast}:

\begin{definition}[Observational Term Dynamism]
  We say $\senvone \vdash \stone : \sBone$ is \emph{observationally
    less dynamic} than $\senvtwo \vdash \sttwo : \sBtwo$, written
  $\senvone\ltdyn\senvtwo \vDash \stone \obsapprox \sttwo : \sBone
  \ltdyn \sBtwo$ when
  \[
  \senvone \vDash \obcast\sBone\sBtwo\stone \obsapprox
  \stletvert{\sxin{2,1}}{\obcast{\sAin{1,1}}{\sAin{2,1}}{\sxin{1,1}}}{
    \begin{stackTL}
      \vdots\\
      \stletvert{\sxin{2,n}}{\obcast{\sAin{1,n}}{\sAin{2,n}}{\sxin{1,n}}}{
      \sttwo}
    \end{stackTL}
  } : \sBtwo  
  \]
  where $\senvone =
  \sxin{1,1}:\sAin{1,1},\ldots,\sxin{1,n}:\sAin{1,n}$ and $\senvtwo =
  \sxin{2,1}:\sAin{2,1},\ldots,\sxin{2,n}:\sAin{2,n}$. Or,
  abbreviated as:
  \[
  \senvone \vDash \obcast\sBone\sBtwo\stone \obsapprox \stlet{\senvtwo}{\obcast\senvone\senvtwo\senvone}{\sttwo}:\sBtwo
  \]
\end{definition}

Note that we have chosen to use the two \emph{upcasts}, but there are
three other ways we could have inserted casts to give $\stone,\sttwo$
the same type: we can use upcasts or downcasts on the inputs and we
can use upcasts or downcasts on the outputs.
We will show based on the ep-pair property of upcasts and downcasts
that all of these are equivalent (\cref{lem:alternative}).

We then define graduality to mean that syntactic term dynamism implies
semantic term dynamism:
\begin{theorem}[Graduality]
\label{thm:graduality}
  If $\tmdynr\senvone\senvtwo\stone\sttwo\sAone\sAtwo$, then $\obstmdynr\senvone\senvtwo\stone\sttwo\sAone\sAtwo$
\end{theorem}
\begin{proof}
  By \cref{lem:log-to-obs:graduality,lem:adequacy,lem:fund-lemma,thm:logical-graduality}.
\end{proof}

Next, we present our logical relations method for proving graduality.
First, to prove an approximation result for terms in $\glangname$, we
will prove approximation for their translations in $\tlangname$,
justified by our adequacy theorem.
Second, to prove observational approximation, we will use our logical
relation, justified by our soundness theorem.
For that we use the following ``logical'' formulation of term dynamism.

\begin{definition}[Logical Term Dynamism]
For any $\sem\senvone\vdash\mtone:\sem\sAone$ and
$\sem\senvtwo\vdash \mttwo : \sem\sAtwo$
with $\senvone\ltdyn \senvtwo$ and $\sAone \ltdyn \sAtwo$, we define
$\logtmdynr\senvone\senvtwo\mtone\mttwo\sAone\sAtwo$ as
\[
\tmlogapprox{\sem\senvone}{\theemb\sAone\sAtwo{\mtone}}{\letembnovert{\mttwo}}{\sem\sAtwo}
\]
where the right hand side is defined analogous to the environment cast
$\obcast\senvone\senvtwo$.
\end{definition}

\begin{lemma}[Logical Term Dynamism implies Observational Term Dynamism]
\label{lem:log-to-obs:graduality}
For any $\senvone\vdash\stone:\sAone$ and $\senvtwo\vdash \sttwo :
\sAtwo$ with $\senvone\ltdyn \senvtwo$ and $\sAone \ltdyn \sAtwo$, if
$\logtmdynr{\senvone}{\senvtwo}{\sem\stone}{\sem\sttwo}{\sAone}{\sAtwo}$
then
$\obstmdynr{\senvone}{\senvtwo}{\stone}{\sttwo}{\sAone}{\sAtwo}$.
\end{lemma}
\begin{proof}
  By \cref{thm:log-to-obs,lem:typed-gradual-approx}.
\end{proof}

Now that we are in the realm of logical approximation, we have all the
lemmas of \secref{sec:lemmas} at our disposal, and we now start
putting them to work.
First, as mentioned before, we show that at least with logical term
dynamism, the use of upcasts was arbitrary; we could have used
downcasts instead.
The property we need is that the upcast and downcast are
\emph{adjoint} (in the language of category theory), also known as a
\emph{Galois connection}, which is a basic consequence of the
definition of ep pair:
\begin{lemma}[EP Pairs are Adjoint]
\label{lem:ep-adjoint}
  For any ep pair $(\mEin{e},\mEin{p}) : \mAone \eppair \mAtwo$, and terms
  $\menv \vdash \mtone : \mAone, ,\menv \vdash \mttwo : \mAtwo$,
  \begin{mathpar}
    \menv \vDash \mEin{e}\hw{\mtone} \logltdyn \mttwo : \mAtwo \quad\iff\quad
    \menv \vDash \mtone \logltdyn \mEin{p}\hw{\mttwo} : \mAone
  \end{mathpar}
\end{lemma}
\begin{proof}
  The two proofs are dual .

  \iflong
  \begin{mathpar}
    \inferrule*[right=Transitivity]
    {\inferrule*[right=EP Pair]
      {~}
      {\menv \vDash \mtone \logltdyn \mEin{p}\hw{\mEin{e}\hw{\mtone}} : \mAone}\\
      \inferrule*[right=Congruence]
      {\inferrule*[right=Assumption]{~}{\menv \vDash \mEin{e}\hw{\mtone} \logltdyn \mttwo : \mAtwo}}
      {\menv \vDash \mEin{p}\hw{\mEin{e}\hw{\mtone}}\logltdyn \mEin{p}\hw{\mttwo} : \mAone}
    }
    {\menv \vDash \mtone \logltdyn \mEin{p}\hw{\mttwo} : \mAone}
    
    \iflong
    \inferrule*[right=Transitivity]
    {\inferrule*[right=Congruence]
      {\inferrule*[right=Assumption]{~}{\menv \vDash \mEin{e}\hw\mtone \logltdyn \mEin{p}\hw{\mttwo} : \mAone}}
      {\menv \vDash \mtone \logltdyn \mEin{e}\hw{\mEin{p}\hw{\mttwo}} : \mAone}\\
     \inferrule*[right=EP Pair]
      {~}
      {\menv \vDash \mEin{e}\hw{\mEin{p}\hw{\mttwo}} \logltdyn \mttwo}
    }
    {\menv \vDash \mEin{e}\hw\mtone \logltdyn \mttwo : \mAtwo}
    \fi
  \end{mathpar}
  \fi
\end{proof}

\begin{lemma}[Adjointness on Inputs]
\label{lem:adj-inp}
If $\menv,\mxone:\mAone \vdash \mtone : \mB$ and $\menv,\mxtwo:\mAtwo
\vdash \mttwo : \mB$, and $\mEin{e},\mEin{p} : \mAone \eppair \mAtwo$,
then
\[ \tmlogapprox{\menv,\mxone:\mAone}{\mtone}{\mtlet{\mxtwo}{\mEin{e}\hw{\mxone}}{\mttwo}}{\mB} \quad \iff\quad  \tmlogapprox{\menv,\mxtwo:\mAtwo}{\mtlet{\mxone}{\mEin{p}\hw{\mxtwo}}{\mtone}}{\mttwo}{\mB} \]
\end{lemma}
\begin{proof}
  By a similar argument to \cref{lem:ep-adjoint}
\end{proof}

\begin{lemma}[Alternative Formulations of Logical Term Dynamism]
\label{lem:alternative}
  The following are equivalent
  \begin{enumerate}
  \item $\tmlogapprox{\sem\senvone}{\theemb\sAone\sAtwo{\mtone}}{\letembnovert{\mttwo}}{\sem\sAtwo}$
  \item $\tmlogapprox{\sem\senvone}{\mtone}{\letembnovert{\theprj\sAone\sAtwo{\mttwo}}}{\sem\sAtwo}$
  \item $\tmlogapprox{\sem\senvone}{\letprjnovert{\mtone}}{{\theprj\sAone\sAtwo{\mttwo}}}{\sem\sAtwo}$
  \item $\tmlogapprox{\sem\senvone}{\theemb\sAone\sAtwo{\letprjnovert{\mtone}}}{{\mttwo}}{\sem\sAtwo}$
  \end{enumerate}
\end{lemma}
\begin{proof}
  By induction on $\senvone$, using \cref{lem:ep-adjoint} and \lemref{lem:adj-inp}
\end{proof}

Finally, to prove the graduality theorem, we do an induction over all
the cases of syntactic term dynamism.
Most important is the cast case $\obcast\sAone\sBone\stone \ltdyn
\obcast\sAtwo\sBtwo\sttwo$ which is valid when $\sAone \ltdyn \sAtwo$
and $\sBone\ltdyn\sBtwo$.
We break up the proof into 4 atomic steps using the factorization of
general casts into an upcast followed by a downcast
(\cref{lem:up-down-factorization}): $\mE_{\obcast\sAone\sAtwo}
\equidyn \mectxt_{p,\sAtwo,\sdynty}\hw{\mE_{e,\sAone,\sdynty}}$.
The four steps are upcast on the left, downcast on the left, upcast on
the right, and downcast on the right.
These are presented as rules for logical dynamism in
\figref{fig:term-dynamism-casts}.
Each of the inference rules accounts for two cases.
The \textsc{Cast-Right} rule says first that if $\mtone \logltdyn \mttwo :
\sAone \ltdyn \sAtwo$ that it is OK to cast $\mttwo$ to $\sBtwo$, as
long as $\sBtwo$ is more dynamic than $\sAone$, and the cast is either
an upcast or downcast.
Here, our explicit inclusion of $\sAone \ltdyn \sBtwo$ in the syntax
of the term dynamism judgment should help: the rule says that adding
an upcast or downcast to $\mttwo$ results in a more dynamic term than
$\mtone$, \emph{whenever it is even sensible to ask}: i.e., if it were
not the case that $\sAone \ltdyn \sBtwo$, the judgment would not be
well-formed, so the judgment holds whenever it makes sense!
The \textsc{Cast-Left} rule is dual.

\begin{figure}
  \begin{mathpar}
    \inferrule*[right=Cast-Right]
    {\logtmdynr{\senvone}{\senvtwo}{\mtone}{\mttwo}{\sAone}{\sAtwo}
      \and \sAone \dynr \sBtwo\and
      (\sAtwo \ltdyn \sBtwo \vee \sBtwo \ltdyn \sAtwo)
    }
    {\logtmdynr{\senvone}{\senvtwo}{\mtone}{\mE_{\obcast\sAtwo\sBtwo}\hw{\mttwo}}{\sAone}{\sBtwo}}

    \inferrule*[right=Cast-Left]
    {\logtmdynr{\senvone}{\senvtwo}{\mtone}{\mttwo}{\sAone}{\sAtwo}
      \and \sBone \ltdyn \sAtwo
      \and (\sAone \ltdyn \sBone \vee \sBone \ltdyn \sAone)}
    {\logtmdynr{\senvone}{\senvtwo}{\mE_{\obcast{\sAone}{\sBone}}\hw{\mtone}}{\mttwo}{\sBone}{\sAtwo}}

  \end{mathpar}
  \caption{Term Dynamism Upcast, Downcast Rules}
  \label{fig:term-dynamism-casts}
\end{figure}

These four rules, combined with our factorization of casts into upcast
followed by downcast suffice to prove the congruence rule for casts
(we suppress the context $\senvone \ltdyn \senvtwo \vDash$, which is the
same in each line):
\begin{mathpar}
  \inferrule*[right=\cref{lem:up-down-factorization}]
  {
    \inferrule*[right=Cast-Right]{
      \inferrule*[right=Cast-Left]{
        \inferrule*[right=Cast-Left]{
          \inferrule*[right=Cast-Right]{
            {{{\sem\stone}} \ltdyn {{\sem\sttwo}} : \sAone \ltdyn \sAtwo}
          }
          {{{\sem\stone}} \ltdyn {\mE_{e,\sAtwo,\sdynty}\hw{\sem\sttwo}} : \sAone \ltdyn \sdynty}
        }
        {{{\mE_{e,\sAone,\sdynty}\hw{\sem\stone}} \ltdyn {\mE_{e,\sAtwo,\sdynty}\hw{\sem\sttwo}}} : \sdynty \ltdyn \sdynty}
      }
      {{\mE_{p,\sBone,\sdynty}\hw{\mE_{e,\sAone,\sdynty}\hw{\sem\stone}} \ltdyn {\mE_{e,\sAtwo,\sdynty}\hw{\sem\sttwo}}} : \sBone \ltdyn \sdynty}
    }
    {\mE_{p,\sBone,\sdynty}\hw{\mE_{e,\sAone,\sdynty}\hw{\sem\stone}} \ltdyn \mE_{p,\sBtwo,\sdynty}\hw{\mE_{e,\sAtwo,\sdynty}\hw{\sem\sttwo}} : \sBone \ltdyn \sBtwo}
  }
  {{\sem{\obcast\sAone\sBone\stone}} \ltdyn {\sem{\obcast{\sAtwo}{\sBtwo}\sttwo}} : {\sBone} \ltdyn {\sBtwo}}
\end{mathpar}

Next, we show the 4 rules are valid, as simple consequences of
the ep pair property and the decomposition theorem.
Also note that while there are technically 4 cases, each comes in a
pair where the proofs are exactly dual, so conceptually speaking there
are only 2 arguments.
\begin{lemma}[Upcast, Downcast Dynamism]
\label{lem:up-down-cases}
  The four rules in \figref{fig:term-dynamism-casts} are valid.
\end{lemma}
\begin{proof}
  In each case we choose which case of \cref{lem:alternative} is
  simplest.
  \begin{enumerate}
  \item \textsc{Cast-Left} with $\sAone \ltdyn \sBone \ltdyn
    \sAtwo$. We need to show
    $\mE_{e,\sBone,\sAtwo}\hw{\mE_{e,\sAone,\sBone}\hw{\mtone}} \ltdyn
    \letemb{\mttwo}$. By decomposition and congruence,
    $\mE_{e,\sBone,\sAtwo}\hw{\mE_{e,\sAone,\sBone}\hw{\mtone}}
    \equidyn \mE_{e,\sAone,\sAtwo}$ so the conclusion holds by
    transitivity and the premise.

  \item \textsc{Cast-Right} with $\sAone \ltdyn \sBtwo \ltdyn
    \sAtwo$. We need to show $\letprj{\mtone} \ltdyn
    \mE_{p,\sAone,\sBtwo}\hw{\mE_{p,\sBtwo,\sAtwo}\hw{\mttwo}}$.  By
    decomposition and congruence,
    $\mE_{p,\sAone,\sBtwo}\hw{\mE_{p,\sBtwo,\sAtwo}\hw{\mttwo}}
    \equidyn \mE_{p,\sAone,\sAtwo}\hw{\mttwo}$, so the conclusion
    holds by transitivity and the premise.
    
  \item \textsc{Cast-Left} with $\sBone \ltdyn \sAone \ltdyn
    \sAtwo$. We need to show $\mE_{p,\sBone,\sAone}\hw{\letprj\mtone}
    \ltdyn \mE_{p,\sBone,\sAtwo}\hw{\mttwo}$. By decomposition,
    $\mE_{p,\sBone,\sAtwo}\hw{\mttwo} \equidyn
    \mE_{p,\sBone,\sAone}\hw{\mE_{p,\sAone,\sAtwo}\hw{\mttwo}}$, so by
    transitivity it is sufficient to show
    \[ 
    \mE_{p,\sBone,\sAone}\hw{\letprjnovert\mtone}
    \ltdyn
    \mE_{p,\sBone,\sAone}\hw{\mE_{p,\sAone,\sAtwo}\hw{\mttwo}}
    \]
    which follows by congruence and the premise.    
  \item \textsc{Cast-Right} with $\sAone \ltdyn \sAtwo \ltdyn \sBtwo$.
    We need to show $\mE_{e,\sAone,\sBtwo}\hw{\mtone} \ltdyn
    \mE_{e,\sAtwo,\sBtwo}\hw{\letemb{\mttwo}}$.

    By decomposition, $\mE_{e,\sAone,\sBtwo}\hw{\mtone} \equidyn
    \mE_{e,\sAtwo,\sBtwo}\hw{\mE_{e,\sAone,\sAtwo}\hw{\mtone}}$, so
    by transitivity it is sufficient to show
    \[
    \mE_{e,\sAtwo,\sBtwo}\hw{\mE_{e,\sAone,\sAtwo}\hw{\mtone}}
    \ltdyn
    \mE_{e,\sAtwo,\sBtwo}\hw{\letembnovert{\mttwo}}
    \]
    which follows by congruence and the premise.
  \end{enumerate}
\end{proof}

Finally, we prove the graduality theorem by induction on syntactic
term dynamism derivations, finishing the proof of
\cref{thm:graduality}.
\begin{theorem}[Logical Graduality]
\label{thm:logical-graduality}
  If $\senvone \ltdyn \senvtwo \vdash \stone \ltdyn \sttwo : \sAone$, then
  $\logtmdynr{\senvone}{\senvtwo}{\sem\stone}{\sem\sttwo}{\sAone}{\sAtwo}$.
\end{theorem}
\iflong
\begin{proof} By induction on syntactic term dynamism rules.
  \begin{enumerate}
  \item To show To show
    $\inferrule{\logtmdynr{\senvone}{\senvtwo}{\sem\stone}{\sem\sttwo}{\sBone}{\sBtwo}}{\logtmdynr{\senvone}{\senvtwo}{\sem{\obcast\sAone\sBone\stone}}{\sem{\obcast{\sAtwo}{\sBtwo}\sttwo}}{\sBone}{\sBtwo}}$
      we use \cref{lem:up-down-cases} and the argument above.
  \item $\inferrule
    {{\senvone}\dynr{\senvtwo}\and
      \sAone\dynr\sBone\and
      \senvonepr \dynr \senvtwopr}
    {\semtmdynr{\senvone,\sxone:\sAone,\senvtwo}{\senvtwo,\sxtwo:\sAtwo,\senvtwopr}{\sxone}{\sxtwo}{\sAone}{\sAtwo}}$
    We need to show:
    \[ \sem{\senvone} \vDash \theemb{\sAone}{\sAtwo}{\mxone} \ltdyn
    \letemb{\mxtwo}\]

    Since embeddings are pure \cref{lem:emb-pure,lem:pure-subst} we
    can substitute them in and then the two sides are literally the
    same.
    \[ \letemb{\mxtwo} \equidyn \mxtwo[\theemb{\sAone^i}{\sAtwo^i}{\mxone^i}/\mxtwo^i] = \theemb{\sAone}{\sAtwo}{\mxone} \]

  \item $\inferrule
    {\semtmdynr{\senvone}{\senvtwo}{\stone}{\sttwo}{\sAone}{\sAtwo}\and \sjudgtyprec{\sAonepr}{\sAtwopr}}
    {\semtmdynr{\senvone}{\senvtwo}{\stinj{\stermone}}{\stinj{\stermtwo}}{\ssumty{\sAone}{\sAonepr}}{\ssumty{\sAtwo}{\sAtwopr}}}$
    Expanding definitions, we need to show:
    \[
    \mtcasevert{\mtinj{\sem{\stone}}}
           {\mx}{\mtinj({\theemb{\sAone}{\sAtwo}{\mx}})}
           {\mx}{\mtinjpr({\theemb{\sAonepr}{\sAtwopr}{\mxpr}})}
    \ltdyn
    \letemb{\sem{\sttwo}}
    \]
    By open $\beta$ (\cref{lem:open-beta}), the left side can be
    reduced, which we can then substitute into due to linearity of
    evaluation contexts (\cref{lem:evctx-linear}):
    \begin{align*}
      \mtcase{\mtinj{\sem{\stone}}}
           {\mx}{\mtinj({\theemb{\sAone}{\sAtwo}{\mx}})}
           {\mx}{\mtinjpr({\theemb{\sAonepr}{\sAtwopr}{\mxpr}})}
           & \equidyn \mtletvert{\mx}{\sem{\stone}}{\mtinj({\theemb{\sAone}{\sAtwo}{\mx}})}\\
           & \equidyn (\mtinj({\theemb{\sAone}{\sAtwo}{\mx}}))[\sem{\stone}/\mx]\\
           & = \mtinj({\theemb{\sAone}{\sAtwo}{\sem{\stone}}})
    \end{align*}
    So by transitivity it is sufficient to show
    \[ \mtinj({\theemb{\sAone}{\sAtwo}{\sem{\stone}}}) \ltdyn \mtinj{\left(\letemb{\sem{\sttwo}}\right)}\]
    which follows by congruence (\cref{lem:cong}).
  \item $\inferrule
    {\semtmdynr{\senvone}{\senvtwo}{\stone}{\sttwo}{\sAonepr}{\sAtwopr}\and \sjudgtyprec{\sAone}{\sAtwo}}
    {\semtmdynr{\senvone}{\senvtwo}{\stinjpr{\stermonepr}}{\stinjpr{\stermtwopr}}{\ssumty{\sAone}{\sAonepr}}{\ssumty{\sAtwo}{\sAtwopr}}}$
    essentially the same as the previous case.
  \item $\inferrule
    {\semtmdynr{\senvone}{\senvtwo}{\stone}{\sttwo}{\ssumty{\sAone}{\sAonepr}}{\ssumty{\sAtwo}{\sAtwopr}}\\
      \semtmdynr{\senvone,\sxone:\sAone}{\senvtwo,\sxtwo:\sAtwo}{\ssone}{\sstwo}{\sBone}{\sBtwo}\\
      \semtmdynr{\senvone,\sxonepr:\sAonepr}{\senvtwo,\sxtwopr:\sAtwopr}{\ssonepr}{\sstwopr}{\sBone}{\sBtwo}}
    {\semtmdynr{\senvone}{\senvtwo}{\stcase{\stone}{\sxone:\sAone}{\ssone}{\sxonepr:\sAonepr}{\ssonepr}}{\stcase{\sttwo}{\sxtwo:\sAtwo}{\sstwo}{\sxtwopr:\sAtwopr}{\sstwopr}}{\sBone}{\sBtwo}}$
    Expanding definitions, we need to show
    \[
    \theemb{\sBone}{\sBtwo}{\mtcasevert {\sem{\stone}}
      {\mtinj \mxone}{\sem{\msone}}
      {\mtinjpr \mxonepr}{\sem{\msonepr}}
    }
    \ltdyn
    \letemb{\mtcasevert{\sem{\sttwo}}{\mtinj\mxtwo}{\sem{\mstwo}}{\mtinjpr\mxtwopr}{\sem{\mstwopr}}}
    \]
    First, we do some simple rewrites: on the left side, we use a
    commuting conversion to push the embedding into the continuations:
    \[
    \theemb{\sBone}{\sBtwo}{\mtcasevert{\sem{\stone}}
      {\mtinj \mxone}{\sem{\msone}}
      {\mtinjpr \mxonepr}{\sem{\msonepr}}}
    \equidyn
    \mtcasevert{\sem{\stone}}
      {\mxone}{\theemb{\sBone}{\sBtwo}{\sem{\msone}}}
      {\mxonepr}{\theemb{\sBone}{\sBtwo}{\sem{\msonepr}}}
    \]
    And on the right side we use the fact that embeddings are pure and
    so can be moved freely:
    \[
    \letemb{\mtcasevert{\sem{\sttwo}}{\mtinj\mxtwo}{\sem{\mstwo}}{\mtinjpr\mxtwopr}{\sem{\mstwopr}}}
    \equidyn
    \mtcasevert{\letemb{\sem{\sttwo}}}
           {\mxtwo}{\letemb{\sem{\mstwo}}}
           {\mxtwopr}{\letemb{\sem{\mstwopr}}}
    \]
    Next as with many of the elim forms, we ``ep-expand'' the
    discriminee on the left side, and then simplify based on the
    definition of
    $\theprj{\sAone\splus\sAonepr}{\sAtwo\splus\sAtwopr}\cdot$, using
    the case-of-case commuting conversion and open $\beta$
    \cref{lem:comm-conv,lem:open-beta}:
    \begin{align*}
      \mtcasevert{\sem{\stone}}
      {\mxone}{\theemb{\sBone}{\sBtwo}{\sem{\msone}}}
      {\mxonepr}{\theemb{\sBone}{\sBtwo}{\sem{\msonepr}}}
      &\equidyn
      \mtcasevert{\theprj{\sAone\splus\sAonepr}{\sAtwo\splus\sAtwopr}{\theemb{\sAone\splus\sAonepr}{\sAtwo\splus\sAtwopr}{\stone}}}
      {\mxone}{\theemb{\sBone}{\sBtwo}{\sem{\msone}}}
      {\mxonepr}{\theemb{\sBone}{\sBtwo}{\sem{\msonepr}}}\\
      \text{(definition)}&= \mtcasevert{\left({\mtcasevert{\theemb{\sAone\splus\sAonepr}{\sAtwo\splus\sAtwopr}{\stone}}
          {\mxtwo}{\mtinj{\theprj{\sAone}{\sAtwo}{\mxtwo}}}
          {\mxtwopr}{\mtinj{\theprj{\sAonepr}{\sAtwopr}{\mxtwopr}}}
        }\right)}
      {\mxone}{\theemb{\sBone}{\sBtwo}{\sem{\msone}}}
      {\mxonepr}{\theemb{\sBone}{\sBtwo}{\sem{\msonepr}}}\\
      \text{(comm conv \ref{lem:comm-conv}, open $\beta$ \ref{lem:open-beta})}& \equidyn
      \mtcasevert{\theemb{\sAone\splus\sAonepr}{\sAtwo\splus\sAtwopr}{\stone}}
      {\mxtwo}{\mtletvert{\mxone}{\theprj{\sAone}{\sAtwo}{\mxtwo}}{{\theemb{\sBone}{\sBtwo}{\sem{\msone}}}}}
      {\mxtwopr}{\mtletvert{\mxonepr}{\theprj{\sAonepr}{\sAtwopr}{\mxtwopr}}{\theemb{\sBone}{\sBtwo}{\sem{\msonepr}}}}
    \end{align*}
    Then the final step follows by congruence and adjointness on
    inputs \cref{lem:cong,lem:adj-inp}:
    \[
    \mtcasevert{\theemb{\sAone\splus\sAonepr}{\sAtwo\splus\sAtwopr}{\stone}}
      {\mxtwo}{\mtletvert{\mxone}{\theprj{\sAone}{\sAtwo}{\mxtwo}}{\theemb{\sBone}{\sBtwo}{\sem{\msone}}}}
      {\mxtwopr}{\mtletvert{\mxonepr}{\theprj{\sAonepr}{\sAtwopr}{\mxtwopr}}{\theemb{\sBone}{\sBtwo}{\sem{\msonepr}}}}
      \ltdyn
    \mtcasevert{\letemb{\sem{\sttwo}}}
           {\mxtwo}{\letemb{\sem{\mstwo}}}
           {\mxtwopr}{\letemb{\sem{\mstwopr}}}
    \]

  \item $\inferrule{{\senvone}\dynr{\senvtwo}}
    {\semtmdynr{\senvone}{\senvtwo}{\stunit}{\stunit}{\sunitty}{\sunitty}}$.
    Expanding we need to show
    \[ \theemb{\sunitty}{\sunitty} \mtunit \ltdyn \letemb{\mtunit} \]
    By definition, the left side is just $\mtunit$ and the right side
    after a substitution, valid because embeddings are pure
    \cref{lem:emb-pure,lem:pure-subst}.

  \item $\inferrule
    {\semtmdynr{\senvone}{\senvtwo}{\stone}{\sttwo}{\sAone}{\sAtwo}\\
      \semtmdynr{\senvone}{\senvtwo}{\ssone}{\sstwo}{\sBone}{\sBtwo}}
    {\semtmdynr{\senvone}{\senvtwo}{\stpair{\stone}{\ssone}}{\stpair{\sttwo}{\sstwo}}{{\sAone} \stimes {\sBone}}{{\sAtwo} \stimes {\sBtwo}}}$. Expanding definitions, we need to show
    \[
    \theemb{\sAone \stimes \sBone}{\sAtwo \stimes \sBtwo}
           {\mtpair{\sem\stone}{\sem\ssone}}
    \ltdyn
    \letemb{\mtpair{\sem\sttwo}{\sem\sstwo}}
    \]
    On the right, we duplicate the embeddings, justified by
    \cref{lem:emb-pure,lem:pure-subst}, to set up congruence:
    \[
    \letemb{\mtpair{\sem\sttwo}{\sem\sstwo}}
    \equidyn
    \mtpair{\letemb{\sem\sttwo}}{\letemb{\sem\sstwo}}
    \]

    On the left, we use linearity of evaluation contexts to lift the
    terms out, then perform some open $\beta$ reductions and put the
    terms back in:
    \begin{align*}
      \theemb{\sAone \stimes \sBone}{\sAtwo \stimes \sBtwo}
             {\mtpair{\sem\stone}{\sem\ssone}}
      & \equidyn
      \mtletvert{\sx}{\sem\stone}{
      \mtletvert{\sy}{\sem\ssone}{
      \mtmatchpair{\sx}{\sy}{\mtpair{\sx}{\sy}}{\mtpair{\theemb{\sAone}{\sAtwo}{\sx}}{\theemb{\sBone}{\sBtwo}{\sy}}}
      }
      }\\
      \by{open $\beta$}{lem:open-beta}&\equidyn
      \mtletvert{\sx}{\sem\stone}{
      \mtletvert{\sy}{\sem\ssone}{
      {\mtpair{\theemb{\sAone}{\sAtwo}{\sx}}{\theemb{\sBone}{\sBtwo}{\sy}}}
      }
      }\\
      \by{linearity}{lem:evctx-linear}&\equidyn
      {\mtpair{\theemb{\sAone}{\sAtwo}{\sem\stone}}{\theemb{\sBone}{\sBtwo}{\sem\ssone}}}
    \end{align*}
    With the final step following by congruence (\cref{lem:cong}) and the premise:
    \[
    {\mtpair{\theemb{\sAone}{\sAtwo}{\sem\stone}}{\theemb{\sBone}{\sBtwo}{\sem\ssone}}}
    \ltdyn
    \mtpair{\letemb{\sem\sttwo}}{\letemb{\sem\sstwo}}
    \]
  \item $\inferrule
    {\semtmdynr{\senvone}{\senvtwo}{{\stone}}{{\sttwo}}{\spairty{\sAone}{\sAonepr}}{\spairty{\sAtwo}{\sAtwopr}}\\ \semtmdynr{\senvone,{\sxone:\sAone},{\sxonepr:\sApr}}{\senvtwo,{\sxtwo:\sAtwo},{\sxtwopr:\sAtwopr}}{{\stonepr}}{{\sttwopr}}{\sBone}{\sBtwo}
    }
    {\semtmdynr{\senvone}{\senvtwo}{\stmatchpair{\sxone:\sAone}{\sxonepr:\sApr}{\stone}{\stonepr}}{\stmatchpair{\sxtwo:\sAtwo}{\sxtwopr:\sAtwopr}{\sttwo}{\sttwopr}}{\sBone}{\sBtwo}}$ Expanding definitions, we need to show
    \[
    \sem\senvone\vDash
    \theemb{\sBone}{\sBtwo}{{\mtmatchpair{\mxone}{\mxonepr}{\sem\stone}{\sem\stonepr}}}
    \ltdyn
    \letemb{\mtmatchpair{\mxtwo}{\mxtwopr}{\sem\sttwo}{\sem\sttwopr}}
    : \sem\sBtwo
    \]
    On the right side, in anticipation of a use of congruence, we push
    the embeddings in \cref{lem:emb-pure,lem:pure-subst}:
    \[
    \letemb{\mtmatchpair{\mxtwo}{\mxtwopr}{\sem\sttwo}{\sem\sttwopr}}
    \equidyn
    \mtmatchpairvert{\mxtwo}{\mxtwopr}{\left(\letemb{\sem\sttwo}\right)}{\letemb{\sem\sttwopr}}
    \]
    On the left side, we perform a commuting conversion, ep
    expand the discriminee and do some open $\beta$ reductions to
    simplify the expression.
    \begin{align*}
      \theemb{\sBone}{\sBtwo}{{\mtmatchpair{\mxone}{\mxonepr}{\sem\stone}{\sem\stonepr}}}
      & \equidyn
      {\mtmatchpair{\mxone}{\mxonepr}{\sem\stone}{\theemb{\sBone}{\sBtwo}{\sem\stonepr}}}\\
      \by{ep pair}{lem:dyn-der-ep}& \equidyn
      {\mtmatchpairvert{\mxone}{\mxonepr}{\theprj{\spairty{\sAone}{\sAonepr}}{\spairty{\sAtwo}{\sAtwopr}}{\theemb{\spairty{\sAone}{\sAonepr}}{\spairty{\sAtwo}{\sAtwopr}}{\sem\stone}}}{\theemb{\sBone}{\sBtwo}{\sem\stonepr}}}\\
      \text{(definition)}&=
      {\mtmatchpairvert{\mxone}{\mxonepr}
        {\mtmatchpairvert{\mxtwo}{\mxtwopr}
          {\theemb{\spairty{\sAone}{\sAonepr}}{\spairty{\sAtwo}{\sAtwopr}}{\sem\stone}}
          {\mtpair{\theprj{\sAone}{\sAtwo}{\mxtwo}}{\theprj{\sAonepr}{\sAtwopr}{\mxtwo}}}}
        {\theemb{\sBone}{\sBtwo}{\sem\stonepr}}}\\
      \by{linearity}{lem:evctx-linear}&\equidyn
      {\mtmatchpairvert{\mxone}{\mxonepr}
        {\mtmatchpairvert{\mxtwo}{\mxtwopr}
          {\theemb{\spairty{\sAone}{\sAonepr}}{\spairty{\sAtwo}{\sAtwopr}}{\sem\stone}}
          {\mtletvert{\mxone}{\theprj{\sAone}{\sAtwo}{\mxtwo}}{
           \mtletvert{\mxonepr}{\theprj{\sAonepr}{\sAtwopr}{\mxtwopr}}
           {\mtpair{\mxone}{\mxonepr}}}}}
        {\theemb{\sBone}{\sBtwo}{\sem\stonepr}}}\\
      \by{comm. conv.}{lem:comm-conv}&\equidyn
      {\mtmatchpairvert{\mxtwo}{\mxtwopr}
          {\theemb{\spairty{\sAone}{\sAonepr}}{\spairty{\sAtwo}{\sAtwopr}}{\sem\stone}}
          {\mtletvert{\mxone}{\theprj{\sAone}{\sAtwo}{\mxtwo}}{
           \mtletvert{\mxonepr}{\theprj{\sAonepr}{\sAtwopr}{\mxtwopr}}
           {\mtmatchpairvert{\mxone}{\mxonepr}{\mtpair{\mxone}{\mxonepr}}
             {\theemb{\sBone}{\sBtwo}{\sem\stonepr}}}}}}\\
      \by{open $\beta$}{lem:open-beta}&\equidyn
      {\mtmatchpairvert{\mxtwo}{\mxtwopr}
          {\theemb{\spairty{\sAone}{\sAonepr}}{\spairty{\sAtwo}{\sAtwopr}}{\sem\stone}}
          {\mtletvert{\mxone}{\theprj{\sAone}{\sAtwo}{\mxtwo}}{
           \mtletvert{\mxonepr}{\theprj{\sAonepr}{\sAtwopr}{\mxtwopr}}
           {\theemb{\sBone}{\sBtwo}{\sem\stonepr}}}}}
    \end{align*}
    The final step is by congruence and adjointness on inputs (\cref{lem:cong,lem:adj-inp}):
    \[
    {\mtmatchpairvert{\mxtwo}{\mxtwopr}
          {\theemb{\spairty{\sAone}{\sAonepr}}{\spairty{\sAtwo}{\sAtwopr}}{\sem\stone}}
          {\mtletvert{\mxone}{\theprj{\sAone}{\sAtwo}{\mxtwo}}{
           \mtletvert{\mxonepr}{\theprj{\sAonepr}{\sAtwopr}{\mxtwopr}}
           {\theemb{\sBone}{\sBtwo}{\sem\stonepr}}}}}
    \ltdyn
    \mtmatchpairvert{\mxtwo}{\mxtwopr}{\left(\letemb{\sem\sttwo}\right)}{\letemb{\sem\sttwopr}}
    \]
    
  \item $\inferrule
    {\semtmdynr{{\senvone},{\svarone}:\sAone}{{\senvtwo},{\svartwo}:\sAtwo}{\stone}{\sttwo}{\sBone}{\sBtwo}}
    {\semtmdynr{\senvone}{\senvtwo}{\stfun{\svarone}{\sAone}{\stone}}{\stfun{\svartwo}{\sAtwo}{\sttwo}}{\sfunty{\sAone}{\sBone}}{\sfunty{\sAtwo}{\sBtwo}}}$.
    Expanding definitions, we need to show
    \[
    \sem\senvone \vDash
    \mtletvert{\mxin{f}}{\mtfun{\mvarone}{\sem\sAone}{\sem\stone}}
    {\mtfun{\mxtwo}{\mAtwo}
      {\theemb{\sBone}{\sBtwo}{\mxin{f}\,(\theprj{\sAone}{\sAtwo}{\mxtwo})}}}
    \ltdyn
    \letemb{\mtfun{\mxtwo}{\mAtwo}{\sem{\mttwo}}}
    \]

    First we simplify by performing some open $\beta$ reductions on
    the left and let-$\lambda$ equivalence and a commuting conversion
    (\cref{lem:open-beta,lem:let-lambda,lem:comm-conv}):
    \begin{align*}
    \mtletvert{\mxin{f}}{\mtfun{\mvarone}{\sem\sAone}{\sem\stone}}
    {\mtfun{\mxtwo}{\mAtwo}
      {\theemb{\sBone}{\sBtwo}{\mxin{f}\,(\theprj{\sAone}{\sAtwo}{\mxtwo})}}}
    &\equidyn {\mtfun{\mxtwo}{\mAtwo}
      {\theemb{\sBone}{\sBtwo}{\mtletvert{\mxone}{\theprj{\sAone}{\sAtwo}{\mxtwo}}{\sem \mtone}}}}\\
    &\equidyn {\mtfun{\mxtwo}{\mAtwo}{\mtletvert{\mxone}{\theprj{\sAone}{\sAtwo}{\mxtwo}}{\theemb{\sBone}{\sBtwo}{\sem \mtone}}}}
    \end{align*}
    and on the right, we move the embedding into the body, which is
    justified because embeddings are essentially values
    (\cref{lem:emb-pure,lem:pure-subst}):
    \[
    \letemb{\mtfun{\mxtwo}{\mAtwo}{\sem{\mttwo}}}
    \equidyn
    \mtfun{\mxtwo}{\mAtwo}{\letemb{\sem{\mttwo}}}
    \]
    The final step is justified by congruence \cref{lem:cong} and adjointness on inputs \cref{lem:adj-inp} and the premise:
    \[
    {\mtfun{\mxtwo}{\mAtwo}{\mtletvert{\mxone}{\theprj{\sAone}{\sAtwo}{\mxtwo}}{\theemb{\sBone}{\sBtwo}{\sem \mtone}}}}
    \equidyn
    \mtfun{\mxtwo}{\mAtwo}{\letemb{\sem{\mttwo}}}
    \]
    
  \item
    $\inferrule{\semtmdynr{\senvone}{\senvtwo}{\stone}{\sttwo}{\sfunty{\sAone}{\sBone}}{\sfunty{\sAtwo}{\sBtwo}}\and
    \semtmdynr{\senvone}{\senvtwo}{\ssone}{\sstwo}{{\sAone}}{{\sAtwo}}}
    {\semtmdynr{\senvone}{\senvtwo}{\stapp{\stone}{\ssone}}{\stapp{\sttwo}{\sstwo}}{\sBone}{\sBtwo}}$.
    Expanding definitions, we need to show
    \[
    \sem\senvone \vDash
    \theemb{\sBone}{\sBtwo}{\sem{\stone}\,\sem{\ssone}}
    \ltdyn
    \letemb{\sem{\sttwo}\,\sem{\sstwo}}
    : \sem{\sBtwo}
    \]
    First, we duplicate the embedding on the right hand side,
    justified by purity of embeddings, to set up a use of congruence
    later:
    \[
    \letemb{\sem{\sttwo}\,\sem{\sstwo}}
    \equidyn
    \left(\letemb{\sem{\sttwo}}\right)\left(\letemb{\sem{\sstwo}}\right)
    \]

    Next, we use linearity of evaluation contexts
    \cref{lem:evctx-linear} so that we can do reductions at the
    application site without worrying about evaluation order:
    \[
    \theemb{\sBone}{\sBtwo}{\sem{\stone}\,\sem{\ssone}}
    \equidyn
    \mtletvert{\mxin{f}}{\sem{\stone}}{
    \mtletvert{\mxin{a}}{\sem{\ssone}}{
    \theemb{\sBone}{\sBtwo}{\mxin{f}\,\mxin{a}}
    }}
    \]

    Next, we ep-expand $\mxin{f}$ (\cref{lem:dyn-der-ep}) and perform
    some $\beta$ reductions, use the ep property and then reverse the
    use of linearity.
    \begin{align*}
    \mtletvert{\mxin{f}}{\sem{\stone}}{
    \mtletvert{\mxin{a}}{\sem{\ssone}}{
    \theemb{\sBone}{\sBtwo}{\mxin{f}\,\mxin{a}}
    }}
    & \equidyn
    \mtletvert{\mxin{f}}{\sem{\stone}}{
    \mtletvert{\mxin{a}}{\sem{\ssone}}{
    \theemb{\sBone}{\sBtwo}
           {\theprj{\sAone\sto\sBone}{\sAtwo\sto\sBtwo}{\theemb{\sAone\sto\sBone}{\sAtwo\sto\sBtwo}{\mxin{f}}}\,\mxin{a}}}}\\
    \by{open $\beta$}{lem:open-beta}& \equidyn
    \mtletvert{\mxin{f}}{\sem{\stone}}{
    \mtletvert{\mxin{a}}{\sem{\ssone}}{
    \theemb{\sBone}{\sBtwo}{\theprj{\sBone}{\sBtwo}{\left({\theemb{\sAone\sto\sBone}{\sAtwo\sto\sBtwo}{\mxin{f}}}\right)\,\left(\theemb{\sAone}{\sAtwo}{\mxin{a}}\right)}}
    }}\\
    \by{ep pair}{lem:dyn-der-ep}& \ltdyn
    \mtletvert{\mxin{f}}{\sem{\stone}}{
    \mtletvert{\mxin{a}}{\sem{\ssone}}{
    \left({\theemb{\sAone\sto\sBone}{\sAtwo\sto\sBtwo}{\mxin{f}}}\right)\,\left(\theemb{\sAone}{\sAtwo}{\mxin{a}}\right)
    }}\\
    \by{linearity}{lem:evctx-linear}& \equidyn
    \left({\theemb{\sAone\sto\sBone}{\sAtwo\sto\sBtwo}{\sem \stone}}\right)\,\left(\theemb{\sAone}{\sAtwo}{\sem \ssone}\right)
    \end{align*}
    With the final step being congruence \cref{lem:cong}:
    \[
    \left({\theemb{\sAone\sto\sBone}{\sAtwo\sto\sBtwo}{\sem \stone}}\right)\,\left(\theemb{\sAone}{\sAtwo}{\sem \ssone}\right)
    \equidyn
    \left(\letemb{\sem{\sttwo}}\right)\left(\letemb{\sem{\sstwo}}\right)
    \]
  \end{enumerate}
\end{proof}
\fi

\section{Related Work and Discussion}
\label{section:related}

Our analysis of graduality as observational approximation and dynamism
as ep pairs builds on the axiomatic and denotational semantics of
graduality for a call-by-name language presented in
\cite{newlicata2018}.
The semantics there gives axioms of type and term dynamism that imply that
upcasts and downcasts are embedding-projection pairs.
Our analysis here is complementary: we present the graduality theorem
as a concrete property of a gradual language defined with an
operational semantics.
Our graduality logical relation should serve as a concrete
\emph{model} of a call-by-value version of gradual type theory,
similar to the call-by-name denotational models presented there.
Furthermore, we show here how this interpretation of graduality maps
back to a standard cast calculus presentation of gradual typing.

\paragraph{Graduality vs Gradual Guarantee}
\label{sec:rel:gradual-guarantee}

The notion of graduality we present here is based on the \emph{dynamic
  gradual guarantee} by Siek, Vitousek, Cimini, and Boyland
\cite{refined, boyland14}. The dynamic gradual guarantee says that
syntactic term dynamism is an \emph{invariant} of the operational
semantics up to error on the less dynamic side.
More precisely, if $\cdot \vdash t_1 \ltdyn t_2 : A_1 \ltdyn A_2$ then
either $t_1 \stepstar \mho$ or both $t_1,t_2$ diverge or $t_1
\stepstar v_1$ and $t_2 \stepstar v_2$ with $v_1 \ltdyn v_2$.
Observe that when restricting $A_1 = A_2 = 1$, this is precisely
the relation on closed programs out of which we build our definition
of semantic term dynamism.
We view their formulation of the dynamic gradual guarantee as a
syntactic \emph{proof technique} for proving graduality of the system.

Graduality should be easier to formulate for different presentations
of gradual typing because it does not require a second syntactic
notion of term dynamism for the implementation language.
In the proofs of the gradual guarantee in \citet{refined}, they have to
develop new rules for term dynamism for their cast calculus, that they
do not attempt to justify at an intuitive level.
Additionally, they have to change their translation from the gradual
surface language to the cast calculus, because the traditional
translation did not preserve the rigid syntactic formulation of term
dynamism.
In more detail, when a dynamically typed term $t : \dyn$ was applied to a
term $s : A$, in their original formulation this was translated as
\[ \sem{t~s} = (\langle(A \to \dyn) \Leftarrow \dyn\rangle\sem{t})\:\sem{s} \]
but if the term in function position had a function type $t' : \dyn \to \dyn$, it was translated as
\[ \sem{t'\,s} = (\sem{t}\:(\langle \dyn \Leftarrow A\rangle\sem{s})\]
But if $t' \ltdyn t$, we would not have $\sem{t' s} \ltdyn \sem{t s}$ because the function position on the left has type $\dyn \to \dyn$ which is \emph{more dynamic} than on the right which has $A \to \dyn$.
While changing this was perfectly reasonable to do to use their
syntactic proof method, we can see that from the \emph{semantic} point
of view of graduality there was nothing wrong with their original
translation and it could have been validated using a logical relation.

Another significant difference between our work and theirs is that we
identify the central role of embedding-projection pairs in graduality,
and take advantage of it in our proof.
As mentioned above, they add rules to term dynamism for the cast
calculus without justification.
These rules are the generalization of our \textsc{Cast-Right} and 
\textsc{Cast-Left} \emph{without} the restriction that the casts be
upcasts or downcasts:

\begin{mathpar}
    \inferrule*[right=Cast-Right']
    {\tmdynr{\senvone}{\senvtwo}{\stone}{\sttwo}{\sAone}{\sAtwo}
      \and \sAone \dynr \sBtwo}
    {\tmdynr{\senvone}{\senvtwo}{\stone}{{\obcast\sAtwo\sBtwo}{\sttwo}}{\sAone}{\sBtwo}}\quad
    \inferrule*[right=Cast-Left']
    {\tmdynr{\senvone}{\senvtwo}{\stone}{\sttwo}{\sAone}{\sAtwo}
      \and \sBone \ltdyn \sAtwo}
    {\tmdynr{\senvone}{\senvtwo}{{\obcast{\sAone}{\sBone}}{\stone}}{\sttwo}{\sBone}{\sAtwo}}
\end{mathpar}

These are valid rules in our system, but by identifying the subset of
upcasts and downcasts, we prove the validity of the rules from
earlier, intuitive rules: decomposition, congruence, and the ep-pair
properties.
Furthermore, while we do not take these rules as primitive it is
notable that these two rules imply that upcasts and downcasts are
adjoint---i.e., if $\sAone \ltdyn \sAtwo$, the following are provable for
$\st : \sAone$ and $\ss : \sAtwo$:
\[
\st \ltdyn \obcast{\sAtwo}{\sAone}{\obcast\sAone\sAtwo{\st}} \qquad
\obcast{\sAone}{\sAtwo}{\obcast\sAtwo\sAone{\ss}} \ltdyn \ss
\]

\citet{refined} also present a theorem called the
\emph{static} gradual guarantee that pertains to the type checking of
gradually typed programs.
The static gradual guarantee says that if a term $\Gamma \vdash t_1 :
A_1$ type checks, and $t_2$ is syntactically more dynamic, then $\Gamma
\vdash t_2 : A_2$ with a more dynamic type, i.e., $A_1 \sqsubseteq A_2$.
We view this as a \emph{corollary} to graduality. If type checking is
a compositional procedure that seeks to rule out dynamic type errors,
then if $t_1$ is syntactically less dynamic than $t_2$, then it is
also semantically less dynamic, meaning every type error in $t_2$'s
behavior was already present in $t_1$, so it should also type check.

\paragraph{Types as EP Pairs}
\label{section:related:retraction}

The interpretation of types as retracts of a single domain originated
in \citet{scott71} and is a common tool in denotational semantics,
especially in the presence of a convenient \emph{universal} domain.
A retraction is a pair of morphisms $s : A \to B$, $r : B \to A$
that satisfy the retraction property $r \circ s = \id_A$, but not
necessarily the projection property $s \circ r \errord
\id_B$. Thus ep pair semantics can be seen as a more refined
retraction semantics.
Retractions have been used to study interaction between typed and
untyped languages, e.g., see \citet{benton05:embedded, favonia17}.

Embedding-projection pairs are used extensively in domain theory as a
technical device for solving non-well-founded domain equations, such
as the semantics of a dynamic type.
In this paper, our error-approximation ep pairs do not play this role,
and instead the retraction and projection properties are desirable in
their own right for their intuitive meaning for type checking.

Many of the properties of our embedding-projection pairs are
anticipated in \citet{henglein94:dynamic-typing} and \citet{thatte90}.
\citet{henglein94:dynamic-typing} defines a language with a notion of \emph{coercion} $A
\rightsquigarrow B$ that corresponds to general casts, with primitives
of tagging $tc! : tc(\dyn,\ldots) \rightsquigarrow \dyn$ and untagging
$tc? : \dyn \rightsquigarrow tc(\dyn,\ldots)$ for every type
constructor ``$tc$''.
Crucially, Henglein notes that $tc!;tc?$ is the identity modulo efficiency
and that $tc?;tc!$ errors more than the identity.
Furthermore, they define classes of ``positive'' and ``negative''
coercions that correspond to embeddings and projections, respectively,
and a ``subtyping'' relation that is the same as type precision.
They then prove several theorems analogous to our results:
\begin{enumerate}
\item (Retraction) For any pair of positive coercion $p : A \rightsquigarrow B$,
  and negative coercion $n : B \rightsquigarrow A$, they show that
  $p;n$ is equal to the identity in their equational theory.
\item (Almost projection) Dually, they show that $n;p$ is equal to
  the identity \emph{assuming} that $tc?;tc!$ is equal to the identity
  for every type constructor.
\item They show every coercion factors as a positive cast to $\dyn$
  followed by a negative cast to $\dyn$.
\item They show that $A \leq B$ if and only if there exists a positive
  coercion $A \rightsquigarrow B$ and a negative coercion $B
  \rightsquigarrow A$.
\end{enumerate}
They also prove factorization results that are similar to our
factorization definition of semantic type precision, but it is unclear
if their theorem is stronger or weaker than ours.
One major difference is that their work is based on an equational
theory of casts, whereas ours is based on notions of observational
equivalence and approximation of a standard call-by-value language.
Furthermore, in defining our notion of observational error approximation,
we provide a more refined projection property, justifying their use of
the term ``safer'' to compare $p;e$ and the identity.

The system presented in \citet{thatte90}, called ``quasi-static typing''
is a precursor to gradual typing that inserts type annotations into
dynamically typed programs to make type checking explicit.
There they prove a desirable soundness theorem that says their type
insertion algorithm produces an explicitly coercing term that is
minimal in that it errors no more than the original dynamic term.
They prove this minimality theorem with respect to a partial order
$\sqsupseteq$ defined as a logical relation over a domain-theoretic
semantics that (for the types they defined) is analogous to our error
ordering for the operational semantics.
However, they do not define our operational formulation of the
ordering as contextual approximation, linked to the denotational
definition by the adequacy result, nor that any casts form embedding-projection
pairs with respect to this ordering.

Finally, we note that neither of these papers
\cite{henglein94:dynamic-typing,thatte90} extends the analysis to 
anything like graduality.

\paragraph{Semantics of Casts}
\label{section:related:casts}

Superficially similar to the embedding-projection pair semantics are
the \emph{threesome casts} of \citet{siek-wadler10}.
A threesome cast factorizes an arbitrary cast $A \Rightarrow B$
through a third type $C$ as a \emph{downcast} $A \Rightarrow C$
followed by an \emph{upcast} $C \Rightarrow B$, whereas ep-pair semantics
factorizes a cast as an \emph{upcast} $A \Rightarrow \dyn$ followed by a
\emph{downcast} $\dyn \Rightarrow B$.
Threesome casts can be used to implement gradual typing in a
space-efficient manner, the third type $C$ is used to collapse a
sequence of arbitrarily many casts into just the two.
In the general case, the threesome cast $A \Rightarrow C \Rightarrow
B$ is \emph{stronger} (fails more) than the direct cast $A \Rightarrow
B$.
This is the point of threesome casts: the middle type faithfully
represents a sequence of casts in minimal space.
EP pair semantics instead factorizes a cast $A \Rightarrow B$ into an
\emph{upcast} $A \Rightarrow \dyn$ followed by a \emph{downcast} $\dyn
\Rightarrow B$, a factorization already utilized in
\cite{henglein94:dynamic-typing}, and which we showed is
\emph{always} equivalent to the direct cast $A \Rightarrow B$.
We view the benefits of the techniques as orthogonal: the up-down
factorization helps to prove graduality, whereas the down-up
factorization helps implementation.
The fact that both techniques reduce reasoning about arbitrary casts
to just upcasts and downcasts supports the idea that upcasts and
downcasts are a fundamental aspect of gradual typing.

Recently, work on \emph{dependent interoperability} \cite{dagand16,
  dagand18} has identified Galois connections as a semantic
formulation for casting between more and less precise types in a
non-gradual dependently typed language, and conjectures that this
should relate to type dynamism.
We confirm their conjecture in showing that the casts in gradual
typing satisfy the slightly stronger property of being
embedding-projection pairs and have used it to explain the cast semantics
of gradual typing and graduality.
Furthermore, our analysis of the precision rules as compositional
constructions on ep pairs is directly analogous to their library,
which implements ``connections'' between, for instance, function types
given connections between the domains and codomains using Coq's
typeclass mechanism.

\paragraph{Pairs of Projections and Blame}
\label{section:related:projections}
One of the main inspirations for this work is the analysis of
contracts in \citet{findler-blume06}.
They decompose contracts in untyped languages as a pair of
``projections'', i.e., functions $c : \dyn \to \dyn$ satisfying $c
\errordof{\dyn\to\dyn} \id$.
However, they do not provide a rigorous definition or means to prove
this ordering for complex programs as we have.
There is a close relationship between such projections and ep pairs
(an instance of the relationship between adjunctions and (co)monads):
for any ep pair $e,p : A \eppair B$, $e\circ p : B \to B$ is a
projection.
However, we think this relationship is a red herring: instead we think
that a pair of projections is better understood as ep pairs
themselves.
The intuition they present is that one of the projections restricts
the behavior of the ``positive'' party (the term) and the other
restricts the behavior of the ``negative'' party (the continuation).
EP pairs are similar, the projection restricts the positive party by
directly checking, and the embedding restricts the negative party in
the function case by calling a projection on any value received from
its continuation.
However, in our current formulation, it does not even make sense to ask
if each component of our embedding-projection pairs is a projection
because the definition of a projection assumes that the domain and
codomain are the same (to define the composite $c \circ c$).
We conjecture that this can be made sensible by using a PER semantics
where types are relations on untyped values, so that the embedding and
projection have ``underlying'' untyped terms representing them, and
those are projections.

Their analysis of blame was adapted to gradual typing in
\citet{wadler-findler09} and plays a complementary role to our
analysis: they use the dynamism relation to help prove the
blame soundness theorem, whereas we use it to prove graduality.
The fact that they use essentially the same solution suggests there is
\mbox{a deeper connection between blame and graduality than is currently
understood.}

\paragraph{Gradualization}
\label{section:related:gradualization}
The Gradualizer \cite{gradualizer16,gradualizer17} and Abstracting
Gradual Typing (AGT) \cite{AGT} both seek to make language design
for gradually typed languages more systematic.
In doing so they make proving graduality far easier than our proof
technique possibly could: it holds by construction.
Furthermore, these systems also provide a surface-level syntax for
gradual typing and an explanation for gradual type checking, while we
do not address these at all.
However, the downside of their approaches is that they require a rigid
adherence to a predefined language framework.
While our gradual cast calculus as presented fits into this framework,
many gradually typed languages do not.
For instance, Typed Racket, the first gradually typed language ever
implemented \cite{tobin-hochstadt08}, is not given an operational
semantics in the style of a cast calculus, but rather is given a
semantics \emph{by translation} to an untyped language using
contracts.
We could prove the graduality of such a system by adapting our logical
relation to an untyped setting.

We hope in the future to explore the connections between the above
frameworks and our analysis of dynamism as embedding-projection pairs.
We conjecture that both Gradualizer and AGT by construction produce
upcasts and downcasts that satisfy the ep pair properties.
The AGT approach in particular has some similarities that stand out:
their formulation of type dynamism is based on an embedding-projection
pair between static types and sets of gradual types.
However, we are not sure if this is a coincidence or has a deeper
connection to our approach.

\section{Conclusion}
\label{sec:concl}
Graduality is a key property for gradually typed languages as it
validates programmer intuition that adding precise types only results
in stricter type checking.
Graduality is challenging to prove.
Moreover, it rests upon the language's
definition of type dynamism but there has been little guidance on
defining type dynamism, other than that graduality must hold.
We have given a semantics for type dynamism: $\sA \ltdyn \sB$ should
hold when the casts between $\sA, \sB$ form an embedding-projection
pair.
This allows for natural proofs of graduality using a logical relation
for observational error approximation.

Looking to the future, we would like to make use of our semantic
formulation of type dynamism based on ep pairs to design and analyze
gradual languages with advanced features such as parametric
polymorphism, effect tracking, and mutable state.
For parametric polymorphism in particular, we would like to
investigate whether our approach justifies any of the type-dynamism
definitions previously proposed \cite{ahmed17,igarashipoly17}, and the
possibility of proving both graduality and parametricity theorems with
a single logical relation.

\begin{acks}                            
We gratefully acknowledge the valuable feedback provided by Ben
Greenman and the anonymous reviewers.  Part of this work was done at
Inria Paris while Amal Ahmed was a Visiting Professor.

This material is based upon work supported by the National Science Foundation 
under grant CCF-1453796, and the European Research Council under ERC Starting
Grant SECOMP (715753). Any opinions, findings, and conclusions or
recommendations expressed in this material are those of the authors and do not 
necessarily reflect the views of our funding agencies. 

\end{acks}

\bibliography{max}


\end{document}